\newif\ifsubmission
\submissionfalse

\newif\ifnotes
\notesfalse

\newif\ifllncs
\llncsfalse

\newif\ifblind
\blindfalse

\documentclass[11pt]{article}
\usepackage{graphicx} 
\usepackage{fullpage}

\usepackage{amsmath}
\usepackage{amssymb}
\usepackage{xcolor}
\usepackage{tikz}
\usetikzlibrary{matrix,arrows.meta,positioning,calc,decorations.pathreplacing,backgrounds}

\usepackage{mathdots}

\usepackage{microtype}      
\usepackage{libertine}
\usepackage{libertinust1math}
\usepackage{bbm}
\usepackage{xcolor}
\usepackage{multirow,multicol}

\usepackage{caption}
\usepackage{amsthm}
\usepackage{thmtools, thm-restate}
\usepackage{hyperref}
\usepackage{rotating}
\usepackage{cleveref}

\newcommand{\Z}{\mathbb{Z}}
\newcommand{\R}{\mathbb{R}}

\newcommand{\F}{\mathbb{F}}
\newcommand{\C}{\mathbb{C}}
\newcommand{\poly}{\,\mathrm{poly}}
\newcommand{\polylog}{\,\mathrm{polylog}}
\newcommand{\ord}{\mathrm{ord}}

\renewcommand{\epsilon}{\varepsilon}
\renewcommand{\phi}{\varphi}
\def\uv{\mathbf{u}}
\def\vv{\mathbf{v}}
\def\bv{\mathbf{b}}
\def\DB{\mathsf{DB}}
\def\zo{\{0, 1\}}
\def\sh{\mathsf{sh}}
\def\Coeff{\mathsf{Coeff}}
\DeclareMathOperator*{\E}{\mathbb{E}}
\DeclareMathOperator{\Aut}{Aut}

\allowdisplaybreaks

\ifnotes
\newcommand{\authnote}[2]{[{\color{magenta}\textbf{#1:}}~{\color{blue} #2}]}
\newcommand{\vnote}[1]{\authnote{Vinod}{#1}}
\newcommand{\snote}[1]{\authnote{Seyoon}{#1}}
\newcommand{\anote}[1]{\authnote{Aparna}{#1}}
\else
\newcommand{\vnote}[1]{}
\newcommand{\snote}[1]{}
\newcommand{\anote}[1]{}
\fi

\newtheorem{theorem}{Theorem}
\numberwithin{theorem}{section}

\newtheorem{conjecture}{Conjecture}
\numberwithin{conjecture}{section}

\newtheorem{lemma}[theorem]{Lemma}
\numberwithin{lemma}{section}

\newtheorem{corollary}[theorem]{Corollary}
\numberwithin{corollary}{section}

\newtheorem{definition}[theorem]{Definition}
\numberwithin{definition}{section}

\numberwithin{proposition}{section}

\newtheorem*{remark}{Remark}

\numberwithin{problem}{section}

\newtheorem{claim}{Claim}
\numberwithin{claim}{section}

\newtheorem{fact}[theorem]{Fact}
\numberwithin{fact}{section}

\newtheorem{heuristic}[conjecture]{Heuristic}
\numberwithin{heuristic}{section}

\hypersetup{colorlinks=true,linkcolor=blue,urlcolor=blue,citecolor=blue}

\title{\Large Exponentially Fewer-Server PIR from Sparser $S$-Decoding Polynomials}

\ifblind
    \author{}
    \date{}
\else
    \author{Aparna Gupte\\MIT\\\url{agupte@mit.edu} \and Seyoon Ragavan\\MIT\\\url{sragavan@mit.edu}}
\fi

\begin{document}
\maketitle
\thispagestyle{empty}
\begin{abstract}
    We show that under a plausible number-theoretic conjecture, for any constant $s$ there exists an $s$-server private information retrieval (PIR) protocol that on an $n$-bit database requires communication $\exp(O((\log n)^{1/s} (\log \log n)^{1-1/s}))$. This achieves a superpolynomial improvement in communication when $s \geq 4$; previous constructions attaining the same communication required $2^{O(s)}$ servers. Our number-theoretic conjecture is implied by existing conjectures, namely the generalized repunit conjecture and Schinzel's hypothesis H (either one of these conjectures would suffice alone).
    
    Our result builds on the ``matching vector family + $S$-decoding polynomials'' framework pioneered by Efremenko (STOC 2009) and recently refined by Ghasemi, Kopparty, and Sudan (STOC 2025). The main ingredient is a framework for constructing $S$-decoding polynomials with only $k+1$ nonzero coefficients modulo special products of $k$ primes, resolving an open problem posed by Ghasemi and Kopparty (ITCS 2026). By the lower bound shown by Ghasemi and Kopparty, this is the minimum achievable sparsity. We also empirically validate our construction and make our result unconditional for all $s \leq 15$.
    
    We also apply our techniques to regimes where $s$ grows with $n$, showing under a stronger variant of our number-theoretic conjecture that the communication complexity of $s$-server matching-vector PIR can be superpolynomially reduced from the previous state of the art for any $s \leq \exp(o(\sqrt{\log \log n/\log \log \log n}))$.
    
    The main result for $s = O(1)$ and its proof were discovered in a GPT-5.5 Pro conversation prompted by the authors.
\end{abstract}

\ifsubmission
\newpage
\clearpage            
\pagenumbering{arabic}
\else
\tableofcontents
\fi


\section{Introduction}

Private information retrieval (PIR), first introduced by Chor, Kushilevitz, Goldreich, and Sudan~\cite{CGKS98}, is a protocol comprising a user who holds an index $i \in [n]$ and $s$ non-colluding servers who hold a database $\DB \in \zo^n$.
The user sends a question to each server, then each server replies with an answer.
Correctness requires that after receiving the servers' answers, the user can recover the $i$th entry $\DB_i \in \zo$.
The privacy requirement is that each server should individually learn nothing about the index $i$; formally, the marginal distribution of any one server's question should be completely independent of the index $i$.
Subject to these constraints, a key motivating question in the field is to find protocols with as little total communication as possible.
The trivial baseline communication is $n$, where one of the servers sends the entire database to the user.

In the forthcoming discussion, \underline{we will regard $s$ as constant} and $n$ as growing; our work also addresses the regime where $s$ grows with $n$, but the picture there is considerably more complicated for reasons we will discuss in Section~\ref{sec:superconstantoverview}; a summary of our results for growing $s$ can be found in Figure~\ref{fig:phasediagram}. 
With that aside: a straightforward PIR construction based on Reed-Muller codes requires communication $O(n^{1/(s-1)})$.
Using these ideas together with recursive composition led to protocols requiring communication $O(n^{1/(2s-1)})$~\cite{A97} and later $n^{O(\log \log s/(s \log s))}$~\cite{BIKR02}.
Even with these improvements, the qualitative narrative remained that PIR with $O(1)$ servers seemed to require polynomial i.e., $n^{\Omega(1)}$ communication.

The story completely changed with Yekhanin's breakthrough work~\cite{Y08,DBLP:journals/eccc/Raghavendra07} which showed---assuming the infinitude of Mersenne primes---a 3-server PIR protocol achieving communication $n^{O(1/\log \log n)}$.
Soon after, Efremenko~\cite{Efr09} showed unconditionally that with 3 servers, one can attain significantly smaller communication $\exp(\tilde{O}(\sqrt{\log n}))$.
These works laid out a powerful two-ingredient template for constructing surprisingly low-communication PIR:
\begin{itemize}
    \item Let $S \subseteq \Z_m$ with $0 \in S$.
    An \underline{$S$-matching vector family} is a collection of vectors $\uv_1, \ldots, \uv_n, \vv_1, \ldots, \vv_n \in \Z_m^d$ such that $\langle \uv_i, \vv_j \rangle \bmod{m} \in S$ always, and moreover $\langle \uv_i, \vv_j \rangle \equiv 0 \pmod{m} \Leftrightarrow i = j$.
    A lower dimension $d$ ultimately leads to a lower-communication PIR.

    For $m$ a constant prime, a simple rank argument shows that this forces $d \geq n^{\Omega(1)}$.
    The surprise is that for $m$ a product of $k > 1$ distinct primes, one can achieve this property with $d = \exp(\tilde{O}((\log n)^{1/k}))$.
    Moreover, the corresponding set $S := S^*_m$ (defined in Definition~\ref{def:canonicalset}) has size $2^k$.
    This construction of matching vector families is due to Grolmusz~\cite{DBLP:journals/combinatorica/Grolmusz00} and Barrington, Beigel, and Rudich~\cite{Barrington1994}, and has remained the state of the art since then.

    \item Let $\F$ be a finite field and $\gamma \in \F$ a primitive $m$th root of unity. An \underline{$S$-decoding polynomial} is a polynomial $P: \F \to \F$ such that $P(1) = 1$ and $P(\gamma^j) = 0$ for all $j \in S \backslash \{0\}$.
    A sparser $S$-decoding polynomial (i.e., one with fewer nonzero coefficients) leads to a fewer-server PIR.

    One can always achieve a baseline sparsity of $|S^*_m|$ by interpolating a polynomial of degree $|S|-1$ with the desired properties.
    Combining this with the matching vector family for $m$ a product of two distinct primes would give a 4-server PIR with communication $\exp(\tilde{O}(\sqrt{\log n}))$.
    
    The observation by Efremenko is that one can sometimes do better than this baseline: for $m = 7 \times 73 = 511$ and $\F = \F_{512}$, there exists an $S^*_m$-decoding polynomial of sparsity 3.
    Later, this phenomenon was generalized by \cite{Chee2011}, who showed assuming a number-theoretic conjecture that there exist choices of $m$ that are products of an even number $k$ of distinct primes and admit $S^*_m$-decoding polynomials of sparsity $3^{k/2}$.
\end{itemize}
The history of matching-vector PIR took another decisive turn with the work of Dvir and Gopi~\cite{DG16}, which showed that even two servers suffice to achieve communication $2^{\tilde{O}(\sqrt{\log n})}$.
Their construction used matching vector families together with ideas based on derivatives akin to~\cite{WY05}, but did not leverage $S$-decoding polynomials.
Recently, a beautiful work of Ghasemi, Kopparty, and Sudan~\cite{GKS25} (with a later alternative presentation by~\cite{DBLP:conf/tcc/AlonBL25}) showed that the derivative techniques of~\cite{DG16,WY05} could be combined with the sparse $S$-decoding polynomial approach of~\cite{Efr09,Chee2011} to give the following unified picture of all these matching vector constructions:
\begin{theorem}[\cite{GKS25}, see Theorem~\ref{thm:gks} for detailed statement]\label{thm:gksinformal}
    Assume $k = O(1)$.
    Let $m$ be a product of $k$ distinct primes and let $\F$ be a finite field with a primitive $m$th root of unity.
    Assume there exists an $S^*_m$-decoding polynomial over $\F$ of sparsity $\leq s$.
    Then there exists an $s$-server PIR scheme with communication $\exp(O((\log n)^{1/(k+1)} (\log \log n)^{1-1/(k+1)}))$.
\end{theorem}
In particular, the above theorem was the first to show that there is a 3-server PIR protocol with communication $\exp(\tilde{O}((\log n)^{1/3}))$.
Recall that the above theorem always admits a trivial instantiation with $s = 2^k$.
On the other hand, plugging in the $S$-decoding polynomials found by~\cite{Chee2011} implies that---assuming the same number-theoretic conjecture---the above theorem can be instantiated with $s = 3^{k/2}$ if $k$ is even and $s = 2 \cdot 3^{(k-1)/2}$ otherwise. 
This server dependence on $k$ is polynomially better than trivial.

A natural approach to getting even better communication-server tradeoffs for PIR is therefore to find sparser $S^*_m$-decoding polynomials.
Recently, Ghasemi and Kopparty~\cite{DBLP:conf/innovations/GhasemiK26} showed that such a polynomial must necessarily have sparsity $\geq k+1$, leaving an exponential gap between the best upper and lower bounds.
Ghasemi and Kopparty also explicitly raised finding constructions matching this lower bound as an open question.
A similar open question was raised by Beimel and Lasri~\cite{bl26} in the language of secret share conversion (which we discuss more in Section~\ref{sec:relatedwork}).

\subsection{Our Contributions}
The main result of this work is to resolve both formulations of this open question~\cite{DBLP:conf/innovations/GhasemiK26,bl26}; we show under a number-theoretic conjecture that $S^*_m$-decoding polynomials of sparsity $k+1$ exist for suitable primes $p_1, \ldots, p_k$ and $m = p_1\ldots p_k$.
In doing so, we completely close the previously exponential gap between the upper and lower bounds on the sparsity of an $S^*_m$-decoding polynomial.


\begin{theorem}[Sparse decoding polynomials from Conjecture~\ref{conjecture:onetuple-intro}]\label{thm:intromain}
    Assuming Conjecture~\ref{conjecture:onetuple-intro}, for any constant $k$ there exists $m$ that is a product of $k$ distinct primes, a finite field $\F$ with $m \mid |\F| - 1$, and an $S^*_m$-decoding polynomial of sparsity $k+1$.
\end{theorem}

\begin{corollary}[Secret share conversion]
    Under the same conjecture as in Theorem~\ref{thm:intromain}, for any constant $k$ there exists $m$ that is the product of $k$ distinct primes, such that there is a $(k+1)$-party $S^*_m$-share conversion scheme.
\end{corollary}

Plugging Theorem~\ref{thm:intromain} into Theorem~\ref{thm:gksinformal} implies:
\begin{corollary}[PIR]\label{cor:asymptoticpir}
    Assume Conjecture~\ref{conjecture:onetuple-intro}. Then for any constant $s \geq 2$, there exists an $s$-server PIR scheme with communication $\exp(O((\log n)^{1/s} (\log \log n)^{1-1/s}))$.
\end{corollary}
In Section~\ref{sec:empirics}, we empirically validate our construction for small values of $s$, yielding the following corollary:
\begin{corollary}[PIR, unconditional]\label{cor:concrete}
    For any constant $s \in [2, 15]$, there unconditionally exists an $s$-server PIR with communication $\exp(O((\log n)^{1/s} (\log \log n)^{1-1/s}))$.

    Moreover, assuming the primality of 13 specific integers (listed out in Section~\ref{sec:empirics}) that pass probable prime tests, the above is true for $s \in [2, 24]$.
\end{corollary}
Finally, we note that our PIR results also admit the standard alternate interpretation in terms of locally-decodable codes~\cite{DBLP:conf/stoc/KatzT00,trevisan04}:

\begin{corollary}[Locally-decodable codes]
    Under the same Conjecture~\ref{conjecture:onetuple-intro}, for any constant $s$, there exists a smooth $s$-query locally-decodable code mapping $n$ bits to codewords of length $\exp(\exp(O((\log n)^{1/s} (\log \log n)^{1-1/s})))$ over an alphabet of size $\exp(\exp({O}((\log n)^{1/s}(\log \log n)^{1-1/s})))$.
\end{corollary}

The above statements are purely asymptotic; as the statement of Corollary~\ref{cor:concrete} suggests, our explicit constructions therein make use of probable primes $p_1, p_2, \ldots, p_k$ that have on the order of $10^6$ bits and are difficult to deterministically test for primality.
We tabulate the number of servers required by our construction in comparison with previous work for various communication regimes in Table~\ref{tbl:commregimes}.\footnote{A careful reader might notice that our reported numbers for~\cite{Efr09,Chee2011} in even rows and~\cite{GKS25} in odd rows of Table~\ref{tbl:commregimes} are better than those reported in those papers. For example, the number of servers claimed by~\cite{Chee2011} for even $r$ is $8 \cdot 3^{(r-3)/2}$. The reason for this is that these papers made use of Theorem~\ref{thm:tensoring} due to~\cite{itoh2010improved} in the specific case where none of the $m_i$'s are prime, as this was the variant originally stated by~\cite{itoh2010improved}. Looking more closely at the proof, this requirement is not necessary and setting one of the $m_i$'s to be prime readily recovers the results that we claim from the constructions of~\cite{Efr09,Chee2011,GKS25} with essentially no modification. For completeness, we reprove Theorem~\ref{thm:tensoring} in this generality in Appendix~\ref{sec:tensoringproof}. We make no claim to originality in this observation; for example, the $r = 3$ case is briefly mentioned by~\cite{bl26}.}

\begin{table}
\begin{center}
    \large
    \begin{tabular}{|c|c|c|c|c|}
    \hline
    \multirow{2}{*}{\textbf{Communication}} & \multicolumn{4}{|c|}{\textbf{\# of servers}} \\
    \cline{2-5}
    & \textbf{\cite{Efr09,Chee2011}} & \textbf{\cite{DG16}} & \textbf{\cite{GKS25}} & \textbf{This work} \\
    \hline
    \hline
    $\exp(\tilde{O}((\log n)^{1/2})$ & 3 & 2 & 2 & 2 \\
    \hline
    $\exp(\tilde{O}((\log n)^{1/3})$ & 6 & 4 & 3 & 3 \\
    \hline
    $\exp(\tilde{O}((\log n)^{1/4})$ & 9 & 8 & 6 & 4 \\
    \hline
    $\exp(\tilde{O}((\log n)^{1/5})$ & 18 & 16 & 9 & 5 \\
    \hline
    $\exp(\tilde{O}((\log n)^{1/6})$ & 27 & 32 & 18 & 6 \\
    \hline
    $\exp(\tilde{O}((\log n)^{1/7})$ & 54 & 64 & 27 & 7 \\
    \hline
    $\exp(\tilde{O}((\log n)^{1/8})$ & 81 & 128 & 54 & 8 \\
    \hline
    $\exp(\tilde{O}((\log n)^{1/9})$ & 162 & 256 & 81 & 9 \\
    \hline
    \end{tabular}
    \end{center}
    \caption{For $r \in [2, 9]$, we tabulate the number of servers required for PIR with communication $\exp(\tilde{O}((\log n)^{1/r}))$. We compare our constructions with the previous works of~\cite{Efr09,Chee2011,DG16,GKS25}. For all $r$ tabulated, our work requires $r$ servers and~\cite{DG16} requires $2^{r-1}$ servers. For even $r$,~\cite{Efr09,Chee2011} requires $3^{r/2}$ servers while~\cite{GKS25} requires $2 \cdot 3^{(r-2)/2}$ servers. For odd $r$,~\cite{Efr09,Chee2011} requires $2 \cdot 3^{(r-1)/2}$ servers while~\cite{GKS25} requires $3^{(r-1)/2}$ servers.}
    \label{tbl:commregimes}
\end{table}
We next turn to a brief discussion about the number-theoretic conjecture our results rely on.

\subsection{The Number-Theoretic Conjecture}

Our theorems stated thus far hinge on the following conjecture:
\begin{conjecture}\label{conjecture:onetuple-intro}
    For every positive integer $k > 1$, there exists a prime $q \geq k+1$ for which there exist $k$ distinct positive integers $d_1, \ldots, d_k \geq 6$ such that for all $i \in [k]$, $q^{d_i}-1$ has a prime divisor $\geq (q^{d_i}-1)^{1-1/(3k)}$.
\end{conjecture}
We note that it would suffice to ensure that (1) $d_i \geq 3k$; and (2) $(q^{d_i}-1)/(q-1)$ is prime. Integers of the form $(q^d-1)/(q-1)$ are sometimes referred to as \emph{generalized (or base-$q$) repunits}, because their base $q$ expansion consists of $d$ 1's. Base-$q$ repunit primes are a generalization of Mersenne primes, and their distribution is the subject of the generalized repunit conjecture~\cite{Bourdelais2009, dubner1993generalized}, a generalization of the Mersenne prime conjecture. In particular, the generalized repunit conjecture states that there exist infinitely many base-$q$ repunits for any prime $q$. The conjecture we need is only weaker; we only need $(q^{d_i}-1)/(q-1)$ to have a large prime divisor.


Finally, for reasons we will discuss in Section~\ref{sec:superconstantoverview}, Conjecture~\ref{conjecture:onetuple-intro} will be insufficient for us in the regime where the number of servers is superconstant.
What we roughly need is a variant of Conjecture~\ref{conjecture:onetuple-intro} that also specifies how large $q$ and the $d_i$'s need to be (as a function of $k$).
In particular, we would like to make use of a tighter heuristic analysis than just restricting attention to cases where $(q^{d_i}-1)/(q-1)$ is prime.
We state this more complicated conjecture at the beginning of Section~\ref{sec:hardconjecture} and provide a heuristic justification there.

\subsection{Organization of the Paper}

This paper is organized as follows. In Section~\ref{sec:techoverview}, we provide an overview of our techniques. In Section~\ref{sec:prelims}, we set up some notation and preliminaries. In Section~\ref{sec:decpolysmake}, we prove the main technical results underlying Theorem~\ref{thm:intromain}. In Section~\ref{sec:hardconjecture}, we provide additional justification for Conjecture~\ref{conjecture:onetuple-intro} and also state and justify the stronger conjecture that we need for our results in the superconstant \# of servers regime. In Section~\ref{sec:gettopir}, we put these ingredients together to prove our results in the superconstant \# of servers regime. In Section~\ref{sec:empirics}, we report empirical findings that allow us to make Theorem~\ref{thm:intromain} unconditional for $k \leq 14$ and therefore our PIR theorem unconditional for $s \leq 15$ servers. In Section~\ref{sec:concl}, we conclude with a survey of related work and some open questions.


\subsection{Statement on AI Use}

The main result (Theorem~\ref{thm:intromain}) and its proof (Lemmas~\ref{lem:rugisgood} and~\ref{lemma:suffcondsharp}) were discovered by OpenAI's GPT-5.5 Pro model in conversation with the authors.
The authors first asked the LLM to conduct an empirical search for $S^*_m$-decoding polynomials achieving better-than-state-of-the-art sparsity for $m$ being a product of $k$ primes for small values of $k$ (sparsity $\le 5$ for $m$ a product of three primes, and sparsity $\le 8$ for $m$ a product of four primes). After a few rounds of prompting, the model found a general framework to find many such polynomials. The model also stumbled upon a way of building decoding polynomials of sparsity $\leq 5$ for $m$ a product of four primes.
On further prompting, the model generalized this framework to Lemma~\ref{lem:rugisgood}, its instantiation via Lemma~\ref{lemma:suffcondsharp}, and therefore a proof of Theorem~\ref{thm:intromain}.

The authors then identified a number of next steps to produce a more complete picture of how far these techniques can be pushed to obtain better PIR. These steps were executed with the aid of GPT-5.5 Pro, GPT-5.6 Sol Pro, Codex 5.5, and Claude Opus 4.8.
These were (a) Sections~\ref{sec:hardconjecture} and~\ref{sec:gettopir} (adapting these results to the case of super-constant $s$); (b) Lemma~\ref{lemma:gkrigidity} (a converse of Lemma~\ref{lem:rugisgood}); and (c) the empirical validations of our number-theoretic conjectures detailed in Section~\ref{sec:empirics}. 

All text in this paper was written by the authors, sometimes with the aid of draft proofs generated by GPT-5.5 Pro and Claude Opus 4.8.
The authors manually verified all proofs.
Claude Opus 4.8 and 5.5 were used to generate Figures~\ref{fig:rug} and~\ref{fig:phasediagram}, and for other editorial tasks.

\subsection{Acknowledgements}

The authors thank Vinod Vaikuntanathan for helpful suggestions and discussions throughout this project.
SR also thanks Alexandra Henzinger and Edward Pyne for useful discussions about matching-vector PIR.
The authors' access to ChatGPT Pro was supported by the UK AISI alignment project, and SR's access to Claude Max was supported by Jane Street.

\section{Technical Overview}\label{sec:techoverview}

In Sections~\ref{sec:dptorug} and~\ref{sec:rughowto}, we introduce our main technical machinery for building optimally-sparse decoding polynomials.
In Section~\ref{sec:superconstantoverview}, we discuss how our results carry over to the superconstant-server regime.
Finally, in Section~\ref{sec:superconstanttechoverviewnt}, we briefly discuss our approach to justifying the stronger number-theoretic conjecture required for the superconstant-server setting.

We briefly introduce some relevant terminology. Let $m = p_1 \ldots p_k$ be a product of $k$ distinct primes.
Then the subset $S^*_m \subseteq \Z_m$ consists of the $2^k$ elements that are congruent to either 0 or 1 modulo each of the $p_i$'s.
To handle this set, it will be convenient to define the CRT basis $c_1, \ldots, c_k \in \Z_m$, where $c_i \equiv 1 \pmod{p_i}$ and $0 \pmod{p_j}$ for all $j \neq i$.

Let $\F$ be a finite field with a primitive $m$th root of unity $\gamma$.
For a prime $p \mid |\F|-1$, we will let $\mu_p(\F)$ denote the multiplicative subgroup of $p$th roots of unity.
Bearing this in mind, an $S^*_m$-decoding polynomial is a polynomial $P: \F \to \F$ such that $P(1) = 1$ and $P(\gamma^j) = 0$ for all $j \in S^*_m \backslash \{0\}$.
The primary goal of this section is to show our construction (assuming a plausible number-theoretic conjecture) of $m, \F, P$ such that $P$ is an $S^*_m$-decoding polynomial with only $k+1$ nonzero coefficients.

Before we proceed, we briefly remark that standard arguments (see Definition~\ref{def:decpoly} and Lemma~\ref{lemma:rootnomatter}) imply that (a) one need only require $P(1) \neq 0$; and (b) the existence of a decoding polynomial depends on $\F$ but not on the choice of $\gamma$ within $\F$.

\subsection{Decoding Polynomials and Root-of-Unity Grids}\label{sec:dptorug}

\definecolor{rowA}{RGB}{192,57,43}   
\definecolor{rowB}{RGB}{41,98,180}   
\definecolor{rowC}{RGB}{39,124,63}   

\newcommand{\nz}{\boldsymbol{\ast}}  

\begin{figure}
\centering
\begin{tikzpicture}[>=Stealth, font=\large,
  mat/.style={matrix of math nodes, left delimiter={[}, right delimiter={]},
              row sep=6pt, column sep=10pt, inner sep=1pt}]

\matrix (M1) [mat] {
 \nz & \nz & \nz \\
 \nz & \nz & \nz \\
 \nz & \nz & \nz \\
};

\node[right=14pt of M1] (plus) {$+$};

\matrix (M2) [mat, right=14pt of plus] {
 1 & 1 & 1 \\
 1 & 1 & 1 \\
 1 & 1 & 1 \\
};

\node[right=14pt of M2] (eq) {$=$};

\matrix (M3) [mat, right=14pt of eq,
   row 1/.style={nodes={text=rowA}},
   row 2/.style={nodes={text=rowB}},
   row 3/.style={nodes={text=rowC}}] {
 \bullet & \bullet & \bullet \\
 \bullet & \bullet & \bullet \\
 \bullet & \bullet & \bullet \\
};

\coordinate (capline) at ($(M2.south)+(0,-10pt)$);
\node[anchor=north, align=center, font=\small] at (M1|-capline) {Rank $1$\\(nonzero entries)};
\node[anchor=north, font=\small]               at (M2|-capline) {All ones};
\node[anchor=north, font=\small]               at (M3|-capline) {Root-of-unity grid};

\node[text=rowA, anchor=west] (LA) at ($(M3-1-3.east)+(1.35,0)$) {$p_1$-th roots of unity};
\node[text=rowB, anchor=west] (LB) at ($(M3-2-3.east)+(1.35,0)$) {$p_2$-th roots of unity};
\node[text=rowC, anchor=west] (LC) at ($(M3-3-3.east)+(1.35,0)$) {$p_3$-th roots of unity};
\draw[->,rowA,thick] ($(M3-1-3.east)+(0.45,0)$) -- (LA.west);
\draw[->,rowB,thick] ($(M3-2-3.east)+(0.45,0)$) -- (LB.west);
\draw[->,rowC,thick] ($(M3-3-3.east)+(0.45,0)$) -- (LC.west);

\draw[decorate, decoration={brace, amplitude=6pt}]
     ([xshift=8pt]LA.north east) -- ([xshift=8pt]LC.south east);
\node[anchor=west] at ([xshift=17pt]$(LA.east)!0.5!(LC.east)$) {in $\mathbb{F}$};

\end{tikzpicture}
\caption{\emph{Figure depicting a root-of-unity grid for a field $\F$ and modulus $m = p_1p_2p_3$ with $m \mid |\F|-1$. The entries with asterisks and bullets can be arbitrary subject to the stated conditions. The existence of a root-of-unity grid is necessary and sufficient for the existence of an $S^*_m$-decoding polynomial over $\F$ with optimal sparsity $k+1$. A root-of-unity grid is essentially a $k \times k$ matrix which can be written as a sum of a rank 1 matrix with nonzero entries and the all 1's matrix, such that the $i$th row consists only of $p_i$th roots of unity over $\F$.}}
\label{fig:rug}
\end{figure}

\paragraph{Starting point: the Ghasemi-Kopparty lower bound.}
Our starting point is the proof by Ghasemi and Kopparty~\cite[second proof of Theorem 3.6]{DBLP:conf/innovations/GhasemiK26} that an $S^*_m$-decoding polynomial \emph{cannot} have fewer than $k+1$ nonzero coefficients.
Let us write the purported decoding polynomial as $P(X) = \sum_{j = 1}^\ell \alpha_j X^{\beta_j}$ with $\ell < k+1$ nonzero coefficients.
Then the idea behind the proof by~\cite{DBLP:conf/innovations/GhasemiK26} is a simple computation (which we reproduce in Lemma~\ref{lemma:gkrigidity}) that yields elements $z_{i, j} \in \mu_{p_i}(\F)$ for $i \in [k]$ and $j \in [\ell]$ such that:
\begin{align*}
    \sum_{j = 1}^\ell \alpha_j \prod_{i = 1}^k (1+X_iz_{i, j}^{-1}) &= 1 \quad\forall X_1, \ldots, X_k \in \F.
\end{align*}
If $\ell \leq k$, we can force the LHS to vanish with a ``diagonalization'' idea by setting $X_i = -z_{i, i}$ for all $i$, giving the desired contradiction.

Now what if $\ell = k+1$? We can no longer set the $X_i$'s to eliminate all $k+1$ terms being summed on the LHS, however we can set them to eliminate any $k$ of them.
This gives some equations in the $\alpha_j$'s and $z_{i, j}$'s that can be combined (see Lemma~\ref{lemma:gkrigidity}) to arrive at the following necessary condition for a $(k+1)$-sparse decoding polynomial to exist:

\begin{quote}
    \begin{center}
        There exist distinct $a_1, \ldots, a_k \in \F \backslash \{0\}$ and distinct $t_1, \ldots, t_{k+1} \in \F$ such that $1+a_it_j \in \mu_{p_i}(\F) \quad\forall i \in [k], j \in [k+1]$.
    \end{center}
\end{quote}
We refer to this configuration as a \emph{root-of-unity grid}.
Note that we are free to set one of the $t_j$'s to 0 to automatically satisfy all the conditions for that $j$.
We visualize the definition of a root-of-unity grid for $k = 3$ in Figure~\ref{fig:rug}.

\paragraph{Sufficiency of root-of-unity grids.}
Perhaps surprisingly, we can also show that the converse is true! A root-of-unity grid is a necessary and sufficient condition for an $S^*_m$-decoding polynomial of sparsity $k+1$ to exist.
Again, we skip over the relevant computations and state the high-level workings of the proof, which we present in Lemma~\ref{lem:rugisgood}.
For bits $u_1, \ldots u_k \in \zo$, consider the univariate polynomial:
\begin{equation}\label{eq:univariate}
    f_{u_1, \ldots, u_k}(X) := \prod_{i = 1}^k (1+a_iX)^{1-u_i},
\end{equation}
and let $s \in S_m^*$ be written as $\sum_{i \in [k]} c_i u_i$.
Then we can explicitly write out a polynomial $P$ of sparsity $k+1$ such that:
\begin{equation}\label{eq:univariateinterp}
    P(\gamma^s) = \Coeff_{X^k} [f_{u_1, \ldots, u_k}].
\end{equation}
The way this polynomial is constructed is by unpacking how Lagrange interpolation would recover $\Coeff_{X^k} [f_{u_1, \ldots, u_k}]$ from the evaluations of $f_{u_1, \ldots, u_k}$ at the points $t_1, \ldots, t_{k+1}$.
(We will see a simpler computation with some similar features---but only justifying the trivial claim that there exists a decoding polynomial of sparsity $2^k$---in the next part of this overview.
At this point, we are done; by construction, $f_{u_1, \ldots, u_k}$ has degree $k$ if and only if all $u_i$'s are 0, which is in turn equivalent to $s = 0$.

\paragraph{New perspectives on known decoding polynomials.}
Our construction also gives new perspectives on existing constructions of decoding polynomials that we already had.
For example, there is a trivial $2^k$-term decoding polynomial for every $m, \F$ of the form $\prod_{i = 1}^k (X^{c_i} - \gamma^{c_i})$.
In the above terms, one can lay out the this decoding polynomial as follows:
\begin{itemize}
    \item Let $c_1, \ldots, c_k$ in $\Z_m$ be the CRT basis i.e., $c_i \equiv 1 \pmod{p_i}$ and $\equiv 0 \pmod{p_j}$ for all $j \neq i$.
    \item Then for any bits $u_1, \ldots, u_k \in \zo$, we now consider the multilinear and multivariate polynomial (compare with the univariate, degree $\le k$ polynomial in Equation~\eqref{eq:univariate}):
    \begin{equation}\label{eq:multivariate}
        f_{u_1, \ldots, u_k}(X_1, \ldots, X_k) = \prod_{i = 1}^k (1+(\gamma^{c_i} - 1)X_i)^{1-u_i},
    \end{equation}
    \item Indeed, in this case if we let $j = \sum_{i = 1}^k c_iu_i$, the interpolation identity stated in Lemma~\ref{lem:multilinearinterp} can be shown to yield $\Coeff_{X_1X_2\ldots X_k} [f_{u_1, \ldots, u_k}] = P(\gamma^j)$, 
    where $P(X) = (-1)^k X^{- \sum_{i \in [k]}c_i} \prod_{i = 1}^k (X^{c_i} - \gamma^{c_i}) = (-1)^k \prod_{i = 1}^k \left(1 - \gamma^{c_i} X^{-c_i}\right)  = (-1)^k \prod_{i = 1}^k \left(1 - \gamma^{c_i(1-u_i)}\right)$
    is, up to some prefactors, exactly the trivial $2^k$-term decoding polynomial.
\end{itemize}
Secondly, there are known constructions of $3$-term decoders modulo products of two primes~\cite{Efr09,Chee2011}.
Our results clarify that the mechanism enabling these seemingly ad hoc constructions is a root-of-unity grid: in more detail, they work with fields that have distinct nonzero elements $a_1, a_2$ and distinct nonzero elements $t_1, t_2$ such that in the matrix
$$\begin{bmatrix} 1 & 1 \\ 1+a_1t_1 & 1 + a_2t_1 \\ 1+a_1t_2 & 1+a_2t_2\end{bmatrix},$$
all entries in the left column are $p_1$th roots of unity and all entries in the right column are $p_2$th roots of unity.
We will say some more about the $k = 2$ case in the next section.

\subsection{Constructing Root-of-Unity Grids}\label{sec:rughowto}

So far, we have completely characterized when there exists an $S^*_m$-decoding polynomial over sparsity $k+1$, through our notion of a root-of-unity grid.
For what choices of $m = p_1 \cdots p_k$ and $\F$ does such a grid exist?
Prior work~\cite{Efr09,Chee2011} gave two constructions of such grids in the $k = 2$ case:
\begin{itemize}
    \item Efremenko~\cite{Efr09} showed directly that there is an $S^*_{7 \cdot 73}$-decoding polynomial over $\F_{2^9}$ of sparsity 3.
    \item \cite{Chee2011} provided a family of constructions for moduli $m = p_1p_2$ such that $m+1 = 2^d$ for a prime $d$.
\end{itemize}
For general $k$, we put forth two approaches to constructing root-of-unity grids.
The first is based on the probabilistic method and is our main workhorse for the theorems in this work.
The second approach is more algebraic and uses some light Galois theory; although we do not use it for our main results, it has the attractive feature that it provides a more general type of construction that includes Efremenko's original $m = 7 \cdot 73$ example.
Finally, we explain why neither of our approaches seems to capture the construction by~\cite{Chee2011}.
We view this as clear evidence that further improvements on our results (particular for the case of superconstant $s$) should be possible by finding different mechanisms for building root-of-unity grids; we will discuss this more in Section~\ref{sec:openquestions}.

\paragraph{Approach 1: the probabilistic method (Section~\ref{sec:rough-rug}).}
One strategy for understanding this question would be to assume that the characteristic $q$ of $\F$ is $\geq k+1$, and fix the elements $t_1, \ldots, t_{k+1}$ to be $1, \ldots, k-1, k, 0$.
Let $T = \{0, 1, \ldots, k\}$.
Now imagine selecting each $a_i$ independently and uniformly at random over $\F$ (in this overview, we elide technicalities about field elements that we consider possibly being 0).
What is the probability that $1+a_iT \subseteq \mu_{p_i}(\F)$? This approach is attractive because the criteria that each $a_i$ needs to satisfy are independent.
We reason about this probability as follows: 
\begin{itemize}
    \item For each $j \in [k]$, $1+a_i t_j$ is roughly a uniform and random element of $\F$. 
    \item Therefore, the probability that it is an element of $\mu_{p_i}$ is roughly $p_i/|\F|$.
    \item Pretending for a moment that this event is independent for each $j$, the probability that $a_i$ will work for us is roughly $(p_i/\F)^k$.
    This is too small to inspire hope that such an $a_i$ exists; for this, we would want a probability $\geq 1/|\F|$.
\end{itemize}
This should not be surprising; it is already known that even in the $k = 2$ case, combinations of $m$ and $\F$ that admit a 3-term decoding polynomial (equivalently, a root-of-unity grid) are rare~\cite{Chee2011}.
To fix the above idea, we make the following tweak: for each $i \in [k]$, let $d_i$ be minimal such that $q^{d_i}-1$ is divisible by $p_i$.
We pick $\F$ so that it contains a (unique) subfield isomorphic to $\F_{q^{d_i}}$, and \emph{we will only sample $a_i$ from this subfield.}
This changes things dramatically; now the heuristic probability will be $\approx (p_i/q^{d_i})^k$, which can now be much larger.
In particular, if we choose $d_i$ such that $\frac{q^{d_i}-1}{q-1}$ is prime (such a number is known as a \emph{generalized (or base $q$) repunit} because its base $q$ representation consists of \emph{repeated units}---$d_i$ 1's), we can take $p_i$ to be this prime and now the heuristic probability will be $\approx q^{-k}$.
If we set all the $d_i$'s to be $> k$, this will be $> q^{-d_i}$ and we will have succeeded.

At this point, we have heuristically argued that Conjecture~\ref{conjecture:onetuple-intro} ought to imply the existence of $m, \F$ for each $k$ that admit a root-of-unity grid.
Fortunately, this intuition can be made rigorous using the Weil bounds on character sums.
In fact, we can accommodate milder forms of the conjecture---instead of requiring $\frac{q^d-1}{q-1}$ to be prime, we are happy whenever $q^d-1$ has a large prime divisor; we merely discuss the generalized repunit prime variant for its simplicity.
We refer the reader to Section~\ref{sec:suffcond} for details.

\paragraph{Approach 2: kernels of the norm map (Section~\ref{sec:prime-norm-rug}).} 
There is a second, more algebraic approach to building root-of-unity grids, that has the added benefit of capturing Efremenko's original example of a 3-sparse decoding polynomial with $m = 7 \times 73$ over $\F_{2^9}$.
This approach manufactures root-of-unity grid rows $1 + aT$ directly from the norm map $N_{\F_{Q^d}/\F_Q} : \F_{Q^d}^\times \to \F_Q^\times$ that maps $x$ to the product of all its Galois conjugates $x^{1 + Q + \ldots + Q^{d-1}}$ (see Section~\ref{sec:galois} for a review of the very light Galois theory that we use). If $p := \frac{Q^d - 1}{Q-1} = 1 + Q + \ldots + Q^{d-1}$ is prime, then $\mu_p$ is exactly the kernel of the norm map $N_{\F_{Q^d}/\F_Q}$. Therefore, asking for an $a$ such that $1+aT \subseteq \mu_p$ for fixed $T \subseteq \F_Q$ is equivalent to asking for an $a$ such that $1+aT \subseteq \ker N_{\F_{Q^d}/\F_Q}$.  This can be achieved by choosing $a$ to be a root of a carefully designed irreducible polynomial. See Section~\ref{sec:prime-norm-rug} for more details.

\paragraph{Comparison with~\cite{Chee2011} (Section~\ref{sec:cheecomparison}).}
The construction of \cite{Chee2011} gives a different source of root-of-unity grids when $k = 2$ that our theorems do not subsume.
We will discuss this more in Section~\ref{sec:cheecomparison}, but for now we note that we view this as evidence that root-of-unity grids can arise from mechanisms beyond the two sufficient conditions isolated in this paper.
The equivalence between $(k+1)$-sparse decoding polynomials and root-of-unity grids gives a tight structural characterization of what must be constructed, but it does not yet give a sharp arithmetic characterization of when a given pair $(m,\F)$ admits such a grid.
In particular, the construction of~\cite{Chee2011} appears to exploit additional multiplicative structure in finite fields that is not captured either by our probabilistic-method approach or by the norm-kernel approach.

\subsection{PIR with a Superconstant Number of Servers}\label{sec:superconstantoverview}

To understand how the story changes depending on whether the number of servers is constant or superconstant, let us examine the workflow laid out by~\cite{GKS25} to go from an $S^*_m$-decoding polynomial of sparsity $s$ over $\F$ of characteristic $q$, where $m = p_1p_2\cdots p_k$:
\begin{itemize}
    \item Construct an $S^*_{qm}$-matching vector family of size $n$ over $\Z_{qm}$ of some dimension $d$. This will go via the construction of~\cite{Barrington1994,DBLP:journals/combinatorica/Grolmusz00,DGY11}.
    \item Then there exists a PIR scheme with $s$ servers and total communication $O(sd \log |\F|)$.
\end{itemize}
We summarize this in Theorem~\ref{thm:gks}.
When the number of servers $s$ is constant, we can also take $m, |\F|$ to be constant so that the communication is just $O(d)$.
The matching vector construction of~\cite{DBLP:journals/combinatorica/Grolmusz00} works with any product of $k+1$ primes and achieves a dimension of
$$d = \exp(O_m((\log n)^{1/(k+1)}(\log \log n)^{1-1/(k+1)})).$$
The contents of this section are succinctly captured by the phase diagram in Figure~\ref{fig:phasediagram}.

\newcommand{\CC}{\mathsf{CC}}

\begin{figure}
\centering
\begin{tikzpicture}[>=Latex,
  RM/.style     ={red!75!black,   line width=1.3pt},
  ourdec/.style ={blue!65!black,  line width=1.5pt},
  stddec/.style ={orange!92!black,line width=1.3pt},
  common/.style ={violet!65!black,line width=1.6pt},
  grid/.style   ={gray!30, line width=0.4pt},
  vguide/.style ={gray!40, dash pattern=on 1pt off 2pt, line width=0.4pt},
  lead/.style   ={gray!55, line width=0.3pt}
]
  \def\W{10.6}\def\H{11}
  \def\yO{0}\def\ysA{1.0}\def\ysG{1.9}\def\ysB{2.8}\def\ysF{3.7}\def\yP{4.95}
  \def\yRa{6.6}\def\yRb{7.7}\def\yRc{8.8}\def\yRd{9.9}\def\ymax{11}
  \def\xa{0}\def\xb{1.5}\def\xc{3.0}\def\xd{5.0}\def\xe{8.0}\def\xf{10.6}
  \def\xnew{6.5}   
  \def\xg{9.5}     

  \begin{scope}[on background layer]
    \fill[green!16] (\xe,\ysA) -- (\xf,\yO) -- (\xg,\ysA) -- (\xe,\ysG) -- cycle;
  \end{scope}

  \foreach \y in {\yO,\ysA,\ysG,\ysB,\ysF,\yP,\yRa,\yRb,\yRc,\yRd,\ymax}
    {\draw[grid] (0,\y)--(\W,\y);}
  \foreach \x in {\xa,\xb,\xc,\xd,\xnew,\xe,\xf}{\draw[vguide] (\x,0)--(\x,\H);}

  \draw[->,thick] (-0.15,0) -- (\W+0.5,0);
  \draw[->,thick] (0,0) -- (0,\H+0.6);

  \foreach \x in {\xa,\xnew,\xf}{\draw (\x,0)--(\x,-0.12); \draw[lead] (\x,-0.12)--(\x,-0.42);}
  \foreach \x in {\xc,\xe}{\draw (\x,0)--(\x,-0.12); \draw[lead] (\x,-0.12)--(\x,-1.08);}
  \foreach \x in {\xb,\xd}{\draw (\x,0)--(\x,-0.12); \draw[lead] (\x,-0.12)--(\x,-1.74);}
  \node[anchor=north,font=\scriptsize] at (\xa,-0.44) {$\log\log\log n + O(1)$};
  \node[anchor=north,font=\scriptsize] at (\xnew,-0.44) {$\tilde{O}((\log\log n)^{5/9})$};
  \node[anchor=north,font=\scriptsize] at (\xf,-0.44) {$\log\log n$};
  \node[anchor=north,font=\scriptsize] at (\xc,-1.10) {$c\log\log\log n\ \ (c>2)$};
  \node[anchor=north,font=\scriptsize] at (\xe,-1.10) {$(\log\log n)^{2/3}(\log\log\log n)^{2/3}$};
  \node[anchor=north,font=\scriptsize] at (\xb,-1.76) {$(2{+}o(1))\log\log\log n$};
  \node[anchor=north,font=\scriptsize] at (\xd,-1.76) {$\sqrt{\log\log n\,\cdot\,\log\log\log n}$};
  \node[font=\small\bfseries] at (\W/2,-2.80) {$x$-axis: $\log\log \CC$};

  \foreach \y/\lab in {%
     \yO/{$O(1)$},
     \ysA/{$\dfrac{(\log\log n)^{1/3}}{(\log\log\log n)^{2/3}}$},
     \ysG/{$\exp\!\big(\tfrac{(\log\log n)^{1/3}}{(\log\log\log n)^{2/3}}\big)$},
     \ysB/{$\exp\!\big(\sqrt{\tfrac{\log\log n}{\log\log\log n}}\big)$},
     \ysF/{$\exp\!\big(O\big(\tfrac{\log\log n}{\log\log\log n}\big)\big)$},
     \yP/{$(\log n)^{c}\ \ (\forall\,c\in(0,1))$},
     \yRa/{$\dfrac{\log n}{\exp(\tilde{O}((\log\log n)^{2/3}))}$},
     \yRb/{$\dfrac{\log n}{\exp(\tilde{O}(\sqrt{\log\log n}))}$},
     \yRc/{$\dfrac{\log n}{(\log\log n)^{c}}$},
     \yRd/{$\dfrac{\log n}{(\log\log n)^2}$},
     \ymax/{$\dfrac{\log n}{\log\log n}$}}
    {\draw (0,\y)--(-0.12,\y); \node[left,font=\scriptsize] at (-0.18,\y){\lab};}
  \node[rotate=90,font=\small\bfseries] at (-3.6,\H/2) {$y$-axis: \# of servers $s$};

  \draw[ourdec] (\xf,\yO) -- (\xe,\ysA) -- (\xd,\ysB);
  \draw[stddec] (\xf,\yO) -- (\xg,\ysA) -- (\xe,\ysG) -- (\xd,\ysB);
  \draw[common] (\xd,\ysB) .. controls (4.3,3.05) and (3.5,3.5) .. (\xc,\ysF);
  \draw[common] (\xc,\ysF) .. controls (2.6,4.05) and (1.65,4.7) .. (\xb,\yP);
  \draw[common] (\xb,\yP) -- (\xb,\ymax);
  \draw[RM] (\xf,\yO) -- (\xf,\yP);
  \draw[RM] (\xf,\yP) -- (\xe,\yRa) -- (\xd,\yRb) -- (\xc,\yRc) -- (\xb,\yRd) -- (\xa,\ymax);

  \foreach \p in {(\xf,\yP),(\xe,\yRa),(\xd,\yRb),(\xc,\yRc),(\xb,\yRd),(\xa,\ymax)}
    {\fill[red!75!black] \p circle (2.0pt);}
  \foreach \p in {(\xf,\yO),(\xe,\ysA),(\xd,\ysB),(\xc,\ysF),(\xb,\yP),(\xb,\ymax)}
    {\fill[black] \p circle (1.9pt);}
  \fill[orange!92!black] (\xe,\ysG) circle (1.9pt);
  \fill[blue!65!black] (\xnew,\ysG) circle (1.9pt);

  \node[font=\footnotesize\bfseries, red!70!black, rotate=90] at (10.86,2.6) {Reed-Muller};
  \node[font=\footnotesize\bfseries, blue!65!black, fill=white, inner sep=1pt] at (6.7,1.45) {Our Decoder};
  \node[font=\footnotesize\bfseries, orange!92!black, fill=white, inner sep=1pt] at (9.45,1.72) {Std Decoder};
  \node[font=\scriptsize\bfseries, violet!65!black, rotate=90, fill=white, inner sep=1pt] at (1.78,7.0) {Our Decoder $=$ Std Decoder};


\end{tikzpicture}

\caption[Qualitative phase diagram]{\emph{Qualitative phase diagram comparing various PIR constructions across different asymptotic regimes. The $y$-axis is indexed by the number of servers $s$ (as a function of the database size $n$). The $x$-axis is indexed by $\log \log \mathsf{CC}$, where $\mathsf{CC}$ is the communication complexity. The leftmost extreme of $\log \log \log n + O(1)$ corresponds to $\polylog(n)$ communication and the rightmost extreme captures $\poly(n)$ communication and also any $2^{(\log n)^{O(1)}}$ communication.}

\emph{Red curve: PIR based on Reed-Muller codes (\cite{CGKS98}, restated in Theorem~\ref{thm:rmpir}). Orange and violet curves: PIR prior to our work based on matching-vector families and $S^*_m$-decoding polynomials (\cite{Efr09}, restated in Theorem~\ref{thm:manyserver}). Blue and violet curves: PIR in our work---assuming Conjecture~\ref{conjecture:superconstant}---by constructing sparser $S^*_m$-decoding polynomials, stated in Theorems~\ref{thm:fewserver} (green shaded region) and~\ref{thm:intserver} (between green region and violet curve). Green shaded region: regimes of communication complexity where we manage to reduce the server count exponentially compared to~\cite{Efr09} while preserving the communication.}

\emph{As indicated in the diagram, we obtain superpolynomial reductions in communication complexity whenever $s \leq \exp(o(\sqrt{\log \log n/\log \log \log n}))$, and exponential reductions in the server count whenever $\log \log \mathsf{CC} \geq (\log \log n)^{2/3} (\log \log \log n)^{2/3}$.}}
\label{fig:phasediagram}

\end{figure}

\paragraph{Few-server regime (Theorem~\ref{thm:fewserver} and Green Shaded Region in Figure~\ref{fig:phasediagram}).}
Now suppose that the number of servers $s$ grows.
This necessarily forces the number of primes $k$ to grow, which means that $m$ cannot be constant anymore.
Consequently, the hidden dependence on $m$ in the $O_m(\cdot)$ is actually important now; very roughly (see Corollary~\ref{cor:mvparams} for a precise formulation), this dependence hides a multiplicative factor of $m^{1/(k+1)} \geq \Omega(k \log k)$ (since $m$ is a product of $k$ distinct primes).
This is a limitation of any PIR protocol that uses the matching vector constructions of~\cite{DBLP:journals/combinatorica/Grolmusz00}; no matter how many servers we have, the communication cannot drop below $\exp((\log \log n)^{2+o(1)})$.
This is in stark contrast to Reed-Muller PIR, where the communication can be pushed down to $\polylog(n)$ with $O(\log n/\log \log n)$ servers.

The upshot of this is that when building matching-vector PIR for a superconstant number of servers, we need to pay attention to how $m$ grows with $n$.
For example, the construction by~\cite{DG16} uses a decoding polynomial of trivial sparsity $2^k$ and does not require any structure on the $k$ primes, so the optimal choice is to take the first $k$ primes and their product $m$ will grow like $(k \log k)^k$.
In our setting, we can get away with sparsity $k+1$ but at the expense of requiring considerable number-theoretic structure of the primes. This forces the primes to be larger; we heuristically quantify how much larger in Conjecture~\ref{conjecture:superconstant} (which we will discuss in Section~\ref{sec:superconstanttechoverviewnt}).
Under this conjecture, our exponential reduction in the number of servers persists only up to $k \sim (\log \log n)^{1/3}/(\log \log \log n)^{2/3}$.
The takeaway of this is the following:
\begin{quote}
    \begin{center}
        For regimes of communication complexity $\mathsf{CC}$ where $\log \log \mathsf{CC} \geq \Omega((\log \log n)^{2/3}(\log \log \log n)^{2/3})$, our construction (assuming Conjecture~\ref{conjecture:superconstant}) needs exponentially fewer servers than constructions in prior work.
    \end{center}
\end{quote}

\paragraph{Many-server regime (Theorem~\ref{thm:manyserver} and Violet Curve in Figure~\ref{fig:phasediagram}).}
On the other extreme, once the number of servers is $\geq \exp(\Omega(\sqrt{\log \log n/\log \log \log n}))$, Conjecture~\ref{conjecture:superconstant} and our application of it in Section~\ref{sec:gettopir} suggess that our methods do not offer any improvement over prior constructions of matching-vector PIR by~\cite{Efr09,GKS25}.
As we remark in Section~\ref{sec:openquestions}, we believe that different approaches to building exponentially sparser $S$-decoding polynomials should be possible, which may help improve the state of affairs in this many-server regime.
To summarize:
\begin{quote}
    \begin{center}
        For regimes of communication complexity $\mathsf{CC}$ where $\log \log \mathsf{CC} \leq O(\sqrt{\log \log n \cdot \log \log \log n})$, our construction does not appear to improve on prior matching-vector PIR~\cite{Efr09}.
    \end{center}
\end{quote}

\paragraph{Intermediate-server regime (Theorem~\ref{thm:intserver} and Blue Curve Between Green Shaded Region and Violet Curve in Figure~\ref{fig:phasediagram}).}
When the number of servers lies between $(\log \log n)^{1/3}/(\log \log \log n)^{2/3}$ and $\exp(O(\sqrt{\log \log n/\log \log \log n}))$, we can smoothly interpolate between the two extremal approaches detailed so far to find a sweet spot.
To do this, we will need to compose a few different products of $k$ primes, each of which admits a $(k+1)$-sparse decoding polynomial.
This composition technique is standard, due to Itoh and Suzuki~\cite{itoh2010improved}, and employed by~\cite{Chee2011} to build the sparsest-known decoding polynomials prior to our work; we restate their result in Theorem~\ref{thm:tensoring}.
We summarize as follows:
\begin{quote}
    \begin{center}
        For regimes of $s, n$ where $s \leq \exp(o(\sqrt{\log \log n/\log \log \log n}))$, our construction (assuming Conjecture~\ref{conjecture:superconstant}) attains superpolynomially smaller communication complexity than prior matching-vector PIR~\cite{Efr09}.
    \end{center}
\end{quote}

\subsection{Heuristic Justification for Conjecture~\ref{conjecture:superconstant}}\label{sec:superconstanttechoverviewnt}

In the superconstant server regime, as discussed above, it turns out to be optimal to stitch many root-of-unity grids of different sizes over the same characteristic into a single decoding polynomial. Conjecture~\ref{conjecture:superconstant} estimates the cost ($|\F|, m$) of this decoding polynomial, given the sizes of the individual root-of-unity grids we want to stitch together. Taking approach 1 and constructing root-of-unity grids from repunits with large prime factors, estimating this costs boils down to estimating the density of integers $d$ such that $q^d - 1$ has large prime factors. That is, one can reduce Conjecture~\ref{conjecture:superconstant} to a much simpler conjecture about the density of such integers $d$ (Conjecture~\ref{conj:usable-degrees}). 

The density estimate in Conjecture~\ref{conj:usable-degrees} is based on the Dickman heuristic that studies the probability that a random integer has a large prime factor. Treating $q^d-1$ as a random integer, however, ignores a lot of structure---for example, if $d$ is even, $q^d-1 = (q^{d/2} - 1)(q^{d/2} + 1)$, and so $q^d-1$ cannot have a prime factor $\gg q^{d/2}$. One can restrict attention to prime $d$ to minimize such effects, but even this ignores the fact that the largest prime factors of $q^d-1$ can come from a special arithmetic progression. We also account for this structure by heuristically estimating the density of such primes, and corroborate the original estimate given by naive Dickman heuristic.

\section{Preliminaries}\label{sec:prelims}

\paragraph{Notation.}
For a polynomial $P \in \F[X]$, we let $\Coeff_{X^d}(P) \in \F$ denote the coefficient of $X^d$ in $P$.
We use $\Z_m = \Z/m\Z$ to denote the integers modulo $m$. See Section~\ref{sec:alg-nt} for more algebraic and number-theoretic notation and terminology we use.

\subsection{Matching Vector Families}
\label{sec:bg:mv}

We next define matching-vector families and verify some useful facts about them.
Let $a, p_1, p_2, \ldots, p_t$ be distinct primes.
Let $m > 1$ be a positive integer, and let $\Z_m = \Z/m\Z$ denote the ring of integers modulo $m$. 

\begin{definition}[$S$-matching-vector families]
    Let $d$ be a positive integer and $S \subseteq \Z_m$ some subset containing 0.
    We say that two collections of vectors $(U, V)$, where $U = (\uv_1, \dots, \uv_N) \in (\Z_m^d)^N$ and $V = (\vv_1, \dots, \vv_N) \in (\Z_m^d)^N$, form an $S$-matching-vector family over $\Z_m^d$  of size $N$,  if:
    \begin{enumerate}
        \item\label{item:restrictedset} for every $i, j \in [N]$, it holds that $\langle \uv_i, \vv_j \rangle \bmod{m} \in S$, and 
        \item\label{item:selfinnerproduct} for every $i, j \in [N]$, we have $i = j$ if and only if $\langle \uv_i, \vv_j \rangle \equiv 0 \pmod{m}$.
    \end{enumerate}
\end{definition}

\begin{definition}[Canonical set]\label{def:canonicalset}
    Let $\ell \geq 1$ and suppose $p_1, \ldots, p_\ell$ are distinct primes and $m = p_1p_2\ldots p_\ell$.
    Then the canonical set of $\Z_m$ is defined to be the set $S^*_m \subseteq \Z_m$ of residues that are either 0 or 1 mod $p_i$ for all $i \in [\ell]$. (Hence we have $|S^*_m| = 2^\ell$.)
\end{definition}

\begin{theorem}[\cite{Barrington1994,DBLP:journals/combinatorica/Grolmusz00,DGY11}]\label{thm:mvgeneral}
    Let $m$ be the product of $\ell > 1$ distinct primes $p_1, \ldots, p_\ell$, and let $S$ be the canonical set of $\Z_m$.
    Let $h, w$ be positive integers.
    Then there exists a $S^*_m$-matching-vector family over $\Z_m^d$ of size $N$, where:
    \begin{align*}
        N &:= \binom{h}{w},\text{ and}\\
        d &:= \sum_{j = 0}^{\lfloor (mw)^{1/\ell} \rfloor} \binom{h}{j}.
    \end{align*}
\end{theorem}
\begin{proof}[Proof Sketch]
    This is immediate from setting parameters suitably in~\cite[Lemma 11]{DGY11}.
    For each $i \in [\ell]$, we take $e_i \geq 0$ to be minimal such that $p_i^{e_i} > (mw)^{1/\ell}/p_i$.
    This implies that $\prod_i p_i^{e_i} > mw/\prod_i p_i = w$,\footnote{As written,~\cite[Lemma 11]{DGY11} requires $p_i^{e_i} > w^{1/\ell}$ for all $i$, but inspecting the proof implies that it suffices to have $\prod_{i \in [\ell]} p_i^{e_i} > w$. Their proof does not explicitly handle the $e_i = 0$ case but it is straightforward to check that allowing $e_i = 0$ does not affect the proof.} and moreover for any $i$ we have $p_i^{e_i} \leq (mw)^{1/\ell}$.
\end{proof}

\noindent
The following corollary is straightforward; we defer its proof to Appendix~\ref{sec:mvparams}.
It so happens that our main results all only use the second of the two cases stated (where $\eta \geq 2^{\Omega(\ell^2)}$); however, we also include an asymptotically optimal expression for $d$ in the other case for completeness.

\begin{corollary}\label{cor:mvparams}
    Let $\eta = \frac{(\log n)^{\ell-1}}{m}$ and assume $\eta \geq 1$.
    Then for any $n$, there exists a $S^*_m$-matching-vector family over $\Z_m^d$ of size $\geq n$, where:
    \[
        d = \begin{cases}
            \exp\left(O\left(\log n \cdot \frac{1+\frac{1}{\ell}\log \eta}{\eta^{1/\ell}}\right)\right),\text{ if }\eta \leq 2^{O(\ell^2)},\\
            \exp\left(O\left(\log n \cdot \frac{(\log \eta)^{1-1/\ell}}{\ell\eta^{1/\ell}}\right)\right),\text{ if }\eta \geq 2^{\Omega(\ell^2)}.
        \end{cases}
    \]
\end{corollary}

\subsection{$S$-Decoding Polynomials}

Another useful ingredient for us will be that of $S$-decoding polynomials, a notion introduced by~\cite{Efr09}:

\begin{definition}[$S$-decoding polynomials]\label{def:decpoly}
    Let $m \geq 1$ be a positive integer, $S \subseteq \Z_m$ a subset containing 0, and $\F$ a field that contains some primitive $m$th root of unity $\gamma_m$.
    A polynomial $P \in \F[X]$ is called an $(\F, m, S)$-decoding polynomial with respect to $\gamma_m$ if:
    \begin{enumerate}
        \item for every $j \in S \backslash \{0\}$, we have $P(\gamma_m^j) = 0$, and
        \item $P(\gamma_m^0) = P(1) \neq 0$. (This can be strengthened to $P(1) = 1$ without loss of generality by scaling $P$ by a suitable constant factor.)
    \end{enumerate}
    The sparsity of $P$ is the number of nonzero coefficients it has after reducing the exponents modulo $m$. Often, $\F$ and $m$ are clear from the context, and we write $S$-decoding polynomial.
\end{definition}

\begin{remark}
    Strictly speaking, the definition of $P$ as a decoding polynomial should also depend on the choice of root $\gamma_m$.
    However, a straightforward argument (see Lemma~\ref{lemma:rootnomatter}) shows that the existence of $P$ within a certain sparsity is independent of the particular $\gamma_m$ that is chosen.
    So we will abuse terminology slightly and omit $\gamma_m$ when referring to $S$-decoding polynomials.
\end{remark}

We begin by recalling a standard observation about trivially sparse decoding polynomials always existing~\cite{Efr09}:

\begin{claim}\label{clm:trivialdec}
    Let $Q, m$ be positive integers with $m$ prime, $Q$ a prime power, and $m \mid Q-1$.
    (Note that $S^*_m = \{0, 1\}$.)
    Then there exists an $(\F_Q, m, S^*_m)$-decoding polynomial of sparsity 2.
\end{claim}
\begin{proof}
    Let $\gamma \in \F_Q$ be any primitive $m$th root of unity.
    We can simply take $P(X) := X-\gamma$, which has sparsity 2.
    We have $P(\gamma^1) = 0$ and $P(\gamma^0) = P(1) = 1-\gamma \neq 0$, as desired.
\end{proof}

The following theorem is due to~\cite{itoh2010improved}; for completeness, we reproduce its proof in Appendix~\ref{sec:tensoringproof}.

\begin{theorem}[Composition theorem~\cite{itoh2010improved}]\label{thm:tensoring}
    Let $q, m_1, \ldots, m_r$ be pairwise coprime squarefree positive integers with $q$ prime.
    Assume that for every $i \in [r]$ there is an $(\F_{q^{c_i}}, m_i, S^*_{m_i})$-decoding polynomial of sparsity $d_i$ for some integer $c_i \geq 1$.
    Then there is also an $(\F_{q^{\mathrm{lcm}(c_1, \ldots, c_r)}}, m_1\ldots m_r, S^*_{m_1\ldots m_r})$-decoding polynomial of sparsity $d_1 \ldots d_r$.
\end{theorem}

\begin{corollary}\label{cor:trivialdecodertensored}
    Let $\ell \geq 5$ be a positive integer.
    Then there exists $m$ that is a product of $\ell$ distinct primes and a field $\F$ of characteristic 2 such that there exists an $S^*_m$-decoding polynomial over $\F$ of sparsity $\leq 2^\ell$.
    Moreover, we have the following bounds:
    \begin{align*}
        m &\leq (\ell \log \ell)^{\ell} \\
        \log |\F| &\leq \exp\left(O\left(\ell \log \ell\right)\right).
    \end{align*}
\end{corollary}
\begin{proof}
    We apply Theorem~\ref{thm:tensoring} with $q = 2$, $r = \ell$, and $m_1, \ldots, m_\ell$ the first $\ell$ odd primes.
    For each $i \in [\ell]$, let $c_i$ be the multiplicative order of 2 modulo $m_i$.
    Then by Claim~\ref{clm:trivialdec} there exists an $(\F_{2^{c_i}}, m_i, S^*_{m_i})$-decoding polynomial of sparsity 2 for every $i \in [\ell]$.
    So Theorem~\ref{thm:tensoring} tells us that there exists an $(\F_{2^{\mathrm{lcm}(c_1, \ldots, c_\ell)}}, m_1\ldots m_\ell, S^*_{m_1\ldots m_\ell})$-decoding polynomial of sparsity $\leq 2^\ell$.

    It remains to establish the stated bounds on $m := m_1\ldots m_\ell$ and $|\F| = 2^{\mathrm{lcm}(c_1, \ldots, c_\ell)}$.
    First, $m$ is simply the product of the first $\ell$ odd primes, which is $\leq (\ell \log \ell)^\ell$ by Lemma~\ref{lem:primorial}.
    Secondly, $\mathrm{lcm}(c_1, \ldots, c_\ell)$ is exactly the multiplicative order of 2 modulo $m$.
    It must therefore be at most $m$.
    Consequently, we have:
    \begin{align*}
        \log |\F| &\leq O(\mathrm{lcm}(c_1, \ldots, c_\ell)) \\
        &\leq O(m) \\
        &\leq O((\ell \log \ell)^{\ell}) \\
        &= \exp(O(\ell \log \ell)),
    \end{align*}
    as claimed.
\end{proof}

\subsection{Previous Constructions of PIR}

\subsubsection{From Matching Vectors and Decoding Polynomials}

Below we state the main theorem of~\cite{GKS25}, reworded in such a way that it robustly holds even when the number of servers and the modulus $m$ are superconstant.

\begin{theorem}[Essentially Theorem 2.3 of \cite{GKS25}]\label{thm:gks}
    Let $m$ be a product of $k$ distinct primes, let $q$ be a prime with $\gcd(q, m) = 1$, and let $\F$ be a finite field of characteristic $q$ that contains a primitive $m$th root of unity.
    Assume the following two ingredients:
    \begin{itemize}
        \item an $S^*_{qm}$-matching-vector family over $\Z_{qm}^d$ of size $n$; and
        \item an $S^*_m$-decoding polynomial over $\F$ with sparsity $\leq T$.
    \end{itemize}
    Then there is a $T$-server PIR scheme for a database of size $n$ with total communication $O(Td \log |\F|)$.
    In more detail, each query from the user is an element of $\Z_m^d$, and each answer from each server is an element of $\F^{d+1}$.
\end{theorem}
\begin{proof}[Proof Sketch]
    As written, the construction by~\cite{GKS25} has the queries live in $\F^d$ and the answers live in $\F^{d+1}$, but on closer inspection every coordinate of the users' queries is an $m$th root of unity so it can be implicitly represented as an element of $\Z_m$.
\end{proof}

Plugging in the standard matching-vector family construction of~\cite{Barrington1994,DBLP:journals/combinatorica/Grolmusz00} and the trivial $2^k$-sparse decoding polynomial leads to the below theorem~\cite{Efr09,DGY11}.
For completeness, we re-derive this communication bound in Section~\ref{sec:manyserver}.

\begin{theorem}[Many-server matching-vector PIR~\cite{Efr09,DGY11}]\label{thm:manyserver}
    For any $n, s$ such that $32 \leq s \leq \log n$, there exists an $s$-server PIR for a database of size $n$ achieving total communication:
    $$\exp\left(O\left((\log n)^{\frac{1}{\lfloor \log_2 s \rfloor + 1}} \cdot (\log \log n)^{1-\frac{1}{\lfloor \log_2 s \rfloor + 1}} \cdot \log s \cdot \log \log s\right)\right).$$
    This can in turn be more coarsely bounded by:
    $$\exp\left(\exp\left(O\left(\log \log \log n + \frac{\log \log n}{\log s}\right)\right)\right).$$
\end{theorem}

\subsubsection{From Reed-Muller Codes}

Another approach to PIR is that based on Reed-Muller codes~\cite{CGKS98,BIKR02,WY05}.
We summarize the performance of Reed-Muller PIR below:
\begin{theorem}[Reed-Muller-based PIR~\cite{CGKS98,BIKR02,WY05}]\label{thm:rmpir}
    For any $n, s$, there exists an $s$-server PIR protocol achieving total communication:
    \begin{align*}
        \mathsf{CC} &= \min\left(2^{\tilde{O}(s)} \cdot n^{\frac{2 \log \log s}{s \log s}}, s^2 \log s \cdot n^{\frac{1}{2s-1}}\right).
    \end{align*}
    The former expression is due to~\cite{BIKR02} and is preferable for $s \leq \sqrt{\log n} \cdot \poly(\log \log n)$; the latter expression is due to~\cite{CGKS98,WY05} and is preferable for $s \geq \sqrt{\log n} \cdot \poly(\log \log n)$.
\end{theorem}
We note that although the comparison between~\cite{BIKR02} and~\cite{WY05} is somewhat complicated, for the level of coarseness of our comparisons in Figure~\ref{fig:phasediagram}, the improvements by~\cite{BIKR02} in some regimes will not matter.
We could have even used the original $\poly(s) \cdot n^{1/(s-1)}$ communication construction by~\cite{CGKS98} without any impact on the contents of Figure~\ref{fig:phasediagram}.
Our statement of the results by~\cite{BIKR02,WY05} is only for completeness.

\subsection{Algebraic and Number-Theoretic Background}\label{sec:alg-nt}
\paragraph{Notation and terminology.} For a positive integer $n$, we write $P^+(n)$ to denote the largest prime factor of $n$. For integers $a, m$ such that $\gcd(a,m) = 1$, we define $\ord_m(a)$ to be the smallest positive integer $d$ such that $a^d \equiv 1 \pmod{m}$. It is well-known that $a^d \equiv 1 \pmod{m} \Leftrightarrow \ord_m(a) \mid d$. For a positive real number $x$, we write $\pi(x)$ to count the number of prime integers $0 < n \le x$. We use asymptotic notation $O(\cdot)$ to bound the \emph{absolute value} of quantities; for example we write $\epsilon \le O(1/n)$ to mean $|\epsilon| \le O(1/n)$.

\subsubsection{Polynomials and Finite Fields}

\begin{lemma}[Lagrange interpolation]\label{lemma:lagrange}
    Let $P \in \F[X]$ be a polynomial of degree $\leq k$.
    Let $t_1, \ldots, t_{k+1} \in \F$ be distinct points.
    Then we have:
    \[
    P(X) = \sum_{j = 1}^{k+1} P(t_j) \cdot \frac{\prod_{\ell \neq j} (X-t_\ell)}{\prod_{\ell \neq j} (t_j - t_\ell)}.
    \]
    In particular, equating coefficients of $X^k$ implies:
    \[
    \Coeff_{X^k}(P) = \sum_{j = 1}^{k+1} P(t_j) \cdot \frac{1}{\prod_{\ell \neq j} (t_j - t_\ell)}.
    \]
\end{lemma}

\begin{lemma}[Multilinear polynomial interpolation]\label{lem:multilinearinterp}
    Let $P \in \F[X_1, \ldots, X_k]$ be a multilinear polynomial (i.e., it has individual degree $\leq 1$).
    Then we have:
    $$\Coeff_{X_1X_2\ldots X_k} [P] = \sum_{\bv \in\{0,1\}^k} (-1)^{k-|\bv|}P(\bv),$$
    where $|\bv|$ denotes the Hamming weight of $\bv$.
\end{lemma}

\begin{definition}[Multiplicative structure of finite fields]
    Let $\F = \F_Q$ be a finite field of order $Q$.
    We denote by $\F^\times$ either the set $\F \backslash \{0\}$ or the multiplicative group that these elements form, depending on context.
    $\F^\times$ is a cyclic group of order $Q-1$.
    For $\F$ a finite field and $m$ a positive integer dividing $|\F|-1$, we let $\mu_m(\F)$ denote the cyclic subgroup of $\F^\times$ consisting of $m$th roots of unity.
\end{definition}

\subsubsection{Multiplicative Characters}

\begin{definition}[Multiplicative character]\label{def:character}
    For a finite field $\F_Q$, a multiplicative character is a group homomorphism $\chi: \F_Q^\times \to \C$.
    The collection of multiplicative characters forms a cyclic group of order $Q-1$, which we denote by $\widehat{\F_Q^\times}$.
    The order of a particular character $\chi$ is its order in this group.
    The trivial multiplicative character is the identity element of $\widehat{\F_Q^\times}$, and it maps every nonzero element to 1.
    
    We will sometimes abuse notation and extend $\chi$ to a function on all of $\F$ by defining $\chi(0) = 0$.
    Note that this extension still satisfies $\chi(xy) = \chi(x)\chi(y)$ for all $x, y \in \F$.
\end{definition}

\begin{definition}[Annihilator]\label{def:annihilator}
    For $m, Q$ such that $m | Q-1$, we define the annihilator subgroup $\mu_m^\perp(\F) \subseteq \widehat{\F_Q^\times}$ by:
    \[
    \mu_m^\perp(\F) := \{\chi \in  \widehat{\F_{Q}^\times} : \chi(x) = 1 \text{ for all } x \in \mu_m(\F_{Q})\}.
    \]
    $\mu_m^\perp(\F)$ is cyclic and consists of $(Q-1)/m$ elements.
\end{definition}

\begin{lemma}\label{lemma:characterorthogonality}
    For any $x \in \F_Q$, we have:
    \[
    \sum_{\chi \in \mu_m^\perp(\F_Q)} \chi(x) = \begin{cases}
        (Q-1)/m, \text{ if }x \in \mu_m(\F_Q),\\
        0,\text{ otherwise.}
    \end{cases}
    \]
\end{lemma}
\begin{proof}
    If $x \in \mu_m(\F_Q)$, then by definition we have $\chi(x) = 1$ for all $\chi \in \mu_m^\perp(\F_Q)$, so the sum will be $(Q-1)/m$.
    If $x = 0$, then we have $\chi(x) = 0$ for all $\chi$, so the sum will indeed be 0.
    For $x \notin \mu_m(\F_Q) \cup \{0\}$, note that for any $\chi' \in \mu_m^\perp(\F_Q)$ we have:
    \begin{align*}
    \chi'(x) \sum_{\chi \in \mu_m^\perp(\F_Q)} \chi(x) &= \sum_{\chi \in \mu_m^\perp(\F_Q)} (\chi'\chi)(x) = \sum_{\chi \in \mu_m^\perp(\F_Q)} \chi(x) \\ \Rightarrow (\chi'(x) - 1)\sum_{\chi \in \mu_m^\perp(\F_Q)} \chi(x) &= 0.
    \end{align*}
    Since $x \notin \mu_m(\F_Q)$, there will exist a choice of $\chi'$ such that $\chi'(x) \neq 1$, so it follows that the sum must be 0.
\end{proof}

\begin{lemma}[Weil bound, Theorem 2C' on Page 43 of~\cite{schmidt2006equations}]\label{lemma:weil}
    Let $\F_Q$ be a finite field, and let $\chi$ be a nontrivial multiplicative character of $\F_Q^\times$ of order $h > 1$.
    Let $f \in \F_Q[X]$ be a nonconstant polynomial that is not of the form $a \cdot g(X)^h$ for some $a \in \F_Q^\times$ and $g \in \F_Q[X]$.
    Suppose that $k$ is the number of distinct roots of $f$ in the algebraic closure $\overline{\F_Q}$. Then the following holds:
    \begin{align*}
        \left|\sum_{\lambda \in \F_Q} \chi(f(\lambda))\right| \le (k-1) \sqrt{Q}.
    \end{align*}
\end{lemma}

\subsubsection{Galois Theory}\label{sec:galois}

Throughout this section, let $Q$ be a prime power and $d > 1$ an integer.
All definitions and facts here are with respect to the field extension $\F_{Q^d}/\F_Q$.

\begin{definition}[Automorphisms]
    A mapping $\sigma: \F_{Q^d} \to \F_{Q^d}$ is said to be an \emph{automorphism} if it is a field homomorphism, bijective, and satisfies $\sigma(x) = x$ for every $x \in \F_Q$.
    Such automorphisms form a group which we denote by $\Aut(\F_{Q^d}/\F_Q)$.
    
    This group has order exactly $d$. In fact, every automorphism has the form $\sigma(x) = x^{Q^i}$ for some $i \in [0, d-1]$.
\end{definition}

\begin{definition}[Norm]
    For an element $\alpha \in \F_{Q^d}$, we define the norm $N_{\F_{Q^d}/\F_Q}(\alpha)$ to be:
    $$N_{\F_{Q^d}/\F_Q}(\alpha) = \prod_{\sigma \in \Aut(\F_{Q^d}/\F_Q)} \sigma(\alpha) = \alpha^{1+Q+\ldots+Q^{d-1}}.$$
    In particular, when $\alpha \in \F_Q$, we have $N_{\F_{Q^d}/\F_Q}(\alpha) = \alpha^d$.
\end{definition}

\begin{lemma}
    For any $\alpha \in \F_{Q^d}^\times$, we have $N_{\F_{Q^d}/\F_Q}(\alpha) \in \F_Q^\times$.
    In fact, $N_{\F_{Q^d}/\F_Q}$ is a surjective group homomorphism from $\F_{Q^d}^\times$ to $\F_Q^\times$.
\end{lemma}

\subsubsection{Cyclotomic Polynomials}

\begin{definition}[Cyclotomic polynomials]\label{defn:cyclotomic-poly}
    Let $n$ be a positive integer. The $n$th cyclotomic polynomial is the unique monic irreducible polynomial with integer coefficients that is a divisor of $x^n - 1$ and is not a divisor of $x^k - 1$ for any $k < n$.
\end{definition}
\begin{fact}\label{cyclotomic-roots-of-unity}
\begin{align*}
    x^n - 1 = \prod_{d | n} \Phi_d(x).
\end{align*}
\end{fact}




\subsubsection{Prime Density Estimates}

\begin{lemma}[Bertrand, Chebyshev]\label{lem:bertrand}
    If $n > 1$, then there is at least one prime $p$ such that $n < p < 2n$.
\end{lemma}

\begin{lemma}\label{lem:primorial}
    For any $\ell \geq 5$, the product of the first $\ell$ odd primes is $\leq (\ell \log \ell)^{\ell}$.
\end{lemma}

\begin{definition}[Rough numbers]\label{defn:rough}
    Let $\epsilon \in (0, 1/2)$. Recall that for a positive integer $n$, we define $P^+(n)$ to be the largest prime divisor of $n$.
    We say that a positive integer $n$ is $\epsilon$-rough if $P^+(n) \ge n^{(1 - \epsilon)}$.\footnote{We note that this is non-standard terminology; usually, roughness refers to the property where \emph{all} prime divisors of the number are large, not just the largest prime divisor. When $\epsilon < 1/2$, the standard definition is essentially meaningless because it is equivalent to $n$ being prime.}
\end{definition}

\begin{lemma}\label{lem:rough-order}
    Let $q$ be a prime, and let $d \ge 6$ be a positive integer such that $q^d - 1$ is $\epsilon$ rough for some $0 < \epsilon \le 1/6$. Then, if $p = P^+(q^d - 1)$, we must have that $\ord_p(q) = d$.
\end{lemma}
\begin{proof}
    Observe that $\ord_p(q) | d$.
    Suppose for contradiction that $\ord_p(q) = e < d$. Then, $e \le d/2$, so $p| q^e - 1 \Rightarrow p \leq q^e-1 \le q^{d/2}-1$, and so $p \le q^e - 1 \le q^{d/2}-1 \le q^{d/2}$. This is a contradiction, since by hypothesis, we must have that
    \begin{align*}
        p \ge (q^{d}-1)^{(1-\epsilon)} \ge (q^d - 1)^{5/6} > q^{d/2},
    \end{align*}
    where the last inequality follows from the following calculation:
    \begin{align*}
        \frac{(q^d  -1)^{5/6}}{q^{d/2}} &= \left(q^{2d/5} - \frac{1}{q^{3d/5}}\right)^{5/6} \\
        &\ge \left(q^{2d/5} - 1\right)^{5/6} &\text{since } q>1\\
        &\ge (2^{12/5} - 1)^{5/6} &\text{since } q \ge 2, d \ge 6\\
        &\ge 3.27^{5/6} > 1.
    \end{align*}
\end{proof}

\begin{lemma}\label{lem:distinct-degrees-implies-distinct-primes}
    Let $q$ be a prime. Suppose that $d, d' \ge 6$ are distinct positive integers such that $q^d - 1$ and $q^{d'} - 1$ are $\epsilon$-rough, for some $0 < \epsilon \le 1/6$. Then,
    \begin{align*}
        P^+(q^d - 1) \neq P^+(q^{d'} - 1).
    \end{align*}
\end{lemma}
\begin{proof}
    Let $p = P^+(q^d - 1)$ and $p' = P^+({q^{d'} - 1})$. Lemma~\ref{lem:rough-order} shows that $\ord_p(q) = d$, and similarly $\ord_{p'}(q) = d'$. Since the order is well-defined, we are done. 
\end{proof}

\begin{lemma} [\cite{dickman1930frequency}, Equation (1) of \cite{moree2014integers}]\label{lem:dickman}
    For every $\alpha > 0$,
    \begin{align*}
        \lim _{N\to \infty} \Pr_{n \gets[1, N]} [P^+(n) \le N^{\alpha}]  = \rho(1/\alpha),
    \end{align*}
    where $\rho$ is the Dickman-de Bruijn function, and $\rho(1/ \alpha) = 1 - \log (1/\alpha)$ for $1/2 \le \alpha \le 1$.
\end{lemma}


The following strengthening of Mertens theorem can be proven assuming the Generalized Riemann Hypothesis (See Section~\ref{sec:strong-mertens}).
\begin{lemma}[Strong Mertens' Theorem] \label{lem:mertens}
    Assume the GRH for all Dirichlet $L$-functions. There is a universal constant $C > 0$ such that for all integers $m \ge 1$, and $a$ with $\gcd(a, m) = 1$, and all real numbers $y \ge x \ge 8$ such that $m \le x$,
    \begin{align*}
        \left|\Bigl(\sum_{\substack{\text{prime }p \in [x,y] \\ p = a \bmod m}} \frac{1}{p}\Bigr) - \frac{1}{\phi(m)} (\log \log y - \log \log x)\right| \le C \frac{\log x}{\sqrt{x}}.
    \end{align*}
\end{lemma}

\section{New Constructions of $S$-Decoding Polynomials}\label{sec:decpolysmake}

In this section, we first reduce the task of constructing a sparse decoding polynomial to the (equivalent) task of constructing a primitive we call ``root-of-unity grid'' (Definition~\ref{defn:rug}). Section~\ref{sec:equivalent-poly-rug} shows this equivalence. We then give two approaches to constructing these root-of-unity grids: (1) sequences of rough repunits (Section~\ref{sec:rough-rug}) and (2) prime field norms (Section~\ref{sec:prime-norm-rug}). The relevant algebraic and number-theoretic background can be found in Section~\ref{sec:alg-nt}, and will be referenced when required.
Finally, Section~\ref{sec:stitch} shows how to stitch many root-of-unity grids together to obtain a sparse decoding polynomial. This will be relevant in the regime of superconstant number of servers.

\subsection{Equivalence of $S$-Decoding Polynomials and Root-of-Unity Grids}\label{sec:equivalent-poly-rug}

\begin{definition}[Root-of-unity grid]\label{defn:rug}
    Let $m = p_1 \ldots p_k$ be a product of $k$ distinct primes, and let $\F$ be a finite field with $m | (|\F|-1)$.\footnote{Therefore there exists a primitive $m$th root of unity.}
    A $(\F, m)$-\emph{root-of-unity grid} is a collection of elements $a_1, \ldots, a_k \in \F \backslash \{0\}$ and distinct elements $t_1, \ldots, t_{k+1} \in \F$ such that
    \begin{align*}
        1 + a_i t_j \in \mu_{p_i}(\F),
    \end{align*}
    for all $i \in [k]$, $j \in [k+1]$.
\end{definition}

The following straightforward claim is not necessary for our results, but we present it for intuition.
\begin{claim}
    In any root-of-unity grid, the elements $a_1, \ldots, a_k$ must be distinct.
\end{claim}
\begin{proof}
    If they are not, then there exist different $i, i' \in [k]$ such that $1+a_it_j = 1+a_{i'}t_j$ for all $j \in [k+1]$. This common element lies in $\mu_{p_i}(\F) \cap \mu_{p_{i'}}(\F)$, so it must be 1.
    This forces $t_j = 0$ for all $j \in [k+1]$ which contradicts the requirement that the $t_j$'s are distinct.
\end{proof}

\subsubsection{Root-of-Unity Grid $\Rightarrow$ Sparse Decoding Polynomial}

\begin{lemma}\label{lem:rugisgood}
    Let $m = p_1 \cdots p_k$ be a product of $k$ distinct primes, and let $\F$ be a finite field with $m | (|\F| - 1)$. Suppose that there exists a $(\F, m)$-root-of-unity grid.
    Then there exists an $S^*_m$-decoding polynomial over $\F$ of sparsity $k+1$, for some primitive $m$th root of unity $\gamma$.
\end{lemma}
\begin{proof}
    Let $\gamma \in \F$ be a primitive $m$th root of unity, and let $a_1, \ldots, a_k \in \F \backslash \{0\}$ and $t_1, \ldots, t_{k+1} \in \F$ form a $(\F, m)$-root-of-unity grid. That is, $1 + a_i t_j \in \mu_{p_i}(\F)$ for all $i \in [k]$, $j \in [k+1]$.
    For $j \in [k+1]$, define the Lagrange weights
    \begin{align*}
        w_j := \frac{1}{\prod_{\ell \neq j} (t_j - t_\ell)}.
    \end{align*}
    Let $c_1,\ldots,c_k \in \Z_m$ form a CRT basis for $m$ i.e., we have $c_i \bmod p_\ell = 1$ if $\ell = i$ and 0 otherwise.
    For each $i \in [k]$, let $\gamma_i := \gamma^{c_i}$.
    $\gamma_i$ is a generator of $\mu_{p_i}(\F)$, so by the root-of-unity grid condition there exist $e_{i,j}$ such that $\gamma_i^{e_{i,j}} = (1 + a_i t_j)^{-1}$.
    Then, define
    \begin{align*}
        P(X) &:= \sum_{j=1}^{k+1} \left(w_j \cdot \prod_{i = 1}^k (1 + a_i t_j)\right) \cdot X^{\sum_{\ell=1}^k e_{\ell, j} \cdot c_\ell}.
    \end{align*}
    One can check that $P(X)$ has sparsity exactly $k+1$ (no two terms have the same power of $X$ because $t_1, \ldots, t_{k+1}$ are distinct so the tuples $\{(e_{1, j}, e_{2, j}, \ldots, e_{k, j}): j \in [k+1]\}$ are also distinct).
    We now claim that $P(X)$ is an $S^*_m$-decoding polynomial.
    Indeed, consider any $x \in S^*_m$.
    We may write $x = \sum_{i = 1}^k u_i c_i$ where $u_i \in \{0,1\}$ for all $i$.
    Then we have:
    \begin{align*}
        P(\gamma^x) &= \sum_{j=1}^{k+1} \left(w_j \cdot \prod_{i = 1}^k (1 + a_i t_j)\right) \cdot \gamma^{x(\sum_{\ell=1}^k e_{\ell, j} \cdot c_\ell)}\\
        &= \sum_{j=1}^{k+1} \left(w_j \cdot \prod_{i = 1}^k (1 + a_i t_j)\right) \cdot \gamma^{\sum_{\ell, \ell'=1}^k e_{\ell, j} \cdot u_{\ell'} \cdot c_\ell \cdot c_{\ell'}} &\text{writing $x = \sum u_{\ell'} c_{\ell'}$}\\
        &= \sum_{j=1}^{k+1} \left(w_j \cdot \prod_{i = 1}^k (1 + a_i t_j)\right) \cdot \gamma^{\sum_{\ell=1}^k e_{\ell, j} \cdot u_{\ell} \cdot c_\ell} &\text{since $c_\ell c_{\ell'} \equiv \mathbbm{1}[\ell=\ell']c_\ell \pmod{m}$}\\
        &= \sum_{j=1}^{k+1} \left(w_j \cdot \prod_{i = 1}^k (1 + a_i t_j)\right) \cdot \prod_{\ell=1}^k \gamma_\ell^{ e_{\ell, j} \cdot u_{\ell}} &\text{by definition of $\gamma_\ell$}\\
        &= \sum_{j=1}^{k+1} w_j \cdot \prod_{i = 1}^k (1 + a_i t_j) ^{1 - u_{i}} &\text{by definition of $e_{\ell, j}$}.
    \end{align*}
    Now define the univariate polynomial $f_{u_1, \ldots, u_k}(X) = \prod_{i = 1}^k (1 + a_i X) ^{1 - u_{i}}$. 
    It has degree $\leq k$.
    Applying Lemma~\ref{lemma:lagrange} to $f_{u_1, \ldots, u_k}$ and the evaluation points $t_1, \ldots, t_{k+1}$ implies:
    \begin{align*}
        \Coeff_{X^k}(f_{u_1, \ldots, u_k}) &= \sum_{j = 1}^{k+1} f_{u_1, \ldots, u_k}(t_j) \cdot \frac{1}{\prod_{\ell \neq j} (t_j - t_\ell)} \\
        &= \sum_{j = 1}^{k+1} w_j f_{u_1, \ldots, u_k}(t_j) \\
        &= P(\gamma^x),
    \end{align*}
    by the previous computation.
    Now we have two cases:
    \begin{itemize}
        \item If $x \neq 0$, then we have $(u_1, \ldots, u_k) \neq (0, \ldots, 0)$, so $f_{u_1, \ldots, u_k}$ has degree $\leq k-1$, implying $P(\gamma^x) = 0$.
        \item If $x = 0$, then we have $(u_1, \ldots, u_k) = (0, \ldots, 0)$ so we have $P(\gamma^x) = \Coeff_{X^k} (f_{u_1, \ldots, u_k}) = a_1a_2\ldots a_k \neq 0$.
    \end{itemize}
    This completes the proof.
\end{proof}

\subsubsection{Sparse Decoding Polynomial $\Rightarrow$ Root-of-Unity Grid}

We begin with a straightforward preliminary lemma:

\begin{lemma}\label{lem:affinerank1}
    Assume $k \geq 2$.
    Let $\{z_{i, j} \in \F \backslash \{0\}: i \in [k], j \in [k+1]\}$ be such that for any distinct $\ell, a, b \in [k+1]$ and distinct $r,s \in [k]$ we have:
    $$(z_{r, \ell} - z_{r, a})(z_{s, \ell} - z_{s, b}) = (z_{r, \ell} - z_{r, b})(z_{s, \ell} - z_{s, a}) \neq 0.$$
    Then there exist $a_1, \ldots, a_k \in \F \backslash \{0\}$, and distinct $t_1, \ldots, t_{k+1} \in \F$ such that $z_{i, j} = z_{i, 1}+a_it_j$ for all $i \in [k], j \in [k+1]$.
\end{lemma}
\begin{proof}
    For any $r \in [k], a \in [k+1]$, define $M_{r, a} = z_{r, a} - z_{r, 1}$.
    Then for any distinct $a, b \in [2, k+1]$ and $r, s \in [k]$, we have:
    \begin{align*}
        M_{r, a} M_{s, b} &= (z_{r, a} - z_{r, 1})(z_{s, b} - z_{s, 1}) \\
        &= (z_{r, b} - z_{r, 1})(z_{s, a} - z_{s, 1}) \\
        &= M_{r, b} M_{s, a}.
    \end{align*}
    Moreover, we have $M_{r, a} = 0 \Leftrightarrow a = 1$.
    Consequently, the above remains true if either $a = 1$ or $b = 1$; both sides will simply be zero.
    It follows $M$ is a rank 1 matrix i.e., there exist $a_1, \ldots, a_k, t_2, \ldots, t_{k+1} \in \F \backslash \{0\}$ and $t_1 = 0$ such that $M_{i, j} = a_i t_j$ for all $i \in [k], j \in [k+1]$.
    Consequently, we have $z_{i, j} = z_{i, 1} + a_i t_j$.
    It remains to check the nonzero and distinctness conditions:
    \begin{itemize}
        \item The $a_i$'s are nonzero by construction.
        \item Finally, if $t_j = t_{j'}$ this would imply $z_{i, j} = z_{i, j'} \Rightarrow j = j'$. So the $t_j$'s are also distinct.
    \end{itemize}
\end{proof}

The following converse of Lemma~\ref{lem:rugisgood} builds on the second approach by~\cite{DBLP:conf/innovations/GhasemiK26} to show that there does not exist an $S^*_m$-decoding polynomial over $\F$ of sparsity $\leq k$:

\begin{lemma}\label{lemma:gkrigidity}
    Assume $k \geq 2$.
    Let $\F, m$ (with $m$ a product of $k$ distinct primes) be such that there is an $S^*_m$-decoding polynomial over $\F$ of sparsity $k+1$.
    Then there exists a root-of-unity grid.
\end{lemma}
\begin{proof}
    Again, let $\gamma \in \F$ be a primitive $m$th root of unity, $c_1, \ldots, c_k \in \Z_m$ a CRT basis for $m$, and for each $i \in [k]$ let $\gamma_i = \gamma^{c_i}$.
    Let us write the $S$-decoding polynomial as:
    $$P(X) = \sum_{j = 1}^{k+1} \alpha_j X^{\beta_j},$$
    for $\alpha_j \neq 0 \forall j \in [k+1]$.
    Assume without loss of generality that $P(1) = 1$ (by scaling).
    In the following, let $z_{i, j} = \gamma_i^{-\beta_j} \in \mu_{p_i}(\F)$ for every $i \in [k], j \in [k+1]$.
    Then the $S$-decoding property tells us that for any $(u_1, \ldots, u_k) \in \{0, 1\}^k$, we have:
    \begin{align*}
        \sum_{j = 1}^{k+1} \alpha_j \prod_{i = 1}^k \gamma_i^{u_i \beta_j} = P\left(\prod_{i = 1}^k \gamma_i^{u_i}\right) = \begin{cases}
            1,\text{ if }u_i = 0\forall i \\ 
            0,\text{ otherwise.}
        \end{cases}
    \end{align*}
    To analyze this condition, introduce formal variables $X_1, \ldots, X_k$ and notice that:
    \begin{align}
        \sum_{j = 1}^{k+1} \alpha_j \prod_{i = 1}^k (1+X_i z_{i, j}^{-1}) &= \sum_{j = 1}^{k+1} \alpha_j \left(\sum_{u_1, \ldots, u_k \in \zo} \prod_{i = 1}^k X_i^{u_i} \gamma_i^{u_i\beta_j}\right) \nonumber \\
        &= \sum_{u_1, \ldots, u_k \in \zo} \left(\prod_{i = 1}^k X_i^{u_i}\right) \cdot \sum_{j = 1}^{k+1} \alpha_j \prod_{i = 1}^k \gamma_i^{u_i\beta_j} \nonumber \\
        &= 1.\label{eq:funkyidentity}
    \end{align}
    The idea is to argue that $\{z_{i, j}: i \in [k], j \in [k+1]\}$ forms our desired root-of-unity grid (up to some rescaling).
    To this end, consider any permutation $\sigma: [k+1] \to [k+1]$.
    Then we can substitute $X_i = -z_{i, \sigma(i)}$ into~\eqref{eq:funkyidentity}.
    The only term on the LHS that will survive is that corresponding to $j = \sigma(k+1)$, yielding:
    \begin{align}
        \alpha_{\sigma(k+1)} \prod_{i = 1}^k \left(1-\frac{z_{i, \sigma(i)}}{z_{i, \sigma(k+1)}}\right) &= 1 \nonumber\\
        \Rightarrow \prod_{i = 1}^k (z_{i, \sigma(k+1)} - z_{i, \sigma(i)}) &= \frac{1}{\alpha_{\sigma(k+1)}} \prod_{i = 1}^k z_{i, \sigma(k+1)}.\label{eq:funkyidentitysubin}
    \end{align}
    We begin by noting that the RHS of~\eqref{eq:funkyidentitysubin} is clearly nonzero, so this must also be true of the LHS.
    Consequently, for any $i \in [k]$ and distinct $j, j' \in [k+1]$ we have $z_{i, j} \neq z_{i, j'}$. It follows that $\beta_j \not\equiv \beta_{j'} \pmod{p_i}$, so in particular the $\beta_j$'s are distinct modulo $m$.

    Now, fix distinct $\ell, a, b \in [k+1]$ and distinct $r, s \in [k]$.
    We consider a permutation $\sigma$ on $[k+1]$ such that $\sigma(k+1) = \ell, \sigma(r) = a, \sigma(s) = b$.
    We also define $\sigma'$ to be identical to $\sigma$ except we swap to get $\sigma'(r) = b, \sigma'(s) = a$.
    Then~\eqref{eq:funkyidentitysubin} implies that:
    \begin{align}
        \prod_{i = 1}^k (z_{i, \ell} - z_{i, \sigma(i)}) &= \frac{1}{\alpha_{\ell}} \prod_{i = 1}^k z_{i, \ell} \nonumber \\
        &= \prod_{i = 1}^k (z_{i, \ell} - z_{i, \sigma'(i)}) \nonumber \\
        \Rightarrow (z_{r, \ell} - z_{r, a})(z_{s, \ell} - z_{s, b}) &= (z_{r, \ell} - z_{r, b})(z_{s, \ell} - z_{s, a}). \label{eq:affineequality}
    \end{align}
    At this point, we apply Lemma~\ref{lem:affinerank1} to deduce that there exist $a'_1, \ldots, a'_k \in \F \backslash \{0\}$ and distinct $t_1, \ldots, t_{k+1} \in \F$ such that $z_{i, j} = z_{i, 1} + a'_i t_j$ for all $i \in [k], j \in [k+1]$.
    To construct our root of unity grid, we will take $a_i = a'_i/z_{i, 1}$ for all $i \in [k]$ and the same $t_1, \ldots, t_{k+1}$. We check the desired conditions:
    \begin{itemize}
        \item For each $i, j$, we have $1+a_it_j = (z_{i, 1} + a'_it_j)/z_{i, 1} = z_{i, j}/z_{i, 1} \in \mu_{p_i}(\F)$.
        \item We have $a'_i \neq 0 \Rightarrow a_i \neq 0$ for all $i$.
        \item We are given by Lemma~\ref{lem:affinerank1} that $t_1, \ldots, t_{k+1}$ are distinct.
    \end{itemize}
    The conclusion follows.
\end{proof}

\subsection{Approach 1: Rough Repunits Imply Root-of-Unity Grids}\label{sec:suffcond}\label{sec:rough-rug}
In this section, we show in that to construct a root-of-unity grid, it suffices to find a sequence of rough repunits,\footnote{Repunits in base $q$ are numbers of the form $\frac{q^d - 1}{q-1}$, not $q^d - 1$, but we do not make this distinction since a rough repunit (that is, a repunit with a large prime factor) will automatically give us a rough number of the form $q^d - 1$ and vice versa.} that is, numbers of the form $q^{d_1} - 1, \ldots, q^{d_k} - 1$ such that the largest prime factor of $q^{d_i} - 1$ is at least $(q^{d_i} - 1)^{1- 1/(3k)}$ (see Corollary~\ref{cor:roughness-rug}).
One can intuit the results of this section in terms of the probabilistic method, as we sketched in Section~\ref{sec:rughowto}.

\begin{lemma}[Sufficient conditions for a root-of-unity grid]\label{lemma:suffcondsharp}
    Let $k \ge 2$ be an integer, and let $q$ be a prime. Suppose there exist positive integers $d_1, \ldots, d_k$ and pairwise distinct primes $p_1, \ldots, p_k$ such that the following hold:
    \begin{itemize}
        \item $p_i \mid q^{d_i} - 1$ for all $i \in [k]$;
        \item $q^{\gcd(d_1, \ldots, d_k)} \geq k+1$; 
        \item for all $i \in [k]$, we have that
        \begin{align*}
            \frac{q^{d_i} - k}{h_i^k} > (k-1) \sqrt{q^{d_i}} (1-1/h_i^k) + 1,
        \end{align*}
        where $h_i := \frac{q^{d_i} - 1}{p_i}$. 
    \end{itemize}
    Then, there exists a $(\F_{q^D}, m)$-root-of-unity grid, where $D = \mathrm{lcm}(d_1, \ldots, d_k) $ and $ m = p_1 \cdots p_k$.
\end{lemma}


\begin{proof}
    Recall that a product $m = p_1 \cdots p_k$ of $k$ primes is defined to have a root-of-unity grid if there exists a field $\F$ containing a primitive $m$th root of unity $\gamma$, $a_1, \ldots, a_k \in \F^\times$, and distinct $t_1, \ldots, t_{k+1} \in \F$ such that $1 + a_i t_j \in \mu_{p_i}(\F)$ for all $i \in [k]$, $j \in [k+1]$. Note that without loss of generality, we can take $t_{k+1} = 0$.

    Suppose $m$ is such that there exist $q, d_1, \ldots, d_k, p_1, \ldots, p_k$ satisfying the preconditions of the statement.
    Fix arbitrary $T = \{t_1, \ldots, t_k\} \subseteq \F_{q^{\gcd(d_1, \ldots, d_k)}} \setminus \{0\}$. Such a $T$ can be chosen since $q^{\gcd(d_1, \ldots, d_k)} \geq k+1$.
    Applying Claim~\ref{clm:weil-counting} with $Q = q^{\gcd(d_1, \ldots, d_k)}$, $d = d_i/\gcd(d_1, \ldots, d_k)$, and $p = p_i$ for each $i \in [k]$, we know that there exists $a_i \in \F_{q^{d_i}}^\times$ such that $1 + a_i T \subseteq \mu_{p_i} (\F_{q^{d_i}})$.

    To finish, we take $\F = \F_{q^D}$, where $D = \mathrm{lcm}(d_1, \ldots d_k)$, to be a common extension of all the $\F_{q^{d_i}}$ fields and the $a_1, \ldots, a_k$ chosen above.
    $\F$ has a primitive $m$th root of unity since $p_i \mid q^{d_i} - 1 \Rightarrow p_i \mid q^D - 1 \forall i \Rightarrow m \mid q^D - 1 = |\F| - 1$.
    This yields
    the desired root-of-unity grid.
\end{proof}

\begin{claim}\label{clm:weil-counting}
    Let $k \ge 2$ be an integer, and let $Q \geq k+1$ be a prime power. Let $T = \{t_1, \ldots, t_k\} \subseteq \F_Q  \setminus \{0\}$ be an arbitrary set of $k$ distinct nonzero elements. Suppose $d, p$ are positive integers with $p$ prime such that $p \mid Q^d-1$ and
    \begin{align}
        \frac{Q^d - k}{h^k} &> (k-1)\sqrt{Q^d} (1-1/h^k) + 1,\label{eq:weilprecondition}
    \end{align}
    where $h := \frac{Q^d - 1}{p}$.
    Then, there exists some $a \in \F_{Q^d} \setminus \{0\}$ such that $1 + a t_j \in \mu_p(\F_{Q^d})$ for every $j \in [k]$.
\end{claim}


\begin{proof}
We begin by disposing of the case where $h = 1$.
In this case, $\mu_p(\F_{Q^d}) = \F_{Q^d}^\times$, so all we need to ensure is that $1+at_j \neq 0$ for all $j \in [k]$.
This rules out exactly $k$ values of $a$, so it suffices to check that $Q^d-1 > k$, which is immediate from~\eqref{eq:weilprecondition}.
So assume from now on that $h > 1$.

For fixed $p, d$, the existence of some $a \in \F_{Q^d} \backslash \{0\}$ such that $1 + a T \subseteq \mu_p(\F_{Q^d})$ is equivalent to the integer
\begin{align*}
    N_{p, d} := \# \{\lambda \in \F_{Q^d} : 1 + \lambda T \subseteq \mu_p(\F_{Q^d})\}
\end{align*}
being strictly greater than 1 (it will be at least 1 due to the trivial case of $\lambda = 0$). Let $h := \frac{Q^d - 1}{p}$, and let $\mu_p^\perp \le \widehat{\F_{Q^d}^\times}$ be the annihilator of $\mu_p$ as defined in Definition~\ref{def:annihilator}. 
In the following, we will use the convention specified in Definition~\ref{def:character} and also define $\chi(0) = 0$ for any character $\chi$ on $\F_{Q^d}^\times$.
Then $N_{p, d}$ can be expressed as the character sum:
\begin{align*}
    N_{p,d} &= \sum_{\lambda \in \F_{Q^d}} \prod_{j \in [k]} \mathbbm{1}\left[1 + \lambda t_j \in \mu_p(\F_{Q^d}) \right]\\
    &= \sum_{\lambda \in \F_{Q^d}} \prod_{j \in [k]} \left(\frac{1}{h} \sum_{\chi \in \mu_p^\perp} \chi(1 + \lambda t_j)\right) \text{ (Lemma~\ref{lemma:characterorthogonality})}\\
    &= \frac{1}{h^k} \sum_{\chi_1, \ldots, \chi_k \in \mu_p^\perp} \sum_{\lambda \in \F_{Q^d}} \left(\prod_{j \in [k]} \chi_j(1 + \lambda t_j)\right).
\end{align*}
Then, the all-trivial character tuple $(\chi_1, \ldots, \chi_k) = (1, \ldots, 1)$ contributes $\frac{Q^d - k}{h^k}$ to the sum, because there are exactly $k$ values of $\lambda$ such that $1+\lambda t_j = 0$ for some $j$.

Recall from Definition~\ref{def:annihilator} that $\mu_p^\perp$ is a cyclic group of order $h$.
So letting $\chi$ be a generator of $\mu_p^\perp$, we can rewrite $\chi_j = \chi^{r_j}$ for $r_j \in \{1, 2, \ldots, h\}$.
\begin{align*}
    N_{p,d} &= \frac{Q^d - k}{h^k} + \frac{1}{h^k} \sum_{(r_1, \ldots, r_k) \in \{1, \ldots, h\}^k \setminus \{h^k\}} \sum _{\lambda \in \F_{Q^d}} \left(\prod_{j \in [k]} \chi^{r_j}(1 + \lambda t_j)\right)\\
    &= \frac{Q^d - k}{h^k} + \frac{1}{h^k} \sum_{(r_1, \ldots, r_k) \in \{1, \ldots, h\}^k \setminus \{h^k\}} \sum _{\lambda \in \F_{Q^d}} \chi \left(\prod_{j \in [k]} (1 + \lambda t_j)^{r_j}\right)\\
    &\ge \frac{Q^d - k}{h^k} - \frac{1}{h^k} \sum_{(r_1, \ldots, r_k) \in \{1, \ldots, h\}^k \setminus \{h^k\}} \left|\sum _{\lambda \in \F_{Q^d}} \chi (f_{r_1, \ldots, r_k}(\lambda))\right|\\
\end{align*}
where $f_{r_1, \ldots, r_k}(\lambda) := \prod_{j \in [k]} (1 + \lambda t_j)^{r_j}$, and the above inequality follows from the triangle inequality.
We now apply the Weil bound for character sums (Lemma~\ref{lemma:weil}) to $f_{r_1,\ldots, r_k}$ and $\chi$.
We check the preconditions:
\begin{itemize}
    \item We are assuming that $h > 1$.
    \item $\chi$ is defined to be a generator of $\mu_p^\perp$, so it has order $h$.
    \item Note that the factors $(1+\lambda t_j)^{r_j}$ are pairwise coprime (viewed as elements of $\F_Q[\lambda]$) by distinctness of the $t_j$'s.
    Consequently, $f_{r_1, \ldots, r_k}$ can only be an $h$th power (up to scaling) if each of these factors is an $h$th power (up to scaling).
    For this to be true, we need $r_j \equiv 0 \pmod{h}$ for all $j$, which contradicts the fact that we are only applying the bound when $(r_1, \ldots, r_k) \neq (h, \ldots, h)$.
    \item Finally, it is clear from definition that $f_{r_1, \ldots, r_k}$ has $k$ distinct roots.
\end{itemize}
Therefore applying Lemma~\ref{lemma:weil} tells us that:
\begin{align*}
    N_{p,d} &\ge \frac{Q^d - k}{h^k} - \frac{1}{h^k} \sum_{(r_1, \ldots, r_k) \in \{1, \ldots, h\}^k \setminus h^k} (k-1) \sqrt{Q^d}\\
    &\ge \frac{Q^d - k}{h^k} - \frac{h^k - 1}{h^k} (k-1) \sqrt{Q^d}\\
    &= \frac{Q^d - k}{h^k} - (k-1)\sqrt{Q^d} (1-1/h^k),
\end{align*}
which by assumption is $> 1$, as desired.
\end{proof}




\begin{corollary}[Rough repunits suffice for root-of-unity grids]\label{cor:roughness-rug}
    Let $k \ge 2$ be an integer. Suppose there exists a prime $q$ and pairwise distinct positive integers $d_1, \ldots, d_k$ such that the following hold:
    \begin{itemize}
        \item $q \geq k+1$ and $d_i \ge 6$ for all $i \in [k]$, and 
        \item $q^{d_i}-1$ is $\epsilon$-rough, for $\epsilon = 1/(3k)$ for all $i \in [k]$.
    \end{itemize}
    Then, there exists a $(\F_{q^D}, m)$-root-of-unity grid where $D := \mathrm{lcm}(d_1, \ldots, d_k) $, $m := p_1 \cdots p_k$, and $p_i := P^+(q^{d_i}-1)$.
    
\end{corollary}


%

\begin{proof}
    By Lemma~\ref{lem:distinct-degrees-implies-distinct-primes}, we see that the primes $p_i$ must be pairwise distinct. Since $q \geq k + 1$, we automatically get that $q^{\gcd(d_1, \ldots, d_k)} \geq k+1$, so all that remains is to deduce the third bullet of the hypothesis in Lemma~\ref{lemma:suffcondsharp}.

    Since $p_i \ge (q^{d_i} - 1)^{1 - 1/(3k)}$ for all $i$, we get that
    \begin{align*}
        h_i := \frac{q^{d_i} - 1}{p_i} \le (q^{d_i} - 1)^{1/(3k)} \le q^{d_i/(3k)}.
    \end{align*}
    Now, writing $x = q^{d_i/6} \geq k$, we have:
    \begin{align*}
        \frac{q^{d_i} - k}{h_i^k} - (k-1)\sqrt{q^{d_i}} (1-1/h_i^k) &> \frac{q^{d_i} - k}{q^{d_i/3}} - (k-1)\sqrt{q^{d_i}} \\
        &= x^4 - k/x^2 - (k-1)x^3 \\
        &= x^3(x-(k-1)) - k/x^2 \\
        &\geq k^3 - k/x^2 \\
        &\geq k^3 - 1/k \\
        &> 1.
    \end{align*}
    Therefore, all the preconditions of Lemma~\ref{lemma:suffcondsharp} are met, and the statement follows.
\end{proof}




\subsection{Approach 2: Prime Field Norms Imply Root-of-Unity Grids}\label{sec:prime-norm-rug}

In this section, we present a more algebraically-structured approach to constructing root-of-unity grids.
This will not be used in any of our main results, but we show it for completeness and as an abstract explanation of why Efremenko's decoding polynomial for $m = 7 \cdot 73$ and $k = 2$ exists~\cite{Efr09} (which was not captured by the results of~\cite{Chee2011}).

\begin{lemma}\label{lem:normprimegivesrug}
    Let $Q$ be a prime power and $d > 1$ a positive integer such that $(Q^d-1)/(Q-1) = p$ is prime.
    Let $T = \{t_1, \ldots, t_k\} \subseteq \F_Q^\times$ be an arbitrary set of $k$ distinct nonzero elements.
    Assume there exists a polynomial $P(X) \in \F_Q[X]$ such that:
    \begin{itemize}
        \item $\deg P = d$ and $P(0) = 1$;
        \item $P(t) = 1$ for every $t \in T$; and
        \item the polynomial $f(X) := X^d \cdot P(-1/X) \in \F_Q[X]$ is irreducible.
    \end{itemize}
    Then there exists $a \in \F_{Q^d}^\times$ such that $1+aT \subseteq \mu_p(\F_{Q^d})$.
\end{lemma}
\begin{proof}
    Since $f \in \F_Q[X]$ is an irreducible polynomial over $\F_Q$ of degree $d$, we know that all its roots must appear in $\F_{Q^d}$.
    We also know that $f(0) \neq 0$.
    We will take $a \in \F_{Q^d}^\times$ to be any root of $f$ and argue that it satisfies the given conditions.
    Firstly, note that $f \in \F_Q[X]$ so for any automorphism $\sigma \in \Aut(\F_{Q^d}/\F_Q)$ we have $f(\sigma(a)) = 0$.
    So $a, a^Q, a^{Q^2}, \ldots, a^{Q^{d-1}}$ are all roots of $f$.
    Moreover, they must be distinct; if some two of these are equal then powering both sides implies that there is a $d' < d$ such that $a^{Q^{d'}} = a$, which means that $a \in \F_{Q^{d'}}$ which contradicts that fact that $a$'s minimal polynomial over $\F_Q$ has degree $d$.
    Also, $f(X)$ must be monic since $P(0) = 1$.
    Consequently, we must have:
    $$f(X) = \prod_{i = 0}^{d-1} (X-a^{Q^{i}}).$$
    Now, for any $t \in \F_Q^\times$, we have:
    \begin{align*}
        (1+at)^p &= (1+at)^{(Q^d-1)/(Q-1)} \\
        &= (1+at)^{1+Q+Q^2+\ldots+Q^{d-1}} \\
        &= \prod_{i = 0}^{d-1} (1+at)^{Q^i} \\
        &= \prod_{i = 0}^{d-1} (1+a^{Q^i}t) &\text{since $t \in \F_Q \Rightarrow t^{Q^i} = t$} \\
        &= (-t)^d \prod_{i = 0}^{d-1} (-1/t-a^{Q^i}) \\
        &= (-t)^d f(-1/t) \\
        &= (-t)^d \cdot (-1/t)^d \cdot P(t) \\
        &= P(t).
    \end{align*}
    So for any $t \in T$, we have $(1+at)^p = P(t) = 1 \Rightarrow 1+at \in \mu_p(\F_{Q^d})$, as desired.
\end{proof}

This approach also explains Efremenko’s example $m=7 \times 73$ over $\F_{2^9}$: taking $Q = 8$ and $d = 3$ so that $(8^3-1)/(8-1)=73$ is prime, we can obtain the $73$-row of the root-of-unity grid from the above argument over the extension $\F_{8^3}/\F_8$. The $7$-row is simpler to construct, by using the observation that $\F_8^\times = \mu_7(\F_{2^9})$. Thus, after choosing a suitable two-point set $T$ inside $\F_8$, the two rows together form a root-of-unity grid over $\F_{2^9}$, and hence give a 3-sparse decoding polynomial for $m=7\times 73$.

\begin{corollary}[Rederivation of Efremenko's decoding polynomial~\cite{Efr09}]
    Let $m = 7 \times 73$. Then there exists an $S^*_m$-decoding polynomial over $\F_{2^9}$ of sparsity 3.
\end{corollary}
\begin{proof}
    Let $p_1 = 7$ and $p_2 = 73$.
    By Lemma~\ref{lem:rugisgood}, it suffices to show that there is a root-of-unity grid over $\F_{2^9}$.
    Take $\theta \in \F_{2^3}$ such that $\theta^3+\theta+1 = 0$.
    For our root-of-unity grid, we will take $T = \{1, \theta, 0\}$.
    First, we will take $a_1 = \theta$. We have $1+a_1T = \{1+\theta, 1+\theta^2, 1\}$.
    These are all nonzero elements of $\F_8$, so this choice of $a_1$ works precisely because $\F_8^\times = \mu_7(\F_{2^9})$.

    To choose $a_2$, we will use Lemma~\ref{lem:normprimegivesrug} with $Q = 8$ and $d = 3$.
    We will take:
    \begin{align*}
        P(X) &= 1 + X(X-1)(X-\theta) \\
        &= 1 + X(X+1)(X+\theta) &\text{since we are in characteristic 2} \\
        \Rightarrow f(X) &= X^3 \left(1+(-1/X)(-1/X-1)(-1/X-\theta)\right) \\
        &= X^3+(X+1)(X\theta+1) \\
        &= X^3+\theta X^2 + (1+\theta)X + 1.
    \end{align*}
    By definition of $P$, we have $P(0) = P(1) = P(\theta) = 1$.
    A direct check verifies that $f(X)$ has no roots over $\F_8$ and is therefore irreducible.
    Thus Lemma~\ref{lem:normprimegivesrug} is applicable so there exists $a_2 \in \F_{2^9}^\times$ such that $1+a_2 \{1, \theta\} \subseteq \mu_{73}(\F_{2^9})$.
    Noting that $1 + a_2 \cdot 0 = 1 \in \mu_{73}(\F_{2^9})$, we are done.
\end{proof}

\subsection{Comparison with~\cite{Chee2011}}\label{sec:cheecomparison}

Here, we briefly explain why none of our results appear to capture the results by~\cite{Chee2011} in the $k = 2$ case.
Their theorem says that if $m=p_1p_2$ and $m+1=2^d$ for a prime $d$, then there is an $S^*_m$-decoding polynomial of sparsity $3$ over $\F_{2^d}$.
By the equivalence we have shown, this is the same as saying that there is a root-of-unity grid over $\F_{2^d}$ with respect to $m$.

Firstly the rough-repunit/Weil-bound approach in Section~\ref{sec:rough-rug} requires each $p_i$ to be a very large divisor of $q^{d_i}-1$ (at the very least, much larger than $q^{d_i/2}$).
In the setting of~\cite{Chee2011}, however, the two primes $p_1,p_2$ are the two prime factors of the same number $2^d-1$.
It is clearly not possible for both $p_1$ and $p_2$ to be larger than $2^{d/2}$, so the rough-repunit approach does not explain why this configuration gives rise to a root-of-unity grid.

The norm-kernel construction of Section~\ref{sec:prime-norm-rug} also does not recover the theorem of~\cite{Chee2011} in general.
We saw that it explains Efremenko's example $m=7\cdot 73$ because $73=(8^3-1)/(8-1)$ is prime (in addition to some other conditions).
By contrast, in the theorem of~\cite{Chee2011}, neither $p_1$ nor $p_2$ need arise as a prime of the form $(Q^e-1)/(Q-1)$ for a convenient intermediate field $\F_Q$.
Thus the norm-kernel argument gives a conceptual explanation of Efremenko's original example, but not of the full family constructed in~\cite{Chee2011}.

\subsection{Stitching Together Simultaneous Root-of-Unity Grids}\label{sec:stitch}
In some parameter regimes, in particular, when the number of servers is superconstant, the optimal thing to do will be to construct many simultaneous root-of-unity grids over the same characteristic and stitch them together into a sparse decoding polynomial.

\begin{lemma}[Stitching together many root-of-unity grids to get a sparse decoding polynomial]\label{lem:decoderfromconjecture}
    Let $t \ge 0, r > 0$ and $2 \le k_1 \le \ldots \le k_r$ be integers.
    Suppose there exists a prime $q$, integers $\{d_{i,j}\}_{i \in [r], j \in [k_i]} \cup \{d_i'\}_{i \in [t]}$,  and pairwise distinct primes $\{p_{i,j}\}_{i \in [r], j \in [k_i]} \cup \{p_i'\}_{i \in [t]}$ such that the following conditions hold:
    \begin{itemize}
        \item For every $i \in [r]$, there is a $(\F_{q^{D_i}}, m_i)$-root-of-unity grid, where
        \begin{align*}
            m_i &= \prod_{j \in [k_i]} p_{i,j} \text{, and }\\
            D_i &= \mathrm{lcm}(d_{i,1}, \ldots, d_{i, k_i}).
        \end{align*}
        \item For every $i \in [t]$, $p_i' | q^{d_i'} - 1$.
    \end{itemize}
    Then there exists an $S^*_m$-decoding polynomial over $\F_{q^D}$ with sparsity $\leq 2^t(k_1+1)\ldots (k_r+1)$, where
    \begin{align*}
        m &:= \left(\prod_{i=1}^t p'_i\right) \cdot \left(\prod_{i=1}^r m_i\right), \text{and}\\
        D &:= \mathrm{lcm}(d_1', \ldots, d_t', D_1, \ldots, D_r).
    \end{align*}
\end{lemma}
\begin{proof}
    We have all of the following:
    \begin{itemize}
        \item There is an $S^*_{p'_i}$-decoding polynomial over $\F_{q^{d'_i}}$ of sparsity 2 for all $i \in [t]$.
        This is by Claim~\ref{clm:trivialdec}.
        \item For every $i \in [r]$, there is an $S^*_{m_i}$-decoding polynomial over $\F_{q^{D_i}}$ of sparsity $k_i+1$. This is by Lemma~\ref{lem:rugisgood} applied to the given root-of-unity grids.
    \end{itemize}
    We now compose these decoding polynomials together using Theorem~\ref{thm:tensoring} to obtain the conclusion.
\end{proof}

We also note that when $r=0$, one can instantiate the above construction without requiring any conjectures.
This is exactly Corollary~\ref{cor:trivialdecodertensored}.

\section{Number-Theoretic Conjectures for Rough Repunits}\label{sec:hardconjecture}
In this section, we state and justify two number-theoretic conjectures that imply the existence of root-of-unity grids via the rough repunit-approach. The two conjectures are for two different regimes: a constant number of servers, for which a single root-of-unity grid suffices (Section~\ref{sec:constant-server-conjecture}) and superconstant number of servers, for which the optimal construction of the decoding polynomials (see Section~\ref{sec:gettopir} for the optimization in the superconstant server regime) will be to stitch together several simultaneous root-of-unity grids occurring over the same characteristic (Section~\ref{sec:superconstnat-server-conjecture}).

\subsection{Constant \# of Servers Regime}\label{sec:constant-server-conjecture}

We begin by recalling Conjecture~\ref{conjecture:onetuple-intro} that we rely on for the constant-server case, as well as our main theorem:

\begin{conjecture}[Generalized repunit rough numbers, restated from Conjecture~\ref{conjecture:onetuple-intro}]\label{conjecture:onetuple-main}
    For every positive integer $k > 1$, there exists a prime $q \geq k+1$ for which there exist $k$ distinct positive integers $d_1, \ldots, d_k \geq 6$ such that $q^{d_i}-1$ is $1/(3k)$-rough for all $i \in [k]$.
\end{conjecture}

\begin{theorem}[Sparse decoding polynomials from Conjecture~\ref{conjecture:onetuple-main} (restated from Theorem~\ref{thm:intromain})]\label{thm:intromain-restate}
    Assuming Conjecture~\ref{conjecture:onetuple-main}, for any constant $k$ there exists $m$ that is a product of $k$ distinct primes, a finite field $\F$ with $m \mid |\F| - 1$, and an $S^*_m$-decoding polynomial of sparsity $k+1$.
\end{theorem}

\begin{proof}
    If $k = 1$, this is immediate from Claim~\ref{clm:trivialdec}.
    So assume from now that $k > 1$.
    Assuming Conjecture~\ref{conjecture:onetuple-main}, there exist $q, d_1, \ldots, d_k$ satisfying the preconditions of Corollary~\ref{cor:roughness-rug}.
    The corollary gives us a finite field $\F$ and modulus $m$ that is a product of $k$ distinct primes such that there is an $(\F, m)$-root-of-unity grid.
    Then Lemma~\ref{lem:rugisgood} tells us that there exists an $S^*_m$-decoding polynomial over $\F$ of sparsity $k+1$, as desired.
\end{proof}
Now that we have proven Theorem~\ref{thm:intromain}, we turn our attention to justifying the plausibility of Conjecture~\ref{conjecture:onetuple-main}.

\paragraph{Justification for Conjecture~\ref{conjecture:onetuple-main}.}
As noted in the introduction, it would suffice to ensure that $d_i \geq 3k$ and $(q^{d_i}-1)/(q-1)$ is prime. The generalized repunit conjecture~\cite{Bourdelais2009, dubner1993generalized}, a generalization of the Mersenne prime conjecture, states that there exist infinitely many base-$q$ repunits for prime $q$. The conjecture we need is only weaker; we only need $(q^{d_i}-1)/(q-1)$ to be a ``rough'' number, or in other words, have a large prime divisor.

There are two different ways to justify that the above conjecture ought to be true:
\begin{itemize}
    \item \emph{Fix $q$ then pick the $d_i$'s.}
    As we have discussed, assuming the generalized repunit conjecture is enough to show that this approach will work.
    To justify why the generalized repunit conjecture might be true, since $d$ is an odd prime there is no ``obvious reason'' preventing $\frac{q^d-1}{q-1}$ from being prime.
    So assuming $\frac{q^d-1}{q-1}$ is as likely to be prime as a random integer of the same size, the probability that a particular $d$ works for us is $\sim \frac{1}{d \log q}$.
    Summing over all primes $d \geq 3k$, the number of good primes $d$ is heuristically:
    $$\approx \sum_{d \geq 3k\text{ prime}} \frac{1}{d \log q} = \frac{1}{\log q} \sum_{d \geq 3k\text{ prime}} \frac{1}{d},$$
    which diverges to $\infty$ by Mertens's theorem.

    \item \emph{Fix the $d_i$'s then pick $q$.}
    An alternative strategy would be to pick arbitrary distinct primes $d_1, \ldots, d_k \geq 3k$, then conjecture that a good value of $q$ should exist.
    Indeed, this is implied by Schinzel's hypothesis H~\cite{schinzel1958certaines,Bateman1962} applied to the $k+1$ polynomials $X, \frac{X^{d_1}-1}{X-1}, \ldots, \frac{X^{d_k}-1}{X-1}$.
    Informally, this hypothesis states that for any tuple of irreducible polynomials, they should simultaneously evaluate to primes infinitely often unless this claim can be refuted by looking at the polynomials' evaluations modulo some fixed prime $p$.

    We take this opportunity to justify the applicability of Schinzel's hypothesis H: first, note that since the $d_i$'s are prime these polynomials are all irreducible.
    Secondly, we need to rule out the existence of a prime $p$ that divides one of these polynomials for every evaluation point $x$.
    We can take $x$ to be a generator of the multiplicative group $\F_p^\times$, then we have $p \nmid x$ and for some $d_i$ we have $p \mid x^{d_i}-1 \Rightarrow p-1 \mid d_i$, which forces $p-1 \in \{1, d_i\}$. The latter case is impossible since $p, d_i$ are both primes with $d_i \geq 6$, so they cannot be consecutive.
    So this leaves us with the former case of $p = 2$.
    However, it is easy to see that when $x$ is odd all these polynomials are odd (since the $d_i$'s are odd), so this case cannot happen either.

    So if Schinzel's hypothesis H were true, it would follow that there are infinitely many positive integers $q$ such that $q, \frac{q^{d_1}-1}{q-1}, \ldots, \frac{q^{d_k}-1}{q-1}$ are all prime, as claimed.
\end{itemize}

\subsection{Superconstant \# of Servers Regime}\label{sec:superconstnat-server-conjecture}
To optimize the constructions of $S$-decoding polynomials for the superconstant server regime, it turns out to be optimal to stitch together many root-of-unity grids over the same characteristic. To carry out this optimization and estimate costs of the final PIR scheme, we need to estimate the parameters (such as $m, q$, field size) of the decoding polynomial obtained via this stitching. The following conjecture does this.

\begin{conjecture}\label{conjecture:superconstant}
    Let $t \ge 0, r > 0$ and $2 \le k_1 \le \ldots \le k_r$ be integers.
    Then there exists a prime $q$, integers $\{d_{i,j}\}_{i \in [r], j \in [k_i]} \cup \{d_i'\}_{i \in [t]}$,  and pairwise distinct primes $\{p_{i,j}\}_{i \in [r], j \in [k_i]} \cup \{p_i'\}_{i \in [t]}$ such that the following conditions hold:
    \begin{itemize}
        \item For every $i \in [r]$, there is a $(\F_{q^{D_i}}, m_i)$-root-of-unity grid, where
        \begin{align*}
            m_i &= \prod_{j \in [k_i]} p_{i,j} \text{, and }\\
            D_i &= \mathrm{lcm}(d_{i,1}, \ldots, d_{i, k_i}).
        \end{align*}
        \item For every $i \in [t]$, $p_i' | q^{d_i'} - 1$.
    \end{itemize}
    Moreover, we can upper bound the prime $q$, modulus $m := \left(\prod_{i = 1}^r \prod_{j=1}^{k_i} p_{i,j} \right) \cdot \left(\prod_{i \in [t]} p_i'\right)$ and size of the field $\F_{q^D}$, where $D := \mathrm{lcm}(D_1, \ldots, D_r, d_1', \ldots, d_t')$ as follows:
    \begin{align*}
    q &\le 2 k_r, \\
    m &\le \exp\left\{ O\left( t \log(t + K) + (\log k_r) \cdot \sum_{i=1}^r k_i^2 K_i \log K_i \right) \right\}, \text{ and}\\
    \log |\F_{q^D}| &\le  (\log k_r) \cdot  \left(\prod_{i=1}^r O(k_i K_i \log K_i)^{k_i}\right) \cdot\exp(O(t \log (t+K) )),
\end{align*}
where $K_i := \sum_{j=i}^r k_j$ and $K = K_1 = \sum_{j=1}^r k_j$.
\end{conjecture}

\begin{remark}
    We note that although we use the full generality of Conjecture~\ref{conjecture:superconstant} to find the optimal choice of decoding polynomials, our final PIR constructions require the conjecture to hold only for some choices of $t, r, k_1, \ldots, k_r$. In particular, we only ever consider cases where $t = 0$ and all the $k_1 = \ldots = k_r$.
\end{remark}

Our strategy for finding $r$ small root-of-unity grids, given $k_1, \ldots, k_r$ will be to first pick a characteristic $q$, and then search for \emph{prime} degrees $d_{i,j}$ such that $q^{d_{i,j}} - 1$ is $(1/(3k_i))$-rough, for all $i \in [r], j \in [k_i]$---we call such degrees \emph{usable}. We note that conversion of rough repunits $q^{d_{i,j}} - 1$ to root-of-unity grids does not require the $d_{i,j}$'s to be prime; in the following conjecture we insist on prime degrees because it is easier to reason about their density. Having picked enough usable degrees $d_{i,j}$, we set $p_{i,j} = P^+(q^{d_{i,j}} - 1)$, and the rest of the root-of-unity grid follows from these choices. Then, estimating the expected size of $m, q, D$ boils down to estimating the density of usable primes.

\begin{conjecture}[Density of Usable Degrees]\label{conj:usable-degrees}
    For every $\epsilon \in (0,1/6]$, prime $q$ and integer $\Delta \ge 6$, such that $q^{5\Delta/12} > \Delta^3 \log q$.
    \begin{align*}
        \Pr_{\substack{\text{prime } d
        \\ d \gets [\Delta, 2\Delta]
        }} [q^d - 1 \text{ is } \epsilon\text{-rough} ] \ge \log \frac{1}{ 1 - \epsilon} + O \left(\frac{1}{\Delta}\right).
    \end{align*}
\end{conjecture}




The purpose of the above conjecture is to give us many root-of-unity grids over the same characteristic that can be stitched together into one large decoding polynomial.
We capture this with the below lemma:



\begin{theorem}
    Assume Conjecture~\ref{conj:usable-degrees} holds. Then, Conjecture~\ref{conjecture:superconstant} holds, and therefore for every choice of $t \ge 0, r > 0$ and $2 \le k_1 \le \ldots \le k_r$, there exists an $S_m^*$-decoding polynomial over $\F_{q^D}$ with sparsity $\le 2^t (k_1 + 1) \cdots (k_r+1)$, where
    \begin{align*}
        q &\le 2 k_r, \\
    m &\le \exp\left\{ O\left( t \log(t + K) + (\log k_r) \cdot \sum_{i=1}^r k_i^2 K_i \log K_i \right) \right\}, \text{ and}\\
    \log |\F_{q^D}| &\le  (\log k_r) \cdot  \left(\prod_{i=1}^r O(k_i K_i \log K_i)^{k_i}\right) \cdot\exp(O(t \log (t+K) )),
    \end{align*}
    where $K_i := \sum_{j=i}^r k_j$ and $K = K_1 = \sum_{j=1}^r k_j$.
\end{theorem}
\begin{proof}
Fix any prime $q > k_r$. By Lemma~\ref{lem:bertrand}, there exists prime $k_r \le q \le 2 k_r$. Call prime $d \ge 6$ \emph{usable for level $i$} if $q^d - 1$ is $\epsilon_i$ rough, where $\epsilon_i = 1/(3k_i)$. For each $i$, we will pick $k_i$ usable primes $d_{i, 1}, \ldots, d_{i, k_i}$ such that $\{d_{i,j}\}_{i,j}$ are pairwise distinct. Therefore, by Corollary~\ref{cor:roughness-rug}, there exists a $(\F_{q^{D_i}}, m_i)$-root-of-unity grid, where $D = \mathrm{lcm}(d_{i, 1}, \ldots, d_{i, k_i})$ and $m_i = p_{i,1} \cdots p_{i, k_i}$, where $p_{i, j} = P^+(q^{d_{i,j}} - 1)$.

By Conjecture~\ref{conj:usable-degrees}, we have that for random prime $d \gets [\Delta_i, 2 \Delta_i]$, large enough\footnote{Note that the precondition for Conjecture~\ref{conj:usable-degrees} that $q^{5\Delta/12} \gg \Delta^3 \log q$ is satsified as long as $\Delta \ge \Omega(\log \log q)$.} $\Delta_i \ge 6, k_i, \log \log q$, we have that $d$ is usable for level $i$ with probability around
\begin{align*}
    \log \frac{1}{1-\epsilon_i} + O \left(\frac{1}{\Delta_i} \right) &= \epsilon_i + O(\epsilon_i^2) + O \left(\frac{1}{\Delta_i} \right)= \Omega\left(\frac{1}{k_i}\right) + O\left(\frac{1}{k_i^2}\right).
\end{align*}
To get $\ell$ usable primes greater than $\Delta_i$, we need to search over roughly $O(k_i \ell)$ primes greater than $\Delta_i$. To make sure that the interval $[\Delta_i, 2 \Delta_i]$ contains enough primes, we would like that $\pi(2 \Delta_i) - \pi(\Delta_i) \ge \Omega(k_i \ell)$ for all $i \in [r]$. By the prime number theorem,
\begin{align*}
    \pi(2 \Delta_i) - \pi(\Delta_i) &= \Theta \left( \frac{2 \Delta_i}{\log 2\Delta_i} - \frac{\Delta_i}{\log \Delta_i} \right),
\end{align*}
so taking $\Delta_i = O(k_i \ell \log (k_i \ell))$ suffices, and we expect to find $d_{i, 1}, \ldots, d_{i, k_i} \le O(k_i \ell \log (k_i \ell))$.

Since $\epsilon_1 \ge \ldots \ge \epsilon_r$, and $\epsilon$-roughness is a stronger condition for smaller $\epsilon$, we have every degree usable at level $i$ must be usable for all levels $i' \le i$. Therefore to ensure that all the $d_{i,j}$'s are distinct, let us start picking from $i = r$ and move downwards. We first pick $k_r$ distinct prime degrees to be usable at level $r$, so we expect $d_{r, 1} , \ldots , d_{r, k_r} \le O(k_r^2 \cdot \log (k_r^2))$. To pick distinct usable degrees at level $i$, we want to avoid all $K_{i+1} = K_i - k_i$ degrees picked so far, and pick $k_i$ additional degrees usable at level $i$, so we expect to find $d_{i, 1}, \ldots, d_{i, k_i} \le O(k_i K_i \cdot \log (k_i K_i))$. Since $\log(k_i K_i) \le 2 \log K_i$ for all $i \in [r]$, we get the slightly simpler bound
\begin{align*}
    d_{i, 1}, \ldots, d_{i, k_i} \le O(k_i K_i \log (K_i)).
\end{align*}
After having picked all the primes $p_{i,j}$ for $i \in [r], j \in [k_i]$, we look for $t$ additional small primes $p_i'$ for $i \in [t]$. Keep in mind that all the primes $\{p_{i,j}\}_{i \in [r],j \in [k_i]} \cup \{p_i'\}_{i \in [t]}$ must be pairwise distinct. To avoid picking the same prime twice, we will simply look for the smallest $t$ primes that do not already occur in the set $\{p_{i,j}\}_{i \in [r],j \in [k_i]} \cup \{q\}$ of size $K+1$. By the prime number theorem, we expect $p_t' \le O((t + K) \log (t + K))$. Now all that remains to do is to find $d_i'$, $i \in [t]$ such that $p_i' | q^{d_i'} - 1$ for all $i \in [t]$. This can be done by picking $d_i' = \ord_{p_i'}(q) \le p_i'$ for all $i \in [t]$.

Finally, we can bound the modulus $m$:
\begin{align*}
    m &= \left(\prod_{i=1}^t p_i'\right) \cdot \left(\prod_{i=1}^r \prod_{i=1}^{k_i} p_{i,j}\right)\\
    &= O((t+K) \log (t + K))^t \cdot q^{\sum_{i=1}^r k_i \cdot O(k_i K_i \log (K_i)) }\\
    &= \exp\left\{ O \left(t \log (t + K) + (\log q) \cdot \sum_{i=1}^r k_i^2 K_i \log (K_i) \right)\right\}\\
    &\le \exp\left\{ O \left(t \log (t + K) + (\log k_r) \cdot \sum_{i=1}^r k_i^2 K_i \log (K_i) \right)\right\}
\end{align*}
and the logarithm of the field size
\begin{align*}
    \log |\F_{q^D}| &= (\log q) \cdot D\\
    &= (\log q) \cdot {\mathrm{lcm}(D_1, \ldots D_r, d_1' ,\ldots, d_t')}\\
    &\le (\log q) \cdot  \left(\prod_{i=1}^t d_i'\right) \cdot \left(\prod_{i=1}^r \prod_{i=1}^{k_i} d_{i,j}\right)\\
    &\le O(\log k_r) \cdot O((t+K) \log (t + K))^t \cdot \left(\prod_{i=1}^r O(k_i K_i \log(K_i))^{k_i}\right)\\
    &\le O(\log k_r) \cdot \exp(O (t \log (t+K))) \cdot \left(\prod_{i=1}^r O(k_i K_i \log(K_i))^{k_i}\right).
\end{align*}
\end{proof}

\subsection{Heuristics Justifying Conjecture~\ref{conj:usable-degrees}}
We present two heuristics (Heuristics~\ref{heur:dickmanqd-1} and \ref{heuristic:q-random-in-p}) that each individually support Conjecture~\ref{conj:usable-degrees}.
First, we introduce some useful notation.
\begin{definition}\label{defn:rough-notation}
    Fix $\epsilon \in (0, 1/6]$. Throughout this section, when $\epsilon$ is clear from the context, we use the following notation:
    \begin{align*}
        \tau_{q, d} &:= q^{d (1-\epsilon)} \ge (q^d - 1)^{(1-\epsilon)}, \text{ and}\\
        N_{q, d} &:= \begin{cases}
            \frac{q^d - 1}{ d (q-1)} &\text{if } d | q-1,\\
            \frac{q^d - 1}{q-1} &\text{otherwise.}
        \end{cases}
    \end{align*}
\end{definition}

\begin{claim}\label{clm:rough-N-tau}
    For any $\epsilon \in (0, 1/6]$ and primes $q, d$, we have that $P^+(q^d-1)\ge \tau_{q,d}$ if and only if $P^+(N_{q,d})\ge \tau_{q,d}$.
\end{claim}

\begin{proof}
    First note that \(N_{q,d}\) is an integer: $q-1 | q^d - 1$ for all $d \ge 1$, and $d | q-1$, then $q = 1 \bmod d$
    \begin{align*}
        \frac{q^d-1}{q-1}& =1+q+\cdots+q^{d-1} = 0 \bmod d.
    \end{align*}
    We will show that every prime divisor of $(q^d - 1)/N_{q,d}$ is strictly smaller than $\tau_{q,d}$. Since $\epsilon \le 1/6$, $q, d \ge 2$, we have that $\tau_{q, d} = q^{d (1-\epsilon)} \ge q^{5d/6}> q-1$.
    Also $\tau_{q,d} \ge q^{5d/6} \ge 2^{5d/6} > d$ for all $d \ge 2$. It then follows that $q^d-1$ has a prime factor at least $\tau_{q,d}$ if and only if $N_{q, d}$ does.
\end{proof}

\subsubsection{Naive Dickman Random Factorization Heuristics}
Dickman characterized the probability that a random number drawn from an interval is smooth (Lemma~\ref{lem:dickman}). As a first heuristic, we model the repunit integers $q^d-1$ as random integers of size $q^d$, and apply the Dickman heuristic to these numbers.

\begin{heuristic}\label{heur:dickmanqd-1}[Dickman heuristic for $q^d-1$]
For every prime $q$, $\alpha > 0$, and integer $\Delta$,
    \begin{align*}
        \Pr_{\substack{\text{prime } d\\ d \gets [\Delta, 2 \Delta]}}[P^+(q^d-1) \le q^{\alpha d}] = \rho\left(\frac{1}{\alpha}\right) + O\left(\frac{1}{\Delta}\right).
    \end{align*}
\end{heuristic}

\begin{lemma}
    Under Heuristic~\ref{heur:dickmanqd-1}, Conjecture~\ref{conj:usable-degrees} holds.
\end{lemma}
\begin{proof}
Since $\tau_{q, d} := q^{(1-\epsilon) d} \ge (q^d-1)^{1-\epsilon}$, we have that that $x \ge \tau_{q, d} \implies x \ge (q^d-1)^{1-\epsilon}$, so applying $x = q^d-1$ this means that
    \begin{align*}
        \Pr_{d} [q^d - 1 \text{ is } \epsilon\text{-rough} ] &\ge \Pr_{d} [ P^+(q^d-1) \ge  \tau_{q, d}]\\
        &= \Pr_d[P^+(q^d-1) \ge q^{d(1-\epsilon)}] &\text{Notation~\ref{defn:rough-notation}}\\
        &= 1 - \rho\left(\frac{1}{1-\epsilon}\right) + \left(\frac{1}{\Delta}\right) &\text{Heuristic~\ref{heur:dickmanqd-1}}\\
        &= \log \frac{1}{1-\epsilon} + \left(\frac{1}{\Delta}\right) &\text{Lemma~\ref{lem:dickman}, since } 5/6 \le \alpha \le 1.
    \end{align*}
\end{proof}

\paragraph{Justification of Heuristic~\ref{heur:dickmanqd-1}.}
For a random (not necessarily prime) integer $d \gets [\Delta, 2 \Delta]$, modeling $q^d-1$ as a random integer of size $q^d$ ignores the cyclotomic structure of such numbers: 
\begin{align*}
    q^d - 1 &= \prod_{d' | d} \Phi_{d'}(q),
\end{align*}
where $\Phi_{d'}(q)$ is the value of the $d'$th cyclotomic polynomial, evaluated at $q$ (see Fact~\ref{cyclotomic-roots-of-unity}). For example, if $d$ is even, every prime factor of $q^d - 1$ must divide at least one of $q^{d/2} - 1, q^{d/2} + 1$, ruling out any hope for $q^d-1$ to be rough. In general, for non-prime $d$, and in the extreme case where $d$ has a lot of (small) prime factors, one would expect to run into similar issues. This is why we assume $d$ is prime in our heuristic; Claim~\ref{clm:rough-N-tau} shows that we are then effectively treating $N_{q,d}$ as a random number that satisfies the Dickman heuristic.

\subsubsection{Heuristics Accounting for the Cyclotomic Structure of $q^d - 1$}
In order to sidestep the cyclotomic structure of $q^d-1$ and use the Dickman heuristic for random integers, the previous section and Heuristic~\ref{heur:dickmanqd-1} restricted attention to prime $d$, effectively treating $N_{q, d}$ as a random integer that satisfies the Dickman heuristic, but this still ignores the fact that $N_{q,d}$ itself has structure that could make it look very different from a random number. 

Concretely, by Lemma~\ref{lem:rough-order} and Claim~\ref{clm:rough-N-tau}, we know that for rough $q^d-1$, $p = P^+(N_{q, d})$ if and only if $\ord_p(q) = d$, which implies that $d | p-1$. In other words, $p = 1 \bmod d$. Therefore, the prime divisors of $N_{q,d}$ come from a sparse arithmetic progression of primes of density $1/ \phi(d)$. In this section, instead of heuristically treating $N_{q, d}$ as a random number, we use heuristically estimate the density of large primes $p$ such that $p | N_{q, d}$. 



\begin{heuristic}\label{heuristic:q-random-in-p}
    Let $\epsilon \in (0,1/6]$ and let $q$ be a prime and let $\Delta$ be an integer 
    \begin{align*}
        \E_{d}  \sum_{\substack{\text{primes}\\p \in [\tau_{q, d}, N_{q, d}]}} \mathbbm{1}[p | N_{q, d}] =
        \E_{d} \Biggl[ \sum_{\substack{\text{primes}\\p \in [\tau_{q, d}, N_{q, d}] \\ p = 1 \bmod d}  } \frac{d-1}{p-1} \Biggr]+  O \left(\frac{1}{\Delta}\right),
    \end{align*}
    where the expectation is taken over random prime $d \gets [\Delta, 2 \Delta]$.
\end{heuristic}
\paragraph{Justification for Heuristic~\ref{heuristic:q-random-in-p}.}
For $p \ge \tau_{q, d}$, we know from the proof of Lemma~\ref{lem:distinct-degrees-implies-distinct-primes} that $p| N_{q, d}$ if and only if $\ord_p(q) = d$. This can happen only if $d | p-1$, i.e., $p = 1 \bmod d$. Since $\F^\times_p$ is cyclic, the number of elements if order $d$ is $\phi(d) = d-1$. If $q$ was sampled uniformly at random from $\F_p^\times$, conditioned on $p = 1 \bmod d$, the probability that $\ord_p(q) = d$ is $(d-1)/(p-1)$. Summing this density and averaging over $d$ yields the expression in the heuristic. In our setting, $q$ is a fixed prime. Heuristic~\ref{heuristic:q-random-in-p} asserts that this quantity nevertheless gives a good estimate in the fixed-$q$ case. 


\begin{lemma}\label{lem:usable-degrees}
    Assume that the GRH for all $L$-functions. 
    Under Heuristic~\ref{heuristic:q-random-in-p}, Conjecture~\ref{conj:usable-degrees} holds.
\end{lemma}
\begin{proof}
    Throughout this proof, we are sampling $d$ to be a random prime in the interval $[\Delta, 2 \Delta]$.
    First, we claim that $q^d - 1$ is $\epsilon$-rough if and only $P^+(N_{q, d}) \ge \tau_{q,d}$.

    Since $\tau_{q, d} \ge (q^d -1)^{1-\epsilon}$, we have that $P^+(q^d-1) \ge  \tau_{q, d}$ implies that $q^d-1$ is $\epsilon$-rough.

    \begin{align*}
        \Pr_{d} [q^d - 1 \text{ is } \epsilon\text{-rough} ] &\ge \Pr_{d} [ P^+(q^d-1) \ge  \tau_{q, d}]\\
        &= \Pr_{d} [ P^+(N_{q, d}) \ge  \tau_{q, d}] &\text{Claim~\ref{clm:rough-N-tau}}\\
        &= \E_d \# \{\text{prime } p | N_{q, d} : p \ge \tau_{q, d}\},
    \end{align*}
    where the last equality holds because $\# \{\text{prime } p | N_{q, d} : p \ge \tau_{q, d}\} \le 1$. To see, this, suppose for contradiction there were two distinct primes $p_1, p_2| N_{q, d}$ such that $ p_1, p_2 \ge \tau_{q, d}$, we would have that
    \begin{align*}
        q^d \ge N_{q, d} \ge p_1 p_2 \ge \tau_{q, d}^2 = q^{2d (1-\epsilon)} \ge q^{5d/3},
    \end{align*}
    a contradiction.
    We can rewrite $\# \{\text{prime } p | N_{q, d} : p \ge \tau_{q, d}\}$ as a sum of an indicator over primes $p$ in the interval
    \begin{align*}
        \Pr_{d} [q^d - 1 \text{ is } \epsilon\text{-rough} ] 
        &\ge \E_d \#\{\text{prime } p | N_{q, d} : p \ge \tau_{q, d}\}\\
        &= \E_{d}  \sum_{\substack{\text{primes}\\p \in [\tau_{q, d}, N_{q, d}]}} \mathbbm{1}[p | N_{q, d}]\\
        %
        %
        &= \E_{d} \sum_{\substack{\text{primes}\\p \in [\tau_{q, d}, N_{q, d}] \\ p = 1 \bmod d}  } \frac{d-1}{p-1} + O\left(\frac{1}{\Delta}\right) &\text{Heuristic~\ref{heuristic:q-random-in-p}}\\
        &= \E_{d} \sum_{\substack{\text{primes}\\p \in [\tau_{q, d}, N_{q, d}] \\ p = 1 \bmod d}  } \frac{d-1}{p} \left(1+ \frac{1}{(p-1)}\right) + O\left(\frac{1}{\Delta}\right)\\
        &= \E_{d} \sum_{\substack{\text{primes}\\p \in [\tau_{q, d}, N_{q, d}] \\ p = 1 \bmod d}  } \frac{d-1}{p} \left(1+ O\left(\frac{1}{\tau_{q, d}}\right)\right) + O\left(\frac{1}{\Delta}\right).\\
    \end{align*}
By Lemma~\ref{lem:mertens}, we can lower bound this quantity, and get that
    \begin{align*}
        \Pr_{d} &[q^d - 1 \text{ is } \epsilon\text{-rough} ] \\
        &\ge  \E_d \left[ (d-1) \cdot \left(\frac{1}{\phi(d)} (\log \log N_{q, d} - \log \log \tau_{q, d}) + O \left(\frac{\log \tau_{q, d}}{\sqrt{\tau_{q, d}}} \right)\right) \cdot \left(1+ O\left(\frac{1}{\tau_{q, d}}\right) \right)\right] + O\left(\frac{1}{\Delta}\right) \\
        &=  \E_d \left[ \log \frac{\log N_{q,d}}{\log \tau_{q, d}}  + O \left(\frac{d \log \tau_{q, d}}{\sqrt{ \tau_{q, d}}} + \frac{d}{\tau_{q, d}}\right) \right] + O\left(\frac{1}{\Delta}\right) \\
        &= \log \frac{1}{1 - \epsilon} + O\left(\frac{1}{\Delta}\right) + O \left(\frac{d \log \tau_{q, d}}{\sqrt{ \tau_{q, d}}}\right) + O\left(\frac{1}{\Delta}\right) ,
    \end{align*}
    where the second equality holds because $d$ is an odd prime, so $\phi(d) = d-1$, and the third equality holds because of the following calculation: 
    \begin{itemize}
        \item If $d \nmid q-1$,
    \begin{align*}
        \frac{\log N_{q,d}}{\log \tau_{q,d}} 
        &=\frac{\log \frac{q^d (1 - 1/q^d)}{q-1}}{\log q^{d(1-\epsilon)}}
        =\frac{d \log q + \log (1 - 1/q^d) - \log(q-1)}{d(1-\epsilon) \log q}
        \ge \frac{d \log q + O(1) + O(\log q)}{d(1-\epsilon) \log q}\\
        &= \frac{1}{1-\epsilon} \left(1 + O\left(\frac{1}{d}\right)\right)
    \end{align*}
    \item If $d | q-1$, we have that $ d \le q-1$, and so $\log d \le \log q$, and an almost identical calculation as above shows that
    \begin{align*}
        \frac{\log N_{q,d}}{\log \tau_{q,d}} \ge \frac{1}{1-\epsilon} \left(1 + O\left(\frac{1}{d}\right)\right).
    \end{align*}
    \end{itemize}
    In both cases, since $\log N_{q, d} / \log \tau_{q,d} = \frac{1}{1-\epsilon} \left(1 + O\left(\frac{1}{d}\right)\right)$, we have that 
    \begin{align*}
        \log \frac{\log N_{q,d}}{\log \tau_{q,d}} = \log \frac{1}{1-\epsilon} + \log \left(1 + O\left(\frac{1}{d}\right)\right) = \log \frac{1}{1-\epsilon} + O\left(\frac{1}{d}\right) = \log \frac{1}{1-\epsilon} + O\left(\frac{1}{\Delta}\right).
    \end{align*}
    Finally, observing that $2 d (\log \tau_{q, d})/  \sqrt{\tau_{q, d}} \le O(d (\log \tau_{q, d}) q^{-5d/12}) \ll O(1/\Delta)$ whenever $q^{5\Delta/12} > \Delta^3 \log q$ finishes the proof.
\end{proof}

\section{PIR from Our Decoding Polynomials}\label{sec:gettopir}

We begin by quickly completing the proof of Corollary~\ref{cor:asymptoticpir} (PIR in the constant-server regime):

\begin{proof}[Proof of Corollary~\ref{cor:asymptoticpir}]
    Let $k = s-1$.
    By Theorem~\ref{thm:intromain}, there exists $m$ that is a product of $k$ distinct primes and a finite field $\F$ with $m \mid |\F|-1$ such that there is an $S^*_m$-decoding polynomial over $\F$ of sparsity $k+1 = s$.
    The conclusion is now immediate by plugging into Theorem~\ref{thm:gksinformal}.
\end{proof}
Now we turn to the case of superconstant $s$.
Here, Theorem~\ref{thm:gksinformal} does not apply out of the box so we will use the following workflow to compute bounds on the PIR communication complexity:
\begin{itemize}
    \item We invoke our root-of-unity grid framework to obtain $q, m, \F$ with bounds on these quantities such that there exists an $S^*_m$-decoding polynomial over $\F$ of sparsity $\leq s$.
    
    This part splits into three regimes (chosen according to what will give the best end guarantee for the communication complexity):
    \begin{itemize}
        \item Section~\ref{sec:smallserver}: when $s \leq O((\log \log n)^{1/3}/(\log \log \log n)^{2/3})$, we will use Conjecture~\ref{conjecture:superconstant} and Lemma~\ref{lem:decoderfromconjecture} with $t = 0$, $r = 1$, and $k_1 = s-1$.
        In other words, we build our decoding polynomial from just one root-of-unity grid.
        This is intuitive as this is what we do in the constant-server case.

        \item Section~\ref{sec:medserver}: when $s$ ranges from $\Omega((\log \log n)^{1/3}/(\log \log \log n)^{2/3})$ to $\exp(o(\sqrt{\log \log n/\log \log \log n}))$, we will use multiple root-of-unity grids of the same size $k_i$, but this size will be smaller than $k$.
        (This still goes via Conjecture~\ref{conjecture:superconstant} and Lemma~\ref{lem:decoderfromconjecture}.)

        \item Section~\ref{sec:manyserver}: finally, when $s \geq \exp(\Omega(\sqrt{\log \log n/\log \log \log n}))$, our techniques combined with Conjecture~\ref{conjecture:superconstant} do not provide any advantage over using the standard construction of decoding polynomials.
        For completeness, we have restated the relevant decoding polynomial result in Corollary~\ref{cor:trivialdecodertensored} and derived the resulting communication guarantee in Theorem~\ref{thm:manyserver}.
    \end{itemize}
    
    \item Once we have bounds on $q, m, |\F|$, we can feed these into Corollary~\ref{cor:mvparams}~\cite{Barrington1994,DBLP:journals/combinatorica/Grolmusz00,DGY11} to bound the dimension $d$ of a matching vector family modulo $qm$ of size $\geq n$.
    \item Finally, we can plug these into Theorem~\ref{thm:gks} to obtain a bound on the PIR communication complexity $\mathsf{CC}$.
\end{itemize}

\subsection{Small Superconstant-Server Regime}\label{sec:smallserver}

\begin{theorem}[PIR in the few-server regime]\label{thm:fewserver}
    Assume Conjecture~\ref{conjecture:superconstant}, and suppose $n, s$ are such that $s = \omega(1)$ and:
    $$s \leq O((\log \log n)^{1/3}/(\log \log \log n)^{2/3}).$$
    Then there exists an $s$-server PIR for a database of size $n$ achieving total communication:
    \begin{align*}
        \mathsf{CC} &= \exp\left(\exp\left(O\left(\frac{\log \log n}{s}\right)\right)\right).
    \end{align*}
\end{theorem}
\begin{proof}
    We begin by setting parameters and calculating the resulting bounds on $m, |\F|$. We invoke Lemma~\ref{lem:decoderfromconjecture} with $t = 0$, $r = 1$, and $k_1 = s-1$.
    The resulting decoding polynomial has sparsity $\leq 2^t (k_1+1) = s$.
    The bounds from the lemma tell us that:
    \begin{align*}
        q &\leq O(s), \\
        m &\leq \exp\left(O\left((\log q) \cdot (s-1)^3 \log (s-1)\right)\right) \\
        &= \exp(O(s^3 (\log s)^2)),\text{ and} \\
        \log |\F| &\leq (\log q) \cdot O((s-1)^2 \log(s-1))^{s-1} \\
        &\leq O(1)^s \cdot s^{2s} \cdot (\log s)^s.
    \end{align*}
    In Corollary~\ref{cor:mvparams}, we will take the modulus to be $qm \leq \exp(O(s^3(\log s)^2))$ and the number of primes $\ell$ is exactly $s$ (accounting for the extra prime we get from $q$ as well).
    \begin{claim}\label{clm:smalleta}
        We have:
        $$\exp({\Omega(s^2)}) \leq \exp((s-1)\log \log n - O(s^3(\log s)^2)) \leq \eta \leq \exp(s \log \log n).$$
    \end{claim}
    \begin{proof}
        We have:
        \begin{align*}
            \eta &= \frac{(\log n)^{\ell-1}}{qm} \\
            &\geq \frac{(\log n)^{s-1}}{\exp(O(s^3(\log s)^2))} \\
            &= \exp\left((s-1)\log \log n - O(s^3(\log s)^2)\right),
        \end{align*}
        and
        $\eta \leq (\log n)^{s-1} \leq \exp(s \log \log n)$.
        The leftmost inequality follows from the fact that $s \log \log n \geq \tilde{\Omega}(s^4)$.
    \end{proof}

    \begin{claim}\label{clm:smalld}
        We have:
        $$d \leq \exp\left(\exp\left(O\left(\frac{\log \log n}{s}\right)\right)\right).$$
    \end{claim}
    \begin{proof}
        By Claim~\ref{clm:smalleta}, we will use the second branch of Corollary~\ref{cor:mvparams} and compute:
        \begin{align*}
            d &= \exp\left(O\left(\log n \cdot \frac{(\log \eta)^{1-1/\ell}}{\ell \eta^{1/\ell}}\right)\right) \\
            &\leq \exp\left(O\left(\log n \cdot \frac{(s \log \log n)^{1-1/s}}{\ell \eta^{1/\ell}}\right)\right) \\
            &\leq \exp\left(O\left(\log n \cdot \frac{(s \log \log n)^{1-1/s}}{s \cdot \exp((1-1/s) \log \log n - O(s^2(\log s)^2))}\right)\right) \\
            &= \exp\left(2^{O(s^2(\log s)^2)} \cdot (\log n)^{1/s} (\log \log n)^{1-1/s}\right) \\
            &= \exp\left(\exp\left(O(s^2(\log s)^2) + \frac{\log \log n}{s} + \left(1-\frac{1}{s}\right)\log \log \log n\right)\right) \\
            &= \exp\left(\exp\left(O\left(\frac{\log \log n}{s}\right)\right)\right).
        \end{align*}
        In the final step, we are using the fact that $s \leq O((\log \log n)^{1/3}/(\log \log \log n)^{2/3})$ to absorb the $O(s^2(\log s)^2)$ term.
    \end{proof}
    \noindent
    Now to complete the proof and bound $\mathsf{CC}$, by Theorem~\ref{thm:gks} it remains to verify that:
    $$s \log |\F| \leq \exp\left(\exp\left(O\left(\frac{\log \log n}{s}\right)\right)\right).$$
    From the earlier computations, we have:
    \begin{align*}
        s \log|\F| &\leq O(1)^s \cdot s^{2s} \cdot (\log s)^s \\
        &\leq \exp(O(s \log s)) \\
        &= \exp(\exp(O(\log s))),
    \end{align*}
    and now the conclusion follows by noting that $s \log s \leq \tilde{O}((\log \log n)^{1/3}) \leq O(\log \log n)$.
\end{proof}

\subsection{Intermediate-Server Regime}\label{sec:medserver}

\begin{theorem}[PIR in the intermediate-server regime]\label{thm:intserver}
    Assume Conjecture~\ref{conjecture:superconstant}, and suppose $s$ is such that:
    $$\Omega\left(\frac{(\log \log n)^{1/3}}{(\log \log \log n)^{2/3}}\right) \leq s \leq \exp\left(o\left(\sqrt{\frac{\log \log n}{\log \log \log n}}\right)\right).$$
    Then there exists an $s$-server PIR for a database of size $n$ achieving total communication:
    $$\exp\left(\exp\left(O\left(\frac{(\log \log n)^{2/3} (\log \log \log n)^{1/3}}{(\log s)^{1/3}} \cdot \left(\log\left(\frac{\log \log n}{(\log \log \log n) \cdot (\log s)^2}\right)\right)^{2/3}\right)\right)\right).$$
\end{theorem}
\begin{proof}
    We begin by setting parameters and calculating the resulting bounds on $m, |\F|$.
    Define
    $$x = \frac{\log \log n}{(\log \log \log n) \cdot (\log s)^2} = \omega(1),$$
    and $a = \Theta((x \log x)^{1/3})$.   
    We will invoke Lemma~\ref{lem:decoderfromconjecture} with $t = 0$, $r = \lfloor \frac{\log s}{\log (a+1)} \rfloor$, and $k_1 = \ldots = k_r = a$.
    The resulting decoding polynomial has sparsity $2^t \prod_{i = 1}^r (k_i+1) = (a+1)^r \leq s$.
    As an additional sanity check, we verify that $a+1 \leq s$; indeed, we have:
    \begin{align*}
        a+1 &\leq s \\
        \Leftarrow (x \log x)^{1/3} &\leq O(s) \\
        \Leftrightarrow x &\leq O(s^3/\log s) \\
        \Leftrightarrow \frac{\log \log n}{\log \log \log n} &\leq O(s^3 \log s) \\
        \Leftrightarrow \Omega\left(\frac{(\log \log n)^{1/3}}{(\log \log \log n)^{2/3}}\right) &\leq s,
    \end{align*}
    which we are assuming to be the case.
    Continuing on, we have $\ell = ar+1$ and $K = ar$.
    For any $i \in [r]$, we have $K_i = a(r-i+1)$.
    \begin{claim}
        We have:
        \begin{align*}
            q &\leq 2a\\
            m &\leq \exp(O(a^3r^2\log a \log \log \log n)) \\
            \log|\F| &\leq \exp(O(ar\log \log \log n)).
        \end{align*}
    \end{claim}
    \begin{proof}
        By the bounds stated in Lemma~\ref{lem:decoderfromconjecture}, we have:
    \begin{align*}
        q &\le 2a, \\
        m &\le \exp\left( O\left((\log q) \cdot \sum_{i=1}^r k_i^2 K_i \log K_i \right) \right) \\
        &= \exp\left( O\left((\log q) \cdot \sum_{i=1}^r a^2 \cdot a(r-i+1) \cdot (\log (a(r-i+1)))\right) \right) \\
        &= \exp\left(O\left(a^3 \log a \cdot \left(\sum_{i = 1}^r (r-i+1) \cdot (\log(a(r-i+1)))\right)\right)\right) \\
        &= \exp\left(O\left(a^3 \log a \cdot \left(r^2 \log a + r^2 \log r\right)\right)\right), \text{ and}\\
        \log |\F| &\le  (\log q) \cdot  \left(\prod_{i=1}^r O(k_i K_i \log K_i)^{k_i}\right) \\
        &\leq (\log q) \cdot \prod_{i = 1}^r O(a^2(r-i+1) \log(a(r-i+1)))^a \\
        &= (\log q) \cdot \exp\left(O\left(a \sum_{i = 1}^r \left(O(1) + \log a + \log(r-i+1) + \log \log (a(r-i+1))\right)\right)\right) \\
        &= \exp\left(O\left(ar(\log a+\log r)\right)\right).
    \end{align*}
    To complete the proof of the claim, it suffices to show that $\log a + \log r \leq O(\log \log \log n)$.
    Indeed, we have:
    \begin{align*}
        \log a &= \Theta(\log x) \\
        &\leq O(\log \log \log n),\text{ and}\\
        \log r &\leq \log \log s \\
        &\leq O(\log \log \log n).
    \end{align*}
    \end{proof}

    \begin{claim}\label{claim:medeta}
        We have:
        $$\exp({\Omega(\ell^2})) \leq \exp(ar\log \log n - O(a^3r^2\log a \cdot \log \log \log n)) \leq \eta \leq \exp(ar\log \log n),$$
        and $ar \leq o(\log \log n)$.
    \end{claim}
    \begin{proof}
        First we lower bound $\eta$:
        \begin{align*}
            \eta &= \frac{(\log n)^{\ell-1}}{qm} \\
            &\geq \frac{(\log n)^{ar}}{\exp(O(a^3r^2 \log a \cdot \log \log \log n))} \\
            &= \exp(ar\log \log n - O(a^3r^2\log a \cdot \log \log \log n)).
        \end{align*}
        The upper bound on $\eta$ is straightforward: $\eta \leq (\log n)^{\ell-1} = \exp(ar\log \log n)$.
        To establish the leftmost inequality, note that $\ell^2 = a^2r^2$, so it suffices to show that:
        \begin{align*}
            ar\log \log n &\geq \omega(a^3r^2 \log a(\log a + \log r)) \\
            \Leftrightarrow \log \log n &\geq \omega(a^2r \log a(\log a + \log r)) \\
            \Leftrightarrow \log \log n &\geq \omega(a^2 \log s(\log a + \log r)).
        \end{align*}
        To check this, note that:
        \begin{align*}
            a^2 \log s(\log a + \log r) &\leq O\left(x^{2/3} (\log x)^{2/3} \cdot \log s \cdot (\log x + \log \log s)\right) \\
            &\leq O\left(\frac{(\log \log n)^{2/3}}{(\log s)^{1/3}} \cdot (\log \log \log n + \log \log s)\right) \\
            &\leq o(\log \log n),
        \end{align*}
        as claimed.
        Lastly, we have already shown that $a^2r^2 \leq o(ar\log \log n) \Rightarrow ar \leq o(\log \log n)$.
    \end{proof}

    \begin{claim}\label{claim:medd}
        We have:
        $$d \leq \exp\left(\exp\left(O\left(\frac{(\log \log n)^{2/3} (\log \log \log n)^{1/3}}{(\log s)^{1/3}} \cdot \left(\log\left(\frac{\log \log n}{(\log \log \log n) \cdot (\log s)^2}\right)\right)^{2/3}\right)\right)\right).$$
    \end{claim}
    \begin{proof}
        By Claim~\ref{claim:medeta}, we can use the second branch of Corollary~\ref{cor:mvparams} and compute:
        \begin{align*}
            d &= \exp\left(O\left(\log n \cdot \frac{(\log \eta)^{1-1/\ell}}{\ell \eta^{1/\ell}}\right)\right) \\
            &\leq \exp\left(O\left(\log n \cdot \frac{(ar\log \log n)^{1-1/\ell}}{\ell \eta^{1/\ell}}\right)\right) \\
            &\leq \exp\left(O\left(\log n \cdot \frac{(ar\log \log n)^{1-1/\ell}}{ar \exp\left((1-\frac{1}{ar+1})\log \log n - O(a^2r \log a(\log\log\log n))\right)}\right)\right) \\
            &= \exp\left(\exp\left(O\left(\frac{\log \log n}{ar} + \log \log \log n +  a^2r\log a(\log \log \log n)\right)\right)\right).
        \end{align*}
        We address each of the three terms in the parentheses above separately.
        First, we have:
        \begin{align*}
            a^2r\log a(\log \log \log n) &\leq O\left(x^{2/3}(\log x)^{2/3} \cdot \frac{\log s}{\log x} \cdot \log x \cdot (\log \log \log n)\right) \\
            &= O\left(x^{2/3}(\log x)^{2/3} \cdot \log s \cdot (\log \log \log n)\right) \\\\
            &= O\left(\frac{(\log \log n)^{2/3}}{(\log \log \log n)^{2/3} (\log s)^{4/3}} (\log x)^{2/3} \cdot \log s \cdot \log \log \log n\right)\\
            &= O\left(\frac{(\log \log n)^{2/3} (\log \log \log n)^{1/3}}{(\log s)^{1/3}} (\log x)^{2/3}\right),
        \end{align*}
        which is equivalent to the desired bound by definition of $x$.
        From this, we can also infer the same upper bound on $\log \log \log n$.
        As for the first of the three terms, we have:
        \begin{align*}
            \frac{\log \log n}{ar} &= \Theta\left(\frac{\log \log n}{(x \log x)^{1/3} (\log s/\log x)}\right) &\text{by definition of $a, r$} \\
            &= \Theta\left(\frac{\log \log n}{x^{1/3} \log s} \cdot (\log x)^{2/3}\right) \\
            &= \Theta\left(\frac{\log \log n}{\log s} \cdot \frac{(\log \log \log n)^{1/3} \cdot (\log s)^{2/3}}{(\log \log n)^{1/3}} (\log x)^{2/3}\right) \\
            &= \Theta\left(\frac{(\log \log n)^{2/3}(\log \log \log n)^{1/3}}{(\log s)^{1/3}} \cdot (\log x)^{2/3}\right).
        \end{align*}
        This completes the proof of this claim.
    \end{proof}
    Now to complete the proof and bound $\mathsf{CC}$, by Theorem~\ref{thm:gks} it remains to verify that:
    $$s \log |\F| \leq \exp\left(\exp\left(O\left(\frac{(\log \log n)^{2/3} (\log \log \log n)^{1/3}}{(\log s)^{1/3}} (\log x)^{2/3}\right)\right)\right).$$
    First we show that $s$ satisfies this bound:
    \begin{align*}
        \log s &\leq o(\sqrt{\log \log n}) \\
        \Rightarrow (\log s)^{1/3} &\leq o((\log \log n)^{1/6}) \\
        \Rightarrow (\log s)^{1/3} \log \log s &\leq o((\log \log n)^{2/3}) \\
        \Rightarrow \log \log s &\leq \frac{(\log \log n)^{2/3}}{(\log s)^{1/3}} \\
        &\leq O\left(\frac{(\log \log n)^{2/3} (\log \log \log n)^{1/3}}{(\log s)^{1/3}} (\log x)^{2/3}\right) \\
        \Rightarrow s &\leq \exp\left(\exp\left(O\left(\frac{(\log \log n)^{2/3} (\log \log \log n)^{1/3}}{(\log s)^{1/3}} (\log x)^{2/3}\right)\right)\right).
    \end{align*}
    Similarly, we have:
    \begin{align*}
        \log|\F| &\leq \exp(O(ar\log \log \log n)) \\
        &\leq \exp(O(\log \log n \cdot \log \log \log n)) \\
        &\leq \exp(\exp(O(\log \log \log n))) \\
        &\leq \exp\left(\exp\left(O\left((\log \log n)^{1/2}\right)\right)\right) \\
        &\leq \exp\left(\exp\left(O\left(\frac{(\log \log n)^{2/3}}{(\log s)^{1/3}}\right)\right)\right) &\text{by the above bound on $\log s$} \\
        &\leq \exp\left(\exp\left(O\left(\frac{(\log \log n)^{2/3} (\log \log \log n)^{1/3}}{(\log s)^{1/3}} (\log x)^{2/3}\right)\right)\right).
    \end{align*}
\end{proof}

\subsection{Many-Server Regime}\label{sec:manyserver}

It now remains to deal with the regime where $s \geq \exp(\Omega(\sqrt{\log \log n/\log \log \log n}))$.
In this regime, it appears that our techniques together with Conjecture~\ref{conjecture:superconstant} do not offer any advantage at all over prior work~\cite{Efr09}.
Put another way, the communication-minimizing setting of parameters in Conjecture~\ref{conjecture:superconstant} would be to set $r$ to 0---or equivalently, just use the $2^{\ell}$-sparse decoding polynomial modulo a product of $\ell$ primes given by Corollary~\ref{cor:trivialdecodertensored}.
This is exactly the state-of-the-art approach prior to our work~\cite{Efr09}, which we summarized in Theorem~\ref{thm:manyserver}.
For completeness, we re-derive its communication bound here, making no claim to originality.

\begin{proof}[Proof of Theorem~\ref{thm:manyserver}]
    We begin by setting parameters and calculating the resulting bounds on $m, |\F|$.
    To do this, we invoke Corollary~\ref{cor:trivialdecodertensored} with $\lfloor \log_2 s \rfloor$ primes (note that the corollary is applicable as $s \geq 32 \Rightarrow \lfloor \log_2 s \rfloor \geq 5$) to obtain $m, \F$ of characteristic $q = 2$ such that there is an $S^*_m$-decoding polynomial over $\F$ of sparsity $\leq 2^{\lfloor \log_2 s \rfloor} \leq s$ and we have the bounds:
    \begin{align*}
        m &\leq (\log_2 s \log \log_2 s)^{\log_2 s} \\
        \log|\F| &\leq \exp\left(O(\log s \log \log s)\right).
    \end{align*}
    Then for the purposes of Corollary~\ref{cor:mvparams}, we have $\ell = \lfloor \log_2 s \rfloor + 1$.

    \begin{claim}\label{clm:largeeta}
        We have:
        $$\exp(\Omega(\ell^2)) \leq \frac{(\log n)^{\lfloor \log_2 s \rfloor}}{2(\ell \log \ell)^\ell} \leq \eta \leq (\log n)^{\lfloor \log_2 s \rfloor}.$$
    \end{claim}
    \begin{proof}
        First we lower bound $\eta$:
        \begin{align*}
            \eta &= \frac{(\log n)^{\ell-1}}{2m} \\
            &\geq \frac{(\log n)^{\lfloor \log_2 s \rfloor}}{2(\ell \log \ell)^\ell}.
        \end{align*}
        The upper bound on $\eta$ is immediate from $\eta \leq (\log n)^{\ell-1}$.
        Finally, to justify the leftmost inequality, we want to show that:
        \begin{align*}
            \Omega(\ell^2) &\leq \log \left(\frac{(\log n)^{\lfloor \log_2 s \rfloor}}{2(\ell \log \ell)^\ell}\right) \\
            &= \lfloor \log_2 s \rfloor \log \log n - \log 2 - \ell \log \ell - \ell \log \log \ell \\
            \Leftrightarrow \Omega(\ell^2) &\leq \lfloor \log_2 s \rfloor \log \log n.
        \end{align*}
        This follows from noting that $\ell^2 \leq O((\log s)^2) \leq O(\log s \cdot \log \log n)$, since we are assuming that $s \leq \log n$.
    \end{proof}

    \begin{claim}
        We have:
        \begin{align*}
            d &\leq \exp\left(O\left((\log n)^{\frac{1}{\lfloor \log_2 s \rfloor + 1}} \cdot (\log \log n)^{1-\frac{1}{\lfloor \log_2 s \rfloor + 1}} \cdot \log s \cdot \log \log s\right)\right).
        \end{align*}
    \end{claim}
    \begin{proof}
        By Claim~\ref{clm:largeeta}, we use the second branch of Corollary~\ref{cor:mvparams} to compute:
        \begin{align*}
            d &= \exp\left(O\left(\log n \cdot \frac{(\log \eta)^{1-1/\ell}}{\ell \eta^{1/\ell}}\right)\right) \\
            &\leq \exp\left(O\left(\log n \cdot \frac{((\ell-1)\log \log n)^{1-1/\ell}}{\ell \eta^{1/\ell}}\right)\right) \\
            &\leq \exp\left(O\left(\log n \cdot \frac{((\ell-1)\log \log n)^{1-1/\ell}}{\ell (\log n)^{1-1/\ell} /(2^{1/\ell} \ell \log \ell)}\right)\right) \\
            &= \exp\left(O\left(\log n \cdot \frac{(\log \log n)^{1-1/\ell}}{(\log n)^{1-1/\ell} /(\ell \log \ell)}\right)\right)\\
            &= \exp\left(O\left((\log n)^{1/\ell} \cdot (\log \log n)^{1-1/\ell} \cdot \ell \log \ell\right)\right)\\
            &= \exp\left(O\left((\log n)^{\frac{1}{\lfloor \log_2 s \rfloor + 1}} \cdot (\log \log n)^{1-\frac{1}{\lfloor \log_2 s \rfloor + 1}} \cdot \log s \cdot \log \log s\right)\right).
        \end{align*}
    \end{proof}
    Now to complete the proof and bound $\mathsf{CC}$, by Theorem~\ref{thm:gks} it remains to verify that:
    $$s \log |\F| \leq \exp\left(O\left((\log n)^{\frac{1}{\lfloor \log_2 s \rfloor + 1}} \cdot (\log \log n)^{1-\frac{1}{\lfloor \log_2 s \rfloor + 1}} \cdot \log s \cdot \log \log s\right)\right).$$
    The $s$ term on the LHS is clearly absorbed by the $\log s$ on the RHS, and we know that $\log |\F| \leq \exp(O(\log s \log \log s))$ so this term will also be absorbed.

    Finally, we justify the coarser bound of:
    $$\exp\left(O\left((\log n)^{\frac{1}{\lfloor \log_2 s \rfloor + 1}} \cdot (\log \log n)^{1-\frac{1}{\lfloor \log_2 s \rfloor + 1}} \cdot \log s \cdot \log \log s\right)\right) \leq \exp\left(\exp\left(O\left(\log \log \log n + \frac{\log \log n}{\log s}\right)\right)\right).$$
    The $\frac{\log \log n}{\log s}$ term from the RHS absorbs the $(\log n)^{\frac{1}{\lfloor \log_2 s \rfloor + 1}}$ from the LHS.
    The $\log \log \log n$ term from the RHS absorbs the $(\log \log n)^{1 - \frac{1}{\lfloor \log_2 s \rfloor + 1}}$ term on the LHS.
    Finally, we have $\log s \cdot \log \log s \leq \exp(O(\log \log s)) \leq \exp(O(\log \log \log n))$ so the $\log \log \log n$ also absorbs the remaining terms on the LHS.
\end{proof}

\section{Concrete Empirical Results for Small $s$}\label{sec:empirics}

While we believe Conjecture~\ref{conjecture:onetuple-intro} and therefore Theorem~\ref{thm:intromain} and Corollary~\ref{cor:asymptoticpir} to be highly plausible, we still believe it is important to complement this with explicit examples of PIR protocols achieving the guarantees we claim for finite values of $s$, to the extent possible.
A similar empirical study was conducted by \cite{Chee2011} to support their conjectural construction of 3-sparse decoding polynomials modulo products of two primes.
In this section, we present the empirical results that we obtained to this end.

For each prime characteristic $q$, we aim to find $q-1$ primes $d_i$ satisfying the preconditions of Lemma~\ref{lemma:suffcondsharp}, taking $p_i$ to be the largest prime divisor of $q^{d_i} - 1$ and ensuring that the $p_i$'s are distinct.
If we find $k$ such primes, then we know there exists a $(k+1)$-term decoding polynomial with respect to some product of $k$ primes.
Moreover, it turns out that there also exists a $(k'+1)$-term decoding polynomial with respect to some product of $k'$ primes for any $k' \leq k$.
To see this ``monotonicity'' property, simply note that a root-of-unity grid for $k$ primes can be turned into a root-of-unity grid for $k'$ primes by dropping some rows and columns.
(In fact, monotonicity was also shown in more generality for share conversion schemes by Beimel and Lasri~\cite[Theorem 1.4]{bl26}.)

Bearing this in mind, we search over primes $d$ satisfying the inequality of Lemma~\ref{lemma:suffcondsharp} when we take $p$ to be the largest prime divisor of $q^{d} - 1$.
Note that the precondition $q^{\gcd(d_1, \ldots, d_k)} \geq k+1$ is automatically ensured by $q \geq k+1$.
Our search has two distinct components which we detail next.

Before that, we remark that an important distinction is that some of the numbers we are considering are so large that it is not easy to verify their primality with certainty.
Such numbers are tested with ``probable prime'' tests.
Therefore some of our constructions rely on some probable primes actually being prime.
We will explicitly highlight those cases.

\paragraph{Base $q$ repunit primes.}
First, in Table~\ref{tbl:primes} we tabulate for various characteristics $q$ degrees $d$ such that the base $q$ repunit $(q^d-1)/(q-1)$ is a (probable) prime.
For such a $d$, we can take $p = (q^d-1)/(q-1)$ and $h = q-1$, so the third condition in Lemma~\ref{lemma:suffcondsharp} becomes:
\begin{align*}
    \frac{q^{d}-k}{(q-1)^k} &> (k-1)\sqrt{q^{d}}(1-1/(q-1)^k)+1 \\
    \Leftrightarrow q^{d} - k &> (k-1) q^{d/2} ((q-1)^k - 1) + (q-1)^k,
\end{align*}
which holds for any $d \geq 2k+2$. In particular, any $d \geq 2q$ will be usable for every $k \leq q-1$.
Values of $d < 2q$ will also be usable for us in some cases; any such cases are included in Table~\ref{tbl:roughcomposites}.

We include all characteristics $\leq 19$ and after that we include characteristics for which many such $d$ are known.
This table was built using the information at \href{https://pzktupel.de/Primetables/TableRepunitGenS.php}{https://pzktupel.de/Primetables/TableRepunitGenS.php}.
For completeness, we also provide OEIS references for these sequences when known.

\begin{table}
\begin{center}
    \begin{tabular}{|c|c|c|c|}
    \hline
    \multirow{2}{*}{\textbf{Characteristic $q$}} & \multirow{2}{*}{\textbf{OEIS ID}} & \textbf{\# of (Probable) Primes} & \multirow{2}{*}{\textbf{Degrees $d$}} \\
    & & \textbf{(Clipped at $q-1$)} & \\
    \hline
    \hline
    3 & A028491 & 2 & 7, 13 \\
    \hline
    5 & A004061 & 4 & 11, 13, 47, 127 \\
    \hline
    7 & A004063 & 6 & 131, 149, 1699, 14221, \emph{35201, 126037} \\
    \hline
    \multirow{2}{*}{11} & \multirow{2}{*}{A005808} & \multirow{2}{*}{10} & 73, 139, 907, 1907, 2029, 4801, \\
    & & & 5153, 10867, 20161, \emph{293831} \\
    \hline
    \multirow{2}{*}{13} & \multirow{2}{*}{A016054} & \multirow{2}{*}{11} & 137, 283, 883, 991, 1021, 1193 \\
    & & & 3671, \emph{18743, 31751, 101089, 1503503} \\
    \hline
    \multirow{2}{*}{17} & \multirow{2}{*}{A006034} & \multirow{2}{*}{8} & 47, 71, 419, 4799, \emph{35149, 54919} \\
    & & & \emph{74509, 1990523} \\
    \hline
    \multirow{2}{*}{19} & \multirow{2}{*}{A006035} & \multirow{2}{*}{9} & 47, 59, 61, 107, 337, 1061 \\
    & & & \emph{9511, 22051, 209359} \\
    \hline
    \multirow{2}{*}{37} & \multirow{2}{*}{A128003} & \multirow{2}{*}{10} & 181, 251, 463, 521, \emph{7321, 36473} \\
    & & & \emph{48157, 87421, 168527, 249341} \\
    \hline
    43 & A240765 & 6 & \emph{6277, 26777, 27299, 40031, 44773, 194119} \\
    \hline
    67 & & 7 & 367, 1487, 3347, 4451, \emph{10391, 13411, 167449} \\
    \hline
    71 & & 7 & 157, 1583, \emph{31079, 55079, 72043, 150697, 172259} \\
    \hline
    \end{tabular}
    \end{center}
    \caption{For each characteristic $q$, we list here known prime degrees $d \geq 2q$ such that $(q^d-1)/(q-1)$ is either a prime or a probable prime. Probable primes are italicized. The $d \geq 2q$ requirement is to ensure that these primes are usable for every value $k \leq q-1$; some values $d < 2q$ may still be usable for particular values of $k$. We only list up to $q-1$ primes in each row, since this is all we need when $q \geq k+1$.}
    \label{tbl:primes}
\end{table}

\paragraph{Base $q$ repunit rough composites.}
For characteristics $q \leq 11$, we can already attain a root-of-unity grid with $k = q-1$ using just the primes already discussed.
For $q > 11$, we augment our search by including cases where $(q^d-1)/(q-1)$ is a rough composite integer.
By Corollary~\ref{cor:roughness-rug}, this is still sufficient for our purposes.
To this end, we exhaustively search over $d \leq 60000$ and see whether taking $p$ to be the largest prime divisor of $(q^d-1)/(q-1)$ will satisfy the third condition of Lemma~\ref{lemma:suffcondsharp}.
We obtain these prime divisors (when known) using queries to FactorDB.
We tabulate the results (in addition to repunit primes with $d < 2q$ that are still usable for these cases) in Table~\ref{tbl:roughcomposites}.

\begin{table}
\begin{center}
    \begin{tabular}{|c|c|c|}
    \hline
    \multirow{2}{*}{\textbf{Characteristic $q$}} & \textbf{Largest Achievable} & \multirow{2}{*}{\textbf{Degrees $d$ in Addition to Table~\ref{tbl:primes}}} \\
    & $k \leq Q-1$ & \\
    \hline
    \hline
    3 & 2 & \multirow{4}{*}{None required} \\
    \cline{1-2}
    5 & 4 & \\
    \cline{1-2}
    7 & 6 & \\
    \cline{1-2}
    11 & 10 & \\
    \hline
    13 & 12 & 101 \\
    \hline
    17 & 16 & 131, 229, 359, 647, 701, 2269, 2459, 2879 \\
    \hline
    19 & 18 & 311, 347, 907, 1571, 1759, 2539, \emph{7019, 12919, 18947} \\
    \hline
    \multirow{2}{*}{37} & \multirow{2}{*}{23} & 71, 971, 1451, 1559, 1979, \emph{2371, 2683} \\
    & & 3371, \emph{5431, 6173, 9521, 14197, 55661} \\
    \hline
    \multirow{2}{*}{43} & \multirow{2}{*}{17} & 181, 311, 457, 653, 1021, 1283 \\
    & & \emph{2081, 8317, 22543, 22639, 33623} \\
    \hline
    \multirow{2}{*}{67} & \multirow{2}{*}{20} & 727, 1279, 1361, 1567, 2081, \emph{2693, 3539} \\
    & & \emph{3719, 3907, 6373, 8893, 19609, 56171} \\
    \hline
    71 & 15 & 41, 269, 349, 431, 569, 613, 1231, 2711 \\
    \hline
    \end{tabular}
    \end{center}
    \caption{For each characteristic $q$, we tabulate the largest $k \leq Q-1$ for which we found $k$ (probable) primes satisfying the conditions in Lemma~\ref{lemma:suffcondsharp}. We searched over degrees $d \leq 60000$ and took the largest (probable) prime factor of $(q^d-1)/(q-1)$ in each case. These prime factors (when known) were computed using queries to FactorDB. Once again, probable primes are italicized.}
    \label{tbl:roughcomposites}
\end{table}

\paragraph{Summary.}
We now summarize with a proof of Corollary~\ref{cor:concrete}:

\begin{proof}[Proof of Corollary~\ref{cor:concrete}]
On the unconditional side (so without relying on probable primes), the largest $k$ we obtain for Theorem~\ref{thm:intromain} is $k = 14$ for $q = 19$.
Feeding this into the proof of Corollary~\ref{cor:asymptoticpir} gives the claimed PIR for up to 15 servers.
The relevant degrees $d$ are as follows:
\begin{itemize}
    \item $d$ such that $(19^d-1)/18$ is prime: $31, 47, 59, 61, 107, 337, 1061$.
    \item $d$ such that $(19^d-1)/18$ is a rough composite: $83, 311, 347, 601, 907, 953, 1571$. ($d = 1759$ and $2539$ are also available alternatives, although these do not allow us to obtain a tuple for $k \geq 15$ because $d = 31, 83$ are only usable for $k \leq 14$.)

    The cases of $d = 83, 601, 953$ are not shown in Table~\ref{tbl:roughcomposites} because they are not usable for the $k = 18$ result tabulated there, although they are usable for the $k = 14$ result we are showing here.
\end{itemize}
On the conditional side (assuming that integers passing probable prime tests are actually prime), the largest $k$ we obtain for Theorem~\ref{thm:intromain} is $k = 23$ for $q = 37$.
Feeding this into the proof of Corollary~\ref{cor:asymptoticpir} gives the claimed PIR for up to 24 servers.
This relies on 6 probable primes from Table~\ref{tbl:primes} and 7 probable primes from Table~\ref{tbl:roughcomposites}.
For clarity, we list these below:
\begin{itemize}
    \item The 6 probable primes from Table~\ref{tbl:primes} are $\frac{37^d-1}{36}$ for $d = 7321, 36473, 48157, 87421, 168527, 249341$.
    \item The 7 probable primes from Table~\ref{tbl:roughcomposites} are:
    \begin{itemize}
        \item $(37^{2371}-1)/3413169294372$;
        \item $(37^{2683}-1)/307412007186684769953576371198892$;
        \item $(37^{5431}-1)/76456531836$;
        \item $(37^{6173}-1)/3931185602686020572988$;
        \item $(37^{9521}-1)/1419551395838858218572$;
        \item $(37^{14197}-1)/12198192238493172$; and
        \item $(37^{55661}-1)/4007628$.
    \end{itemize}
\end{itemize}
\end{proof}

\section{Conclusion}\label{sec:concl}

For any constant $k$, our work---assuming Conjecture~\ref{conjecture:onetuple-intro}---completely resolves the question of the minimal sparsity achievable by an $S^*_m$-decoding polynomial for $m$ a product of $k$ distinct primes, by showing that sparsity $k+1$ is achievable.
This is the best one can hope for due to lower bounds shown by~\cite{DBLP:conf/innovations/GhasemiK26,bl26}.
Via the results by~\cite{GKS25}, our result implies that for any constant $s \geq 2$ there exists an $s$-server PIR protocol that on an $n$-bit database requires communication $\exp(O((\log n)^{1/s} (\log \log n)^{1-1/s}))$, which is strictly better than the previous state of the art for any $s \geq 4$.
As additional contributions, we have run experiments to remove the need for any reliance on conjectures for $s \leq 15$, and also taken a first step towards reducing the communication of PIR when the number of servers grows with $n$.

We next provide a summary of some related work, before concluding with some future directions.

\subsection{Related Work}\label{sec:relatedwork}

\paragraph{Decoding polynomials.}
The central object of study in our paper is $S$-decoding polynomials.
These were first introduced by Efremenko~\cite{Efr09}, where he introduced this notion as a powerful abstraction for PIR.
He also provided an example of a decoding polynomial attaining nontrivial sparsity: for $m = 511 = 7 \times 73$, he showed that there is an $S^*_m$-decoding polynomial of sparsity 3.
This was followed up by work by \cite{Chee2011}, who demonstrated that under a plausible number-theoretic conjecture, there is an $S^*_m$-decoding polynomial of sparsity 3 for infinitely many $m = p_1p_2$.
More recently, work by Ghasemi and Kopparty~\cite{DBLP:conf/innovations/GhasemiK26} shows lower bounds against the sparsity of $S$-decoding polynomials, in settings including but not limited to the $S^*_m$-decoding setting that is most relevant to this work.

\paragraph{Secret share conversion.}
A closely related object to decoding polynomials is \emph{share conversion}.
This is an elegant abstraction introduced by~\cite{DBLP:conf/coco/BeimelIKO12} that generalizes $S$-decoding polynomials and remains applicable to PIR.
In share conversion, a secret $a \in S \subseteq \Z_m$ is converted into shares $\sh_1, \ldots, \sh_T$ that are distributed to $T$ parties such that each individual share perfectly hides $a$.
The defining requirement is there are local functions $C_1, \ldots, C_T$ with outputs in $\F$ such that $\sum_{i \in [T]} C_i(\sh_i) = 0$ if and only if $a \neq 0$.

To provide some intuition, we show why an $S$-decoding polynomial directly gives share conversion, as was shown by~\cite{DBLP:conf/coco/BeimelIKO12}: let the decoding polynomial be $P(X) = \sum_{i \in [T]} \alpha_i X^{\beta_i}$.
Then to share $a$, we will sample $r \gets \Z_m$ and set $\sh_i = r + a\beta_i$.
The $i$th party's local computation is $C_i(\sh) = \alpha_i \gamma^{\sh}$.
Now we have:
$$\sum_{i \in [T]} C_i(\sh_i) = \sum_{i \in [T]} \alpha_i \gamma^{\sh_i} = \sum_{i \in [T]} \alpha_i \gamma^{r+a\beta_i} = \gamma^r \cdot P(\gamma^a),$$
which is indeed 0 for $a \in S \backslash \{0\}$ and nonzero for $a = 0$.
As mentioned earlier, share conversion also captures constructions that cannot be captured by $S$-decoding polynomials.
In fact, the alternative presentation by~\cite{DBLP:conf/tcc/AlonBL25} of the PIR construction of~\cite{GKS25} adds the generalization that decoding polynomials can be replaced by share conversion from $\Z_m$ to a field $\F$ of characteristic not dividing $m$.
A natural question this raises is whether the results of~\cite{GKS25} can be improved by using a share conversion with few parties in a way that decoding polynomials would not capture.
However, recent follow-up work by Beimel and Lasri~\cite{bl26} answers this negatively; in the setting where the share conversion is from $\Z_m$ to $\F$ and the characteristic of $\F$ does not divide $m$, share conversion and decoding polynomials are equivalent objects.

\paragraph{Locally-decodable codes.}
An essentially equivalent object to PIR is locally-decodable codes (LDCs), introduced by Katz and Trevisan~\cite{DBLP:conf/stoc/KatzT00,trevisan04}.
A PIR protocol with $T$ servers, queries of length $\ell_q := \ell_q(n)$, and answers of length $\ell_a := \ell_a(n)$, implies a smooth $T$-query LDC mapping $n$ bits into a codeword of length $2^{\ell_q}$ over an alphabet of size $2^{\ell_a}$.
We refer the reader to Yekhanin's comprehensive survey~\cite{Yek12} for a discussion of state-of-the-art locally-decodable codes across different regimes.
The cases of Reed-Muller codes and matching-vector codes are also captured by the phase diagram in Figure~\ref{fig:phasediagram}.

\paragraph{Doubly-efficient PIR.}
The only PIR metrics that our work studies are the number of servers and the communication complexity.
One could also ask for a protocol where all parties run in sublinear time $n^{o(1)}$ per query.
Such a protocol is known as doubly-efficient PIR, and necessitates working in an alternate model such as PIR-with-preprocessing~\cite{DBLP:conf/crypto/BeimelIM00}.
The problem of doubly-efficient PIR is very important and has seen exciting advances in recent years~\cite{LMW23}; however, these protocols are ultimately incomparable to those we consider in this work.
They are computationally more efficient, but at the expense of either requiring more communication (albeit still sublinear in $n$) or only being secure against computationally-bounded adversaries.

\subsection{Open Questions}\label{sec:openquestions}

Our work nails down the best communication-server tradeoff achievable for a constant number of servers, while staying within the ``matching vectors + decoding polynomials'' framework~\cite{GKS25} and adhering to the matching vector construction by~\cite{Barrington1994,DBLP:journals/combinatorica/Grolmusz00,DGY11}. This leaves the following outstanding gaps:

\paragraph{Superconstant servers.}
Within the same framework and matching vector construction, what is the best communication-server tradeoff achievable in regimes where $s$ is superconstant in $n$?
As we will outline in Section~\ref{sec:superconstantoverview}, our work makes progress on answering this question but we suspect improvements should be possible.
We see at least a couple of candidate routes towards addressing this:
\begin{itemize}
    \item Find families of primes $p_1, \ldots, p_k$ that occur more frequently and admit decoding polynomials of sparsity $k+1$.
    As discussed in Section~\ref{sec:rughowto}, we seem to be far from a complete understanding of exactly when these decoding polynomials exist so this direction is promising; our work only provides evidence for their existence for infinitely many $m, \F$.

    \item More modestly, find families of $p_1, \ldots, p_k$ that occur more frequently, at the expense of having decoding polynomials that do not quite hit sparsity $k+1$ but attain some intermediate regime of sparsity that is still useful.
    Our work focuses on the sparsity exactly $k+1$ case; it may be beneficial when considering a superconstant number of servers to relax this.
\end{itemize}

\paragraph{Unconditional constructions.}
Our main theorem rests on a number-theoretic conjecture. While we have provided heuristic justification that this conjecture is plausible, it would be preferable to have unconditional statements.
Heuristics of a similar flavor have appeared in previous constructions of matching-vector PIR~\cite{Y08,Chee2011}; we expect proving these heuristics to be extremely difficult since they are related, for example, to the infinitude of Mersenne primes which is wide open.
However, one could hope for a different construction from ours that does not rest on Mersenne-like conjectures.

\paragraph{Better constructions of matching vector families.}
The past decades have seen advances in finding nontrivially sparse $S^*_m$-decoding polynomials, yet tantalizingly the best matching vector family construction has stood since the work of Grolmusz~\cite{DBLP:journals/combinatorica/Grolmusz00} in 2000!
Existing lower bounds~\cite{DBLP:journals/siamcomp/0001DL14} are far from suggesting that this construction is tight.
We believe this to be a difficult but underexplored and high-impact open direction.

\paragraph{New frameworks for PIR.}
Another ambitious direction is to find new ways to build low-communication PIR from algebraic primitives, including---but not limited to---matching vectors.
Once again, all constructions in the constant-server and sub-polynomial-communication regime are encapsulated by the unifying framework of~\cite{GKS25}, but we see no reason why this should be the end of the story.

\bibliographystyle{alpha}
\bibliography{refs}

\appendix

\section{Proof of Corollary~\ref{cor:mvparams}}\label{sec:mvparams}

In all cases, we will proceed by instantiating Theorem~\ref{thm:mvgeneral} for suitably chosen $h$ and $w$.

\paragraph{Case 1: $\eta \leq 2^{O(\ell^2)}$.}
Fix a constant $C_1 > 0$ and take $w = \lfloor C_1 \log n \rfloor$.
Then we have $mw \sim C_1 \frac{(\log n)^\ell}{\eta}$, or equivalently:
\begin{equation}
    (mw)^{1/\ell} \sim C_1^{1/\ell} \frac{\log n}{\eta^{1/\ell}}.
\end{equation}
So since $\eta \geq 1$, this implies $(mw)^{1/\ell} \leq C_1^{1/\ell} \log n$.

Also, set $C_2 > 0$ and $h \sim C_2 \log n$ such that $\binom{h}{w} \geq n$ and $C_2 > 3\max(C_1, 1)$.
Then we have:
\begin{align*}
    \sum_{j = 0}^{\lfloor (mw)^{1/\ell} \rfloor} \binom{h}{j} &= O\left(\binom{h}{(mw)^{1/\ell}}\right) \\
    &= O\left(\binom{C_2 \log n}{C_1^{1/\ell} \log n/\eta^{1/\ell}}\right).
\end{align*}
In the first step above, we are using the well-known inequality that for $a \leq b/3$ we have $\binom{b}{0} + \ldots + \binom{b}{a} = O\left(\binom{b}{a}\right)$, together with the fact that:
$$(mw)^{1/\ell} \leq C_1^{1/\ell} \log n \leq \max(C_1, 1) \log n < \frac{C_2}{3}\log n \sim \frac{h}{3}.$$
Continuing along, if $\eta \leq 2^{O(\ell)}$, then this is $\poly(n)$ which is indeed consistent with the statement of the corollary.
If $\eta \geq 2^{\Omega(\ell)}$, then we have:
\begin{align*}
    \log \binom{C_2 \log n}{C_1^{1/\ell} \log n/\eta^{1/\ell}} &= \Theta\left(\frac{C_1^{1/\ell} \log n}{\eta^{1/\ell}} \cdot \log \frac{C_2 \eta^{1/\ell}}{C_1^{1/\ell}}\right) \\
    &= \Theta\left(\frac{\log n \log(\eta^{1/\ell})}{\eta^{1/\ell}}\right),
\end{align*}
so this construction achieves the claimed performance.

\paragraph{Case 2: $\eta \geq 2^{\Omega(\ell^2)}$.}
Here, we will take $\theta \geq 10$ such that $\theta = \Theta(\log \eta/\ell^2)$ and take $w \sim \frac{\log n}{\theta}$ and $h \sim \frac{\log n}{\theta} \cdot 2^\theta$.
In this regime, we have $\binom{h}{w} \geq (h/w)^w = n$.
Secondly, we have $mw \sim \frac{(\log n)^\ell}{\eta\theta} \Rightarrow (mw)^{1/\ell} \sim \frac{\log n}{\eta^{1/\ell} \theta^{1/\ell}} < h/3$.
Consequently, we once again have:
\begin{align*}
    \sum_{j = 0}^{\lfloor (mw)^{1/\ell} \rfloor} \binom{h}{j} &= O\left(\binom{h}{(mw)^{1/\ell}}\right),\text{ and }\\
    \log \binom{h}{(mw)^{1/\ell}} &= \log \binom{\log n \cdot 2^\theta/\theta}{\log n/(\eta \theta)^{1/\ell}} \\
    &= \Theta\left(\frac{\log n}{(\eta \theta)^{1/\ell}} \log \frac{2^\theta \eta^{1/\ell}}{\theta^{1-1/\ell}}\right) \\
    &= \Theta\left(\frac{\log n}{(\eta \log \eta)^{1/\ell}} \log \frac{2^\theta \eta^{1/\ell}}{\theta^{1-1/\ell}}\right) \\
    &= \Theta\left(\frac{\log n}{(\eta \log \eta)^{1/\ell}} \left(\theta + \frac{\log \eta}{\ell}\right)\right) \\
    &= \Theta\left(\frac{\log n}{(\eta \log \eta)^{1/\ell}} \left(\frac{\log \eta}{\ell}\right)\right) \\
    &= \Theta\left(\frac{(\log n) \cdot (\log \eta)^{1-1/\ell}}{\ell\eta^{1/\ell}}\right),
\end{align*}
as claimed.

\section{Proof of Theorem~\ref{thm:tensoring}}\label{sec:tensoringproof}

\begin{lemma}\label{lemma:rootnomatter}
    Let $S \subseteq \Z_m$ be an arbitrary subset containing 0 and $\F$ a field such that $m | |\F|-1$.
    Assume there exists some primitive $m$th root of unity $\gamma_m \in \F$ and an $S$-decoding polynomial $A$ with respect to $\gamma_m$ of sparsity $T$.
    Then for any primitive $m$th root of unity $\delta_m \in \F$, there exists an $S$-decoding polynomial $B$ with respect to $\delta_m$ of sparsity $T$.
\end{lemma}
\begin{proof}
    Since $\gamma_m, \delta_m$ are primitive $m$th roots of unity, there exists an integer $u$ such that $\gcd(u, m) = 1$ and $\delta_m^u = \gamma_m$.
    So if we define $B(X) := A(X^u)$, we have for any $j \in \Z_m$ that $B(\delta_m^j) = A(\delta_m^{ju}) = A(\gamma_m^j)$.
    The conclusion follows.
\end{proof}

\begin{proof}[Proof of Theorem~\ref{thm:tensoring}]
Let $M=m_1m_2\cdots m_r$, $C=\operatorname{lcm}(c_1,\ldots,c_r)$, and $\F=\F_{q^C}$.
For each \(i\), since there is an \(S^*_{m_i}\)-decoding polynomial modulo \(m_i\) over
\(\F_{q^{c_i}}\), the field \(\F_{q^{c_i}}\) contains a primitive \(m_i\)-th
root of unity. Hence $m_i \mid q^{c_i} - 1 \Rightarrow m_i \mid q^C - 1$.
Because the \(m_i\)'s are pairwise coprime, it follows that their product $M \mid q^C - 1$ as well.
Thus \(\F^\times\) contains a primitive \(M\)-th root of unity. Fix one and call
it \(\gamma\).

For each $i$, set $M_i = M/m_i$ and let $\beta_i = \gamma^{M_i}$ be a primitive $m_i$-th root of unity.
Let $K_i \subseteq \F_{q^C}$ be the unique\footnote{Uniqueness can be seen by noticing that $K_i$ is exactly the splitting field of the polynomial $X^{q^{c_i}} - X$.} subfield that is isomorphic to $\F_{q^{c_i}}$.
Since $\beta_i$ is an $m_i$-th root of unity and $m_i \mid q^{c_i}-1$, we have $\beta_i \in K_i$.
So by Lemma~\ref{lemma:rootnomatter}, there exists an $S^*_{m_i}$-decoding polynomial
$$P_i(X)=\sum_{e\in E_i} c_{i,e}X^e$$
of sparsity $d_i$ over $K_i$ with respect to the root of unity $\beta_i$.
By reducing exponents modulo $m_i$, we may assume that all elements of $E_i$ lie in $\Z_{m_i}$ and $|E_i| = d_i$.
Note also that for each $e$ we have $c_{i, e} \in K_i \Rightarrow c_{i, e} \in \F$.
We may now define the polynomial
$$P(X) = \prod_{i = 1}^r P_i(X^{M_i}) \in \F[X].$$
We claim that $P$ is an $S^*_M$-decoding polynomial mod $M$ over $\F$.
Indeed, for $j \in S^*_M$ we have $P(\gamma^j) \neq 0 \Leftrightarrow P_i(\beta_i^j) = P_i(\gamma^{jM_i}) \neq 0 \forall i$.
It immediately follows that for $j = 0$ we have $P_i(1) \neq 0 \forall i \Rightarrow P(1) \neq 0$.
For $j \in S^*_M \backslash \{0\}$, there must exist an index $i$ such that $j \not\equiv 0 \pmod{m_i}$. For this $i$, we therefore have $P_i(\beta_i^j) = 0 \Rightarrow P(\gamma^j) = 0$.
This proves that $P$ is an $S^*_M$-decoding polynomial.

For sparsity, note that $P$ is a product of $r$ polynomials with sparsities $d_1, \ldots, d_r$.
This immediately implies that $P$ itself has sparsity $\leq d_1\ldots d_r$.
Although this is enough for our purposes, we show for completeness that $P$ has sparsity exactly $d_1\ldots d_r$.
To show this, we want to show that there are no collisions in the map $(e_1, \ldots, e_r) \mapsto \sum_{i = 1}^r e_iM_i\bmod{M}$ for $e_i \in \Z_{m_i}$.
This map is linear so it suffices to show that $\sum_{i = 1}^r e_iM_i \equiv 0 \pmod{M}$ implies that $e_i \equiv 0 \pmod{m_i}$ for all $i$.
For a particular $i^*$, this can be seen by reducing both sides of the congruence modulo $m_{i^*}$. This implies that $e_{i^*} M_{i^*} \equiv 0 \pmod{m_{i^*}} \Rightarrow e_{i^*} \equiv 0 \pmod{m_{i^*}}$, as claimed.
\end{proof}

\section{Proof of Lemma~\ref{lem:mertens}}\label{sec:strong-mertens}

In this section we prove Lemma~\ref{lem:mertens}, assuming GRH for all $L$-functions. First we define the function
\begin{align*}
        \vartheta(x; m, a) &= \sum_{\substack{\text{primes } p \le x \\ p = a \bmod m}} \log p.
    \end{align*}
Before we present the proof, we requirng the following two lemmas from analytic number theory.
\begin{lemma}[Corollary 13.8 of \cite{montgomery2007multiplicative}]\label{lem-grh-input}
    Fix any modulus $m$. Assuming the GRH for all $L$-functions modulo $m$. Let $a$ be an integer such that $\gcd(a,m) = 1$ and $x\ge 2$ be a real number. Then,
    \begin{align*}
        \left|\vartheta(x;m, a) - \frac{x}{\phi(m)}\right| &\le O(\sqrt{x} \cdot (\log x)^2).
    \end{align*}
\end{lemma}

\begin{lemma}[Abel summation formula, Theorem 4.2 of \cite{apostol2013introduction}] \label{lem:abel-sum}
    Let $0 < x < y$ be real numbers.
    Let $(b_n)_{n=0}^\infty$ be a sequence of complex numbers, and let $B(t) = \sum_{n \le t} b_n$ and let $f: \R \to \R$ be a continuously differentiable function on the interval $[x,y]$. Then,
    \begin{align*}
        \sum_{\substack{n \in \mathbb{Z}, \\ x < n \le y}} b_n \cdot f(n) = B(y) f(y) - B(x) f(x) - \int_x^y B(t) f'(t) \, dt.
    \end{align*}
\end{lemma}

\begin{proof}[Proof of Lemma~\ref{lem:mertens}]
    First, define
    \begin{align*}
        b_n &= 
        \begin{cases}
            \log n &\text{prime }n, n = a \bmod m,\\
            0 &\text{otherwise,}
        \end{cases}\\
        B(t) &= \vartheta(t; m, a),\\
        f(t) &= \frac{1}{t \log t}.
    \end{align*}
    Observe that for any $1 < x < y \le \infty$, $f$ is continuously differentiable on $[x,y]$, with 
    \begin{align}\label{eqn:derivative}
        f'(t) = - \frac{\log t + 1}{(t \log t)^2}
    \end{align}
    By Lemma~\ref{lem-grh-input}, we can write $B(t) = \vartheta(t; m, a) = t/\phi(m) + E(t)$, where $|E(t)| \le O(\sqrt{t} (\log t)^2)$.
    Then, applying Lemma~\ref{lem:abel-sum}, we get
    \begin{align*}
    \sum_{\substack{\text{primes } p \in [x,y]\\ p = a \bmod m}} \frac{1}{p} &= \sum_{x < n \le y} b_n \cdot f(n) \\
    &= B(y) f(y) - B(x) f(x) - \int_x^y B(t) f'(t) dt &\text{Lemma~\ref{lem:abel-sum}}\\
    &= \frac{1}{\phi(m)} \left(y f(y) - x f(x) - \int_x^y t f'(t) dt\right) + \xi,
    \end{align*}
    where the penultimate equality follows from writing $B(t) = t/\phi(m) + E(t)$, and calling $\xi = E(y) f(y) - E(x) f(x) - \int_x^y E(t) f'(t) dt$, and the last inequality follows from Claim~\ref{clm:error-term} and the following calculation.
    \begin{align*}
            \frac{1}{\phi(m)} \left(y f(y) - x f(x) - \int_x^y t f'(t) dt\right) 
            &= \frac{1}{\phi(m)}\left(\frac{1}{\log y}- \frac{1}{\log x} - \int_x^y t f'(t) dt\right) &\text{definition of $f$}\\
            &= \frac{1}{\phi(m)}\left(\frac{1}{\log y}- \frac{1}{\log x} - t f(t)|^y_x + \int_x^y f(t) dt\right) &\text{integration by parts}\\
            &= \frac{1}{\phi(m)}\left(\int_x^y f(t) dt\right) &\text{evaluating $tf(t)$ at $x, y$}\\
            &= \frac{1}{\phi(m)} \left( \log \log t | ^y_x\right)\\
            &= \frac{1}{\phi(m)} {\log \log y - \log \log x}.
        \end{align*}
    \begin{claim}\label{clm:error-term}
        \begin{align*}
            |\xi| := \left|E(y) f(y) - E(x) f(x) - \int_x^y E(t) f'(t) dt\right| \le O\left(\frac{\log x}{\sqrt{x}}\right).
        \end{align*}
    \end{claim}
    \begin{proof}
        By Lemma~\ref{lem-grh-input}, we know that $|E(x) f(x)| \le O(\sqrt{x} (\log x)^2 / (x \log x)) = O(\log x/\sqrt{x})$ and similarly for $|E(y) f(y)| \le O(\log y/\sqrt{y})$.
        Now, we bound the integral term:
        \begin{align*}
            \left|\int_x^y E(t) f'(t) dt\right| &\le O \left(\int_x^y \sqrt{t} (\log t)^2 \cdot \frac{\log t + 1}{(t \log t)^2} dt\right) &\text{by Lemma~\ref{lem-grh-input} and $t \ge 1$}\\
            &= O \left(\int_x^y \frac{\log t }{ t^{3/2}} dt\right)\\
            &= O\left( \frac{-2 \log t - 4}{\sqrt{t}} \Bigg|^y_x\right)\\
            &\le O\left(\frac{\log x}{\sqrt{x}}\right).
        \end{align*}
        The claim then follows by the triangle inequality.
    \end{proof}
\end{proof}

\end{document}